\pdfoutput=1
\documentclass[UKenglish, autoref, thm-restate, numberwithinsect]{lipics-v2021}
\title{Dependence and Independence for Reversible Process Calculi}
\author{Clément Aubert}{Augusta University, GA, USA}{caubert@augusta.edu}{0000-0001-6346-3043}{}
\author{Iain Phillips}{Imperial College London, UK}{I.Phillips@imperial.ac.uk}{0000-0001-5013-5876}{}
\author{Irek Ulidowski}{University of Leicester, UK \and AGH University of Science and Technology, Poland}{I.Ulidowski@leicester.ac.uk}{0000-0002-3834-2036}{}
\authorrunning{Aubert, Phillips, Ulidowski}

\ccsdesc[500]{Theory of computation~Process calculi}
\ccsdesc[500]{Computing methodologies~Concurrent programming languages}
\ccsdesc[500]{Theory of computation~Logic and verification}

\keywords{%
Concurrency,
Process algebras,
Reversibility,
Labelled Transition Systems,
Independence,
Events,
Causality,
Conflict,
Bisimulation
}

\input{latex/packages} 				%
\usepackage{ushort}

\newcommand{\Maxop}{\mathrm{max}}
\newcommand{\Max}[1]{\mathopen{\Maxop}\left(#1\right)}

\newcommand{\procset}{\ensuremath{\mathbb{P}}}  %
\newcommand{\kprocset}{\ensuremath{\mathbb{X}}} %

\newcommand{\events}{\ensuremath{\mathbb{E}}} 

\newcommand{\Proc}{\mathsf{Proc}} %
\newcommand{\Lab}{\mathsf{Lab}}
\newcommand{\union}{\cup}
\newcommand{\inter}{\cap}

\newcommand{\eveqt}{\sim}
\newcommand{\sqeqt}{\eveqt}

\newcommand{\Par}{\mid}

\newcommand{\evkeqt}{\eveqt_{\mathsf{k}}}

\newcommand{\prov}[1]{{#1}^\dagger}
\newcommand{\base}[1]{{#1}^\circ}
\newcommand{\kprov}[1]{{#1}_{\keysop}^\dagger}
\newcommand{\kbase}[1]{{#1}_{\keysop}^\circ}
\newcommand{\conn}{\mathrel{\curlyvee}} %

\newcommand{\rev}[1]{\ushortw{#1}} %
\newcommand{\Rev}[1]{\ensuremath{\rev{#1}}} %
\newcommand{\tRev}[1]{\Rev{\text{#1}}}

\newcommand{\ceqt}{\approx}

\newcommand{\cte}{\sharp}
\newcommand{\ind}{\mathrel{\iota}}

\newcommand{\es}{\varepsilon}

\newcommand{\coind}{\mathrel{\mathsf{ci}}} %
\newcommand{\Cf}{\mathrel{\#}} %

\newcommand{\nil}{\mathbf{0}}
\newcommand{\und}[1]{\mathsf{und}(#1)}

\newcommand{\evtkof}[2]{\evop_{#1}(#2)}
\newcommand{\rtran}[1]{\stackrel{#1}{\rightsquigarrow}} %
\newcommand{\len}[1]{|#1|}

\newcommand{\srcof}[1]{\mathsf{src}(#1)}
\newcommand{\tgtof}[1]{\mathsf{tgt}(#1)}
\newcommand{\fwdof}[1]{\mathsf{fwd}(#1)}
\newcommand{\lblof}[1]{\mathsf{lbl}(#1)}
\newcommand{\ip}{\prec}

\newcommand{\maxkeys}[1]{\mathsf{maxkeys}(#1)}
\newcommand{\bts}[1]{\mathsf{bts}(#1)} %
\newcommand{\sizeof}[1]{\mathsf{s}(#1)} %

\newcommand*{\eg}{e.g.\@,\xspace}
\newcommand*{\Eg}{E.g.\@,\xspace}
\newcommand*{\cf}{cf.\@\xspace}
\newcommand*{\ie}{i.e.\@,\xspace}

\renewcommand{\st}{s.t.\@\xspace}
\newcommand*{\resp}{resp.\@\xspace}
\newcommand*{\wrt}{w.r.t.\@\xspace}

\makeatletter
\newcommand*{\etc}{%
	\@ifnextchar{.}%
	{etc}
	{etc.\@\xspace}%
}
\makeatother

\newcommand{\itemreflipics}[1]{%
\textcolor{lipicsGray}{\sffamily\bfseries\upshape\mathversion{bold}\ref{#1}.\@\xspace}%
}

\makeatletter
\DeclareFontEncoding{LS2}{}{\@noaccents}
\makeatother
\DeclareFontSubstitution{LS2}{stix}{m}{n}
\DeclareSymbolFont{largesymbolsstix}{LS2}{stixex}{m}{n}
\DeclareMathDelimiter{\lBrace}{\mathopen} {largesymbolsstix}{"E8}{largesymbolsstix}{"0E}
\DeclareMathDelimiter{\rBrace}{\mathclose}{largesymbolsstix}{"E9}{largesymbolsstix}{"0F}

\DeclareFontEncoding{LS1}{}{}
\DeclareFontSubstitution{LS1}{stix}{m}{n}
\DeclareSymbolFont{stixsymbols}{LS1}{stixscr}{m}{n}
\SetSymbolFont{stixsymbols}{bold}{LS1}{stixscr}{b}{n}

\makeatletter
\NewDocumentCommand{\smallbar}{}{%
	\mathrel{\mathpalette\smallbar@\relax}%
}

\newcommand{\current@math@font}[1]{%
	\ifx#1\displaystyle\textfont\else
	\ifx#1\textstyle\textfont\else
	\ifx#1\scriptstyle\scriptfont\else
	\scriptscriptfont\fi\fi\fi
}
\newcommand{\smallbar@factor}[1]{%
	\ifx#1\displaystyle 1.135\else
	\ifx#1\textstyle 1.128\else
	\ifx#1\scriptstyle 1.09\else
	1.06\fi\fi\fi
}

\newcommand{\smallbar@}[2]{%
	\begingroup
	\sbox\z@{$\m@th#1\mapstochar$}%
	\dimen0=\smallbar@factor{#1}\ht\z@
	\dimen2=\dimeval{2\fontdimen22\current@math@font{#1} 2 - \dimen0}%
	\mbox{%
		$\m@th#1\mkern1mu
		\begin{picture}(0,\dimen0)
			\roundcap
			\linethickness{\fontdimen8\current@math@font{#1}3}
			\Line(0,\dimen2)(0,\dimen0)
		\end{picture}%
		\mkern1mu$%
	}%
	\endgroup
}
\makeatother

\DeclareMathOperator{\ordop}{ord}
\DeclareMathOperator{\rem}{rm}

\newcommand{\card}[1]{\ensuremath{|#1|}}

\newcommand{\orig}[1]{O_{#1}} %
\newcommand{\real}[1]{r(#1)} %
\newcommand{\bs}{\backslash}

\DeclareSymbolFont{symbolsC}{U}{txsyc}{m}{n}
\SetSymbolFont{symbolsC}{bold}{U}{txsyc}{bx}{n}
\DeclareFontSubstitution{U}{txsyc}{m}{n}
\DeclareMathSymbol{\opentimes}{\mathrel}{symbolsC}{93}
\newcommand{\sdep}{\opentimes}
\newcommand{\tsdep}{\mathrel{\sdep^*}} %

\newcommand{\Left}{\mathrm{L}}        %
\newcommand{\Right}{\mathrm{R}}       %
\newcommand{\Dir}{\mathrm{d}}         %
\newcommand{\Edir}{\mathrm{p}}        %
\newcommand{\OpDir}[1]{\overline{#1}} %

\renewcommand{\L}{\Left}              
\newcommand{\R}{\Right}
\newcommand{\D}{\Dir}
\newcommand{\OD}{\OpDir{\D}}
\newcommand{\E}{\Edir}

\newcommand{\lmidr}{{\mid_{\R}}}
\newcommand{\lmidl}{{\mid_{\L}}}
\newcommand{\lmidd}{{\mid_{\D}}}
\newcommand{\lmidod}{{\mid_{\OD}}}
\newcommand{\lmide}{{\mid_{\E}}}

\newcommand{\lplusr}{{+_{\R}}}
\newcommand{\lplusl}{{+_{\L}}}
\newcommand{\lplusd}{{+_{\D}}}
\newcommand{\lplusod}{{+_{\OD}}}
\newcommand{\lpluse}{{+_{\E}}}

\ushortCreate()[.05em](1.2){ushortwt} %

\newcommand{\rlmidr}{{\ushortwt{\mid_{\R}}}}
\newcommand{\rlmidl}{{\ushortwt{\mid_{\L}}}}

\newcommand{\rlplusr}{{\Rev{+_{\R}}}}
\newcommand{\rlplusl}{{\Rev{+_{\L}}}}

\newcommand{\bpair}[2]{\langle #1 , #2 \rangle} %
\newcommand{\cpair}[2]{\bpair{\lmidl #1}{\lmidr #2}} %

\newcommand{\names}{\ensuremath{\mathsf{N}}}

\newcommand{\labelset}{\ensuremath{\mathsf{L}}}
\newcommand{\keyset}{\ensuremath{\mathsf{K}}}
\newcommand{\proofset}{\ensuremath{\mathsf{P}}}
\newcommand{\kplabelset}{\labelset_{\keyset}^{\proofset}}
\newcommand{\klabelset}{\labelset_{\keyset}}
\newcommand{\plabelset}{\labelset^{\proofset}}

\newcommand{\out}[1]{\overline{#1}}

\newcommand{\setst}{ \mid }

\newcommand{\ccs}{\ensuremath{\mathsf{CCS}}\xspace}
\newcommand{\ccsk}{\ensuremath{\mathsf{CCS}{\keyset}}\xspace}

\newcommand{\pccs}{\ensuremath{\mathsf{CCS}^{\proofset}}\xspace}
\newcommand{\pccsk}{\ensuremath{\mathsf{CCS}{\keyset}^{\proofset}}\xspace}

\newcommand{\ire}{IRE\xspace}

\DeclareMathOperator{\stdop}{\mathsf{sd}}
\DeclareMathOperator{\tostdop}{\mathsf{toStd}}
\DeclareMathOperator{\keysop}{\mathsf{keys}}
\DeclareMathOperator{\keyop}{\mathsf{key}}
\DeclareMathOperator{\evop}{\mathsf{ev}}

\DeclareMathOperator{\kfop}{\mathsf{kf}}

\DeclareMathOperator{\keop}{\mathsf{ke}}

\DeclareMathOperator{\lablop}{\ell} %
\DeclareMathSymbol{\kayop}{\mathalpha}{stixsymbols}{"6B} %
\DeclareMathSymbol{\ekayop}{\mathalpha}{stixsymbols}{"62} %

\DeclareMathOperator{\prunop}{\phi}
\DeclareMathOperator{\krop}{\psi}

\newcommand{\std}[1]{\stdop(#1)}

\newcommand{\prun}[1]{\prunop(#1)}
\newcommand{\kr}[1]{\krop(#1)}
\newcommand{\kf}[1]{\kfop(#1)}

\newcommand{\ke}[1]{\keop(#1)}

\newcommand{\kay}[1]{\kayop(#1)}
\newcommand{\ekay}[1]{\labl{#1}[\kay{#1}]}

\newcommand{\key}[1]{\keyop(#1)}
\newcommand{\keys}[1]{\keysop(#1)}
\newcommand{\ev}[1]{\evop(#1)}
\newcommand{\ord}[1]{\ordop(#1)}
\newcommand{\labl}[1]{\lablop(#1)}

\ExplSyntaxOn
\NewDocumentCommand{\Rel}{m O{} O{}} 
{
	\str_case:nnF { #1 }
	{
		{s}{ %
			\tl_if_blank:nTF {#2} {       							%
				\tl_if_blank:nTF {#3} {	  							%
					\mathrel{<} 	    				
				}{
					\mathrel{<^{#3}}					%
			}}{
				\tl_if_blank:nTF {#3} {	  							%
					\mathrel{<\c_math_subscript_token{\scriptscriptstyle{\mathsf{#2}}}} 		%
				}{
					\mathrel{<\c_math_subscript_token{\scriptscriptstyle{\mathsf{#2}}}^{#3}} 	%
			}}
		}
		{i}{ %
			\tl_if_blank:nTF {#2} {       							%
				\tl_if_blank:nTF {#3} {	  							%
					\mathrel{>} 	    				
				}{
					\mathrel{>^{#3}}					%
			}}{
				\tl_if_blank:nTF {#3} {	  							%
					\mathrel{>\c_math_subscript_token{\scriptscriptstyle{\mathsf{#2}}}} 		%
				}{
					\mathrel{>\c_math_subscript_token{\scriptscriptstyle{\mathsf{#2}}}^{#3}} 	%
			}}
		}	
		{b}{ %
			\tl_if_blank:nTF {#2} {       							%
				\tl_if_blank:nTF {#3} {	  							%
					\mathrel{\eveqt} 	    				
				}{
					\mathrel{\eveqt^{#3}}					%
			}}{
				\tl_if_blank:nTF {#3} {	  							%
					\mathrel{\eveqt\c_math_subscript_token{\scriptscriptstyle{\mathsf{#2}}}} 		%
				}{
					\mathrel{\eveqt\c_math_subscript_token{\scriptscriptstyle{\mathsf{#2}}}^{#3}} 	%
			}}
		}
		{r}{ %
			\tl_if_blank:nTF {#2} {       							%
				\tl_if_blank:nTF {#3} {	  							%
					\mathrel{\mathcal{R}} 	    				
				}{
					\mathrel{\mathcal{R}^{#3}}					%
			}}{
				\tl_if_blank:nTF {#3} {	  							%
						\mathcal{R}\c_math_subscript_token{\scriptscriptstyle{\mathsf{#2}}}
				}{
						\mathcal{R}\c_math_subscript_token{\scriptscriptstyle{\mathsf{#2}}}^{#3}
			}}
		}
	}{error}
}
\NewDocumentCommand{\nRel}{m O{} O{}} 
{
	\str_case:nnF { #1 }
	{
		{s}{ %
			\tl_if_blank:nTF {#2} {       							%
				\tl_if_blank:nTF {#3} {	  							%
					\mathrel{\not{<}} 	    				
				}{
					\mathrel{\not{<^{#3}}}					%
			}}{
				\tl_if_blank:nTF {#3} {	  							%
					\mathrel{\not{<\c_math_subscript_token{\scriptscriptstyle{\mathsf{#2}}}}} 		%
				}{
					\mathrel{\not{<\c_math_subscript_token{\scriptscriptstyle{\mathsf{#2}}}^{#3}}} 	%
			}}
		}
		{i}{ %
			\tl_if_blank:nTF {#2} {       							%
				\tl_if_blank:nTF {#3} {	  							%
					\mathrel{\not{>}} 	    				
				}{
					\mathrel{\not{>^{#3}}}					%
			}}{
				\tl_if_blank:nTF {#3} {	  							%
					\mathrel{\not{>\c_math_subscript_token{\scriptscriptstyle{\mathsf{#2}}}}} 		%
				}{
					\mathrel{\not{>\c_math_subscript_token{\scriptscriptstyle{\mathsf{#2}}}^{#3}}} 	%
			}}
		}	
		{b}{ %
			\tl_if_blank:nTF {#2} {       							%
				\tl_if_blank:nTF {#3} {	  							%
					\mathrel{\nsim} 	    				
				}{
					\mathrel{\nsim^{#3}}					%
			}}{
				\tl_if_blank:nTF {#3} {	  							%
					\mathrel{\nsim\c_math_subscript_token{\scriptscriptstyle{\mathsf{#2}}}} 		%
				}{
					\mathrel{\nsim\c_math_subscript_token{\scriptscriptstyle{\mathsf{#2}}}^{#3}} 	%
			}}
		}
	}{error}
}
\ExplSyntaxOff

\definecolor{irek}{HTML}{1b9e77}
\definecolor{iain}{HTML}{d95f02}
\definecolor{clem}{HTML}{7570b3}
\newcommand{\irek}[1]{{\color{irek} Irek: #1}}
\newcommand{\iain}[1]{{\color{iain} Iain: #1}}
\newcommand{\clem}[1]{{\color{clem} Clem: #1}}

\newcommand{\Comment}[1]{}

\definecolor{ind} {RGB}{43,131,186}
\colorlet  {conc}{ind!30}
\definecolor{sdep}{RGB}{215,25,28}
\colorlet  {tsdep}{sdep!30}
\definecolor{conf}{RGB}{171,221,164}
\definecolor{ord} {RGB}{253,174,97}

\tikzset{
sys/.style = {
	line width=.5mm,
	inner sep=0pt,
	outer sep=0pt,
	rounded corners=0.5cm,
	draw = tsdep,
	minimum size=1.5cm,
	text width=1.6cm,
	text centered
},
spl/.style = {
	rectangle split,
	rectangle split horizontal,
	rectangle split parts=2,
	rectangle split part fill={sdep!20,ind!20},
},
sya/.style= {
	thick,
	font = {\LARGE},	
},
rel/.style = {
	line width=1mm, %
},
adj/.style = {
	line width=1mm, %
	dotted,
},
cpt/.style = {
},
u/.style = {yshift=#1\pgflinewidth},
s/.style = {xshift=#1\pgflinewidth},
}

\usetikzlibrary{calc,decorations.pathreplacing}

\newcommand{\tikzmark}[1]{%
	\tikz[overlay,remember picture,baseline] \node[anchor=base] (#1) {};}

\makeatletter
\appto{\endmulticols}{\@doendpe}
\makeatother
\usepackage{latex/macros_arrows}	%
\nolinenumbers
\def\hyph{-\penalty0\hskip0pt\relax} %

\hideLIPIcs %

\begin{document}
	\maketitle
	\begin{abstract}
To refine formal methods for concurrent systems, there are several ways of enriching classical operational semantics of process calculi.  One can enable the auditing and undoing of past synchronisations thanks to communication keys, thus easing the study of true concurrency as a by-product.  Alternatively, proof labels embed information about the origins of actions in transition labels, facilitating syntactic analysis.  Enriching proof labels with keys enables a theory of the relations on transitions and on events based on their labels only.  We offer for the first time separate definitions of dependence relation and independence relation, and prove their complementarity on connected transitions instead of postulating it.  Leveraging the recent axiomatic approach to reversibility, we prove the canonicity of these relations and provide additional tools to study the relationships between e.g., concurrency and causality on transitions and events. Finally, we make precise the subtle relationship between bisimulations based on both forward and backward transitions, on key ordering, and on dependency preservation, providing a direct definition of History Preserving  bisimulation for CCS.
	\end{abstract}

	\begingroup
	\let\clearpage\relax
	\section{Introduction}
\label{sec:intro}

The Calculus of Communicating Systems (\ccs~\cite{milner80lncs}) is the main starting point
 to develop formal methods for concurrent systems, such as the applied \(\pi\)-calculus~\cite{Abadi2018,Aubert2022e}, that can be used to \eg verify communication protocols~\cite{Blanchet2016,Ryan2011,Horne2021}.
At the heart of this formalism lie three core concepts: synchronisation, concurrency of events (or independence of actions), and bisimulation.
In a nutshell, this paper offers an original and definitive answer to the question of defining dependence and independence for a reversible extension of \ccs that uses keys to handle synchronisation, and defines a key-based bisimulation relation that is proven to coincide with another new  bisimulation defined in terms of (in)dependence.

The starting point is the formal treatment of parallel communications, as improved by \emph{communication keys}: \ccs with keys (\ccsk~\cite{PU07}) enables the representation of reversible systems~\cite{Aman2020}, provides the capacity of auditing (and of undoing) previous synchronisations, and proved useful to study true concurrency~\cite{PU07a}.
Orthogonally, independence allows to disentangle complex situations and to understand the roots of certain decisions or bugs thanks to the dual notion of dependence.
However, those latter notions are not always easy to work with because they are sometimes \emph{defined} by complementarity~\cite{aubert2023c,DeganoGP03}, or defined only on coinitial or composable transitions. %
While it has been suspected that the order on keys~\cite{LaneseP21} or the dependence on previous events~\cite{bisim-applied-pi} could be fruitful in designing bisimulations, no definitive answer had been provided thus far---in part because dependence and independence were difficult to manipulate, and because not leveraging the axiomatic approach to reversibility made reasoning cumbersome and lengthy.

The paper is organised as follows: we first recall how \ccs is extended with the so-called proof labels and communication keys to give \pccsk.
A fresh look is then offered on the notions of dependence and independence for \pccsk by \emph{defining them} separately and then proving their complementarity on connected transitions (\eg transitions that can be reached from the same process).
We then present the %
 axiomatic approach to reversibility~\cite{LanesePU20,LPU24} and leverage it to establish relevant properties about our semantics for \pccsk.
After briefly discussing why order on keys and dependence are \enquote{the right notions} to define History Preserving-like bisimulations, we introduce two new bisimulations that preserve key ordering and dependencies, and prove that they coincide for standard \pccsk processes.

	\section{Background and Related Work}
\label{ssec:related}

The first use of proof labels to represent the concurrency relation between CCS transitions is due to Boudol and Castellani~\cite{BC88,BC94}.
Their definition of concurrent transitions applied only to coinitial transitions, but was proven~\cite{DeganoGP03} to coincide with causal semantics for \(\pi\)-calculus~\cite{Degano1999} and for causal trees~\cite{Darondeau1990a}.
This universality and simplicity made proof labels convenient, but other approaches were designed to capture (in)dependence on composable transitions.
State- and pomset-based semantics of \ccs leverage original definitions of independence~\cite[Definition 3.4]{Degano1987} and dependence~\cite[p.~952]{Castellani2001}, respectively,  %
on composable transitions. %
Static~\cite{Aceto1994} and dynamic localities~\cite{Boudol1994} were developed to capture local causality and concurrency on composable transitions, but required to use occurrence transition systems~\cite[Sect.\ 2.6.1]{Castellani2001}.
The first correctness criterion relating independence on coinitial and composable transitions was formulated in the reversible setting~\cite{aubert2023c} and no criterion demanding to prove the complementarity of dependence and independence relations defined separately was, to our knowledge, formulated before.
While a set of belief and mappings supported the canonicity of those notions, to our knowledge they were never proven unique under mild assumptions.

While locality-based equivalences are defined on syntactical models~\cite{Castellani1995,Aceto1994,Boudol1994}, causal equivalences were formalised mostly for semantic models: History Preserving (HP) bisimulation---slightly coarser than Hereditary HP bisimulation~\cite{Bed91,JNW96}---originated on behaviour structures~\cite{Rabinovich1988}, and has been reformulated on Petri nets~\cite{Vogler1991}, causal trees~\cite{Darondeau1990a}, process graphs~\cite{vG96}, configuration structures~\cite{vGG01,PhillipsU12}, modal logics~\cite{BaldanC14,PhillipsU14}, and even
automata~\cite{FAHRENBERG2013165}. 
In the setting of reversible process calculi, HP has been studied in different forms~\cite{PU07,PU07a,Aubert2020b,LaneseP21}, some of which syntactical.
A purely syntactical \enquote{local} HP bisimulation has been defined on \ccs processes~\cite[Sect.\ A.2]{Castellani1995}, but, to our knowledge, the usual HP---\eg that considers the global causal ordering---was never defined on \ccs.
However, it was known that a strong version of causal bisimulation~\cite[Definition 4.1]{Kiehn94} coincided with HP bisimulation~\cite[Theorem 2.1]{Darondeau1990a}, but it required to use conflict labelled event structures.

In conclusion, process algebra fell short on defining dependence and independence on all combinations of transitions (coinitial and composable), while causal equivalences were expressed on models that required to map back-and-forth notions of cause and concurrency.

\Comment{
 \irek{Clément, since you are expert on this, could you please list their works, and that of others (Degano, De Nicola, Montanari, van Glabbeek\& Vaandrager, your own work and any recent work...) relating to concurrency/dependence}.\clem{Ok, will do, but I must admit I am not familiar with Vaandrager's work -- this can be fixed, though.}\irek{Regarding Vaandrager, I was thinking of 
 Rob J. van Glabbeek, Frits W. Vaandrager: Petri Net Models for Algebraic Theories of Concurrency. PARLE (2) 1987: 224-242
 Boudol and Castellani write in their A NON-INTERLEAVING SEMANTICS FOR CCS
BASED ON PROVED TRANSITIONS about this paper:
 
 ``Similarly, van Glabbeek and Vaandrager define in [13]
a partial order semantics for ( one-safe) Petri nets, using the structure
of the net to associate a causality ordering with each firing sequence.''  Perhaps is not that closely related...}
}

	\section{CCSK with Proof Labels} %
\label{sec:proved-lts}

This section recalls the extension of \ccsk with proof labels (\pccsk) \cite{Aub22,aubert2023c}, an LTS in bijection with \ccsk as discussed in \autoref{app:sub-systems}.

\begin{definition}[(Co-)names,  labels and keys]
	\label{def:co-names}
	Let \(\names\) be a set of \emph{names}, ranged over by $a$, $b$ and $c$. %
	A bijection \(\out{\cdot}:\names \to \out{\names}\), whose inverse is also written \(\out{\cdot}\), 
	gives the \emph{complement} of a name. %
	The set of \emph{labels} \(\labelset\) is \(\names \cup \out{\names} \cup\{\tau\}\), and we use \(\alpha\), \(\beta\) (\resp \(\lambda\)) to range over \(\labelset\) (\resp \(\labelset \bs \{\tau\}\)).
	It is convenient to extend the complement mapping to labels by letting \(\out{\tau} = \tau\).
	
	Let \(\keyset\) be a denumerable set of \emph{keys}, ranged over by \(k\), \(m\) and \(n\).
	\emph{Keyed labels}, denoted \(a[k]\), \(a[m]\), \(a[n]\), \(b[k]\), etc., are elements of \(\labelset \times \keyset = \klabelset\).
\end{definition}

\begin{definition}[Operators]
	\label{def:operators}
	The set \(\kprocset\) of \ccsk processes is defined as usual:%

	\noindent\begin{minipage}{.46\linewidth}
		\begin{align}
			X, Y \coloneqq ~ %
			& \nil \tag{Inactive process} \\
			\|~ &  \alpha. X \tag{Prefix} \\
			\|~ &   X \bs \lambda \tag{Restriction}
		\end{align}
	\end{minipage}
	\begin{minipage}{.5\linewidth}
		\begin{align}
			\|~ & X + Y  \tag{Sum} \\ 
			\|~ & X \mid Y \tag{Parallel composition}\\
			\|~ & \alpha[k].X \tag{Keyed prefix}
		\end{align}
	\end{minipage}
	
	As usual, the inactive process \(\nil\) is omitted when preceded by a (keyed) prefix, and the binding power of the operators~\cite[p. 68]{milner80lncs}, from highest to lowest, is 
	\(\bs \lambda\), \(\alpha[k].\) \(\alpha.\), \(\mid\) and  \(+\).
	
	We write \(\keys{X}\) for the set of keys in \(X\). The set of \ccs processes is denoted \(\procset = \{X \setst \keys{X} = \emptyset\}\), %
	and we let \(P\), \(Q\) range over it.
\end{definition}

\begin{definition}[Proof labels]
	\label{def:proof-keyed-labels}
	We let \(\D\) range over the \emph{directions} \(\L\)(eft) and \(\R\)(ight), \(\upsilon\), \(\upsilon_{1}\) and \(\upsilon_{2}\) range over strings in \(\{\lmidd, \lplusd\}^*\), and \(\theta\) range over \emph{proof keyed labels}: %
	\begin{align*}
		\theta & \coloneqq \upsilon \alpha[k]  ~\|~ \upsilon \cpair{\upsilon_{1} \lambda[k]}{\upsilon_{2} \out{\lambda}[k]}%
	\end{align*}

	We write \(\kplabelset\) for the set of proof keyed labels, %
	define \(\ell : \kplabelset \to \labelset\) and \(\kayop: \kplabelset \to \keyset\) %
	as %
	\begin{align*}
	\ell(\upsilon \alpha[k]) & = \alpha &   &   & \ell(\upsilon \bpair{\upsilon_{1}\lambda[k]}{\upsilon_{2} \out{\lambda}[k]}) & = \tau &&& 
	\kay{\upsilon \alpha[k]} & = k      &   &   & \kay{\upsilon \bpair{\upsilon_{1}\lambda[k]}{\upsilon_{2} \out{\lambda}[k]}} & = k %
\end{align*}
\end{definition}

\begin{notation}
	\label{notation:key}
	We sometimes write \(\theta\) with \(\kay{\theta} = k\) as \(\theta[k]\). %
	We let \(\OpDir{\D} = \R\) if \(\D = \L\), else \(\OpDir{\D} = \L\).
	We generally omit \enquote{keyed} and simply write \enquote{proof label}.%
\end{notation}

\begin{definition}[LTS for \ccsk with proof labels~\protect{\cite{Aub22,aubert2023c}}]
	The \emph{labelled transition system (LTS) for \ccsk with proof labels}, denoted by \pccsk, is 
	$(\kprocset, \kplabelset, \pr{fb}[\theta])$ where $\pr{fb}[\theta]$ is the union of transition relations generated by 
	the forward and backward rules given in \autoref{fig:provedltsrulesccskfw}.
	As usual, we let \(\pr{fb}^*\) be the reflexive transitive closure of \(\pr{fb}\).%
\end{definition}

\begin{figure}[ht]
	\begin{tcolorbox}[title = {Action, Prefix and Restriction}, sidebyside]
		\begin{tcolorbox}[adjusted title=Forward]
			\begin{prooftree}
				\hypo{}
				\infer[left label={\(\keys{X} = \emptyset%
					\)}]1[act]{\alpha. X \pr{f}[\alpha][k]  \alpha[k].X}
			\end{prooftree}
			\\[.9em]
			\begin{prooftree}
				\hypo{X \pr{f}[\theta] X'}
				\infer[left label={\(\kay{\theta} \neq k\)}]1[pre]{\alpha[k]. X \pr{f}[\theta] \alpha[k].X'}
			\end{prooftree}
			\\[.9em]
			\begin{prooftree}
				\hypo{ X \pr{f}[\theta] X '}
				\infer[left label={\(\labl{\theta} \notin \{a, \out{a}\}\)}]1[res]{X  \bs a  \pr{f}[\theta] X ' \bs a}
			\end{prooftree}
		\end{tcolorbox}
		\tcblower
		\begin{tcolorbox}[adjusted title=Backward]
			\begin{prooftree}
				\hypo{}
				\infer[left label={\(\keys{X} = \emptyset
					\)}]1[\tRev{act}]{ \alpha[k].X \pr{b}[\alpha][k] \alpha. X}
			\end{prooftree}
			\\[.9em]
			\begin{prooftree}
				\hypo{X' \pr{b}[\theta] X}
				\infer[left label={\(\kay{\theta} \neq k\)}]1[\tRev{pre}]{\alpha[k].X'\pr{b}[\theta] \alpha[k]. X }
			\end{prooftree}
			\\[.9em]
			\begin{prooftree}
				\hypo{ X' \pr{b}[\theta] X}
				\infer[left label={\(\labl{\theta} \notin \{a, \out{a}\}\)}]1[\tRev{res}]{X ' \bs a\pr{b}[\theta] X  \bs a  }
			\end{prooftree}
		\end{tcolorbox}
	\end{tcolorbox}
	\begin{tcolorbox}[title = Parallel, sidebyside]
		\begin{tcolorbox}[adjusted title=Forward]
			\begin{prooftree}
				\hypo{X \pr{f}[\theta] X'}
				\infer[left label={\(\kay{\theta} \notin \keys{Y}\)}]
				1[\(\lmidl\)]{X \Par Y \pr{f}[\lmidl\theta] X' \Par  Y}
			\end{prooftree}
			\\[.9em]
			\begin{prooftree}
				\hypo{X \pr{f}[\upsilon_{\L} \lambda][k]  X'}
				\hypo{Y \pr{f}[\upsilon_{\R} \out{\lambda}][k]  Y'}
				\infer2[syn]{X \Par  Y \pr{f}[\cpair{\upsilon_{\L} \lambda \protect{[k]}}{\upsilon_{\R} \out{\lambda} \protect{[k]}}] X' \Par  Y'}
			\end{prooftree}
		\end{tcolorbox}
		\tcblower
		\begin{tcolorbox}[adjusted title=Backward]
			\begin{prooftree}
				\hypo{X' \pr{b}[\theta] X}
				\infer[left label={\(\kay{\theta} \notin \keys{Y}\)}]
				1[\rlmidl]{X' \Par Y \pr{b}[\lmidl\theta] X \Par  Y}
			\end{prooftree}
			\\[.9em]
			\begin{prooftree}
				\hypo{X' \pr{b}[\upsilon_{\L} \lambda][k]  X}
				\hypo{Y' \pr{b}[\upsilon_{\R} \out{\lambda}][k]  Y}
				\infer%
				2[\tRev{syn}]{X' \Par  Y' \pr{b}[\cpair{\upsilon_{\L} \lambda \protect{[k]}}{\upsilon_{\R} \out{\lambda} \protect{[k]}}] X \Par  Y}
			\end{prooftree}
		\end{tcolorbox}
	\end{tcolorbox}
	
	\begin{tcolorbox}[adjusted title=Sum, sidebyside]
		\begin{tcolorbox}[adjusted title=Forward]
			\begin{prooftree}
				\hypo{X \pr{f}[\theta] X'}
				\infer[left label={\(\keys{Y} = \emptyset\)}]1[\( \lplusl \)]{X + Y \pr{f}[\lplusl \theta] X' + Y}
			\end{prooftree}
		\end{tcolorbox}
		\tcblower
		\begin{tcolorbox}[adjusted title=Backward]
			\begin{prooftree}
				\hypo{X' \pr{b}[\theta] X}
				\infer[left label={\(\keys{Y} = \emptyset\)}]1[\(\rlplusl\)]{X' + Y \pr{b}[\lplusl \theta] X + Y}
			\end{prooftree}
		\end{tcolorbox}
	\end{tcolorbox}
	
	\caption{Forward and backward transition rules for \pccsk (\(\lmidr\),\(\rlmidr\), \(\lplusr\) and \(\rlplusr\) omitted).} %
	
	\label{fig:provedltsrulesccskfw}
\end{figure}

We note that the LTS for \ccsk%
can be obtained from \pccsk by replacing $\kplabelset$ with \(\klabelset\) and $\theta$
with $\ell(\theta)[\kay{\theta}]$ in \autoref{fig:provedltsrulesccskfw}, which corresponds to erasing the \enquote{proved} part of labels.
In the following, we write \(\base{\cdot} = {(\prov{\cdot}})^{-1}\) the bijection between transitions in \ccsk and \pccsk~\cite[Lemma 1]{Aub22}, given in detail in \autoref{app:sub-systems}.

\begin{definition}[Transitions]
	\label{def:transitions-paths}
	A \emph{transition} \(t: X \pr{bf}[\theta][k] Y\) has for \emph{source} \(\srcof{t} = X\), for \emph{target} \(\tgtof{t} = Y\), for \emph{proof label} \(\lblof{t} = \theta\), and for \emph{key} \(\key{t} = k\).
	Transitions \(t_1\), \(t_2\) are \emph{coinitial} if \(\srcof{t_1} = \srcof{t_2}\), \emph{cofinal} if \(\tgtof{t_1} = \tgtof{t_2}\), \emph{composable} if \(\tgtof{t_1} = \srcof{t_2}\), and \emph{adjacent} if they are either coinitial, cofinal or composable (in either order).

	A \emph{path} is a sequence of transitions \(r = t_1 t_2 \cdots t_n\)
	such that \(t_i\) and \(t_{i+1}\), for \(1 \leqslant i < n\), are composable.
	Its source \(\srcof{r}\) is \(\srcof{t_1}\)%
	, its target \(\tgtof{r}\) is \(\tgtof{t_n}\)%
	, its length \(\len{r}\) is \(n\), and 
	its set of keys \(\keys{r}\) is \(\key{t_1} \cup \cdots \cup \key{t_n}\).
	We let $r$ and $s$ range over paths, and $t$ and $u$ range over transitions.
	A path $r$ is \emph{rooted} if $\srcof{r}$ cannot perform a backward transition.

	A transition \(t_1\) \emph{is connected to} a transition \(t_2\) if there exists a path $r$ \st \(\srcof{r} = \srcof{t_1}\) and \(\tgtof{r} = \tgtof{t_2}\). 
	Two transitions are \emph{connected} if one is connected to the other.
\end{definition}

\begin{definition}[Standard and reachable processes]
	\label{def:std}
	We say that \(X\) is \emph{standard} and write \(\stdop(X)\) iff \(\keys{X} = \emptyset\)---equivalently, if \(X\) is a \ccs process.
	If there exists an \emph{origin process} \(\orig{X}\) \st \(\stdop(\orig{X})\) and a rooted path \(r_{X}: \orig{X} \pr{fb}^* X\) %
	then \(X\) is \emph{reachable}.
\end{definition}

Note that \(\orig{X}\) is easily obtained by erasing the keys in \(X\)~\cite{LaneseP21}.
We only consider reachable processes in the rest of the paper.
As \(\pr{f}\) and \(\pr{b}\) are symmetric, we easily obtain:%

\begin{lemma}[%
	Loop Lemma]
	\label{lem:loop_proved}
	For all \(t: X \pr{f}[\theta] Y\), there exists \(\Rev{t}: Y \pr{b}[\theta] X\), and conversely.
	Furthermore, %
	\(\Rev{\Rev{t}} = t\).
\end{lemma}

\Comment{
Note that thanks to the Loop Lemma, all the transitions in \autoref{fig:example-transitions} could be reversed (from \(\pr{f}\) to \(\pr{b}\) and reciprocally) and the diagram would still be correct.
}

\Comment{
As usual, we restrict ourselves to reachable processes: %

\begin{example}
	\label{ex:transitions}
	Some of the processes reachable from \(a[k] \Par (\out{a} + b)\) are presented in \autoref{fig:example-transitions}: note that since we supposed infinitely many keys, \pccsk is infinitely branching as hinted toward by the two transitions labelled \(\lmidr \lplusr b[n]\) and \(\lmidr \lplusr b[m]\).	
	Note that the origin of this process is \(a \Par (\out a + b)\), and that \eg the transitions labelled \(\lmidr \lplusr b[n]\) and \(\lmidr \lplusl \out{a}[m]\) are coinitial.
	Of course, all the transitions pictured are connected and all the processes are reachable.
\end{example}
}

\begin{figure}%
	\begin{center}
	\tikzset{
	lbl/.style = {
		scale=.8,
	}
}

\begin{tikzpicture}[
	x={(1.5, 0)},  %
	y={(0, 1.1)}, %
	baseline,
	anchor=base
	]
	\node (aab) at (0, 0){\(a \Par (\out{a}+b)\)};
	\node (amamb) at (1.2, 1.2){\(a[m] \Par (\out{a}[m] + b)\)};
	\node (aamb) at (1.2, -1.2){\(a \Par (\out{a}[m] + b)\)};
	\node (akab) at (2.4, 0){\(a[k] \Par (\out{a} + b)\)};
	\node (akamb) at (3.5, -1.2){\(a[k] \Par (\out{a}[m] + b)\)};
	\node (akabm) at (5, 0){\(a[k] \Par (\out{a} + b[m])\)};
	\node (akabn) at (3.5, 1.2){\(a[k] \Par (\out{a} + b[n])\)};
		
	\draw[{Bar[]}->] (aab) -- node[lbl, xshift=-45]{\(\cpair{a[m]}{\lplusl \out{a}[m]}\)} (amamb);
	\draw[{Bar}-{Straight Barb[scale=0.8]}, reverse] (akab)  node[lbl, xshift=-50, yshift=6]{\(\lmidl a[k]\)} -- (aab);
	\draw[{Bar[]}->] (akab) -- node[lbl, xshift=0, yshift=5]{\(\lmidr \lplusr b[m]\)} (akabm);
	\draw[{Bar[]}->] (akab) -- node[lbl, xshift=-26, yshift=0]{\(\lmidr \lplusr b[n]\)} (akabn);
	\draw[{Bar[]}->] (akab) -- node[lbl, xshift=25, yshift=0]{\(\lmidr \lplusl \out{a}[m]\)} (akamb);
	\draw[{Bar[]}->] (aab) -- node[lbl, xshift=25, yshift=0]{\(\lmidr \lplusl \out{a}[m]\)} (aamb);
	\draw[{Bar}-{Straight Barb[scale=0.8]}, reverse] (akamb)  node[lbl, xshift=-62, yshift=5]{\(\lmidl a[k]\)} -- (aamb);
\end{tikzpicture}\vspace{-0.5cm}
	\end{center}	
	\caption{A sample of processes reachable from \(a[k] \Par (\out{a} + b)\).}
	\label{fig:example-transitions}
\end{figure}

\begin{example}\label{ex:transitions}
Some of the processes reachable from \(a[k] \Par (\out{a} + b)\) are presented in~\autoref{fig:example-transitions}. 
Since we suppose infinitely many keys, \pccsk is infinitely branching, as suggested by the transitions labelled \(\lmidr \lplusr b[n]\) and \(\lmidr \lplusr b[m]\): there are others for keys different from $m$, $n$ which are not displayed.	
Following \autoref{lem:loop_proved}, all transitions %
could be reversed (from \(\pr{f}\) to \(\pr{b}\) and vice versa). %
The origin process is \(a \Par (\out a + b)\), %
and all transitions are connected.%
\end{example}

\begin{notation}\label{not:inv-path}
	The empty path is denoted
	\(\es\). Given a path $r = t_1 \cdots t_n$, we write $\rev{r}$ for its inverse path $\rev{t_n} \cdots \rev{t_1}$.
	We also let %
	 $\rev{\es} = \es$, and $\fwdof{t} = t$ if $t$ is forward, $\rev{t}$ otherwise.
\end{notation}

	\section{Complementary Relations for Independence and Dependence}
\label{sec:ind-complem}

This section manipulates three relations defined using only proof keyed labels: an independence relation, a dependence relation, and their \enquote{union}, which captures connectedness of transitions (\autoref{prop:connectednessadequacy}).
The important point is that any two proof keyed labels in the connectedness relation are either dependent or independent (\autoref{thm:complementarity}).
While the dependence relation is inspired by existing works on \pccs\footnote{%
	In a nutshell, \pccs is \pccsk without keyed prefixes and where choices are discarded if not executed: its definition is provided in \autoref{ssec:conservativity}.}
 and \pccsk~\cite{aubert2023c,DeganoGP03}, we are not aware of any direct characterisation of dependence and independence that does not postulate their complementarity.
We prove that our independence relation is a conservative extension over the concurrency relation for \pccs in \autoref{ssec:conservativity}, and prove \autoref{prop:connectednessadequacy} and \autoref{thm:complementarity} in \autoref{ssec:complementarity}.

\begin{definition}[Relations on proof labels]
	\label{def:relation-proof-labels}
	Two proof keyed labels \(\theta_1\), \(\theta_2\) are \emph{connected} (\resp \emph{independent}, \emph{dependent}) if \(\theta_1 \conn \theta_2\) (\resp \(\theta_1 \ind \theta_2\), \(\theta_1 \sdep \theta_2\)) can be derived using the rules in \autoref{fig:pccsk-relations}.
\end{definition}

\begin{remark}
	\label{rem:ind-irreflexivity}
	It is easy to prove that \(\ind\) is irreflexive and symmetric, as S\(^1\) is the mirror of S\(^2\).
\end{remark}

\begin{example}
	\label{ex:part2}
	Re-using the labels from \autoref{fig:example-transitions}, we have \eg 
	\[\begin{alignedat}{3}
		\lmidr \lplusl \out{a}[m] & \conn && \cpair{a[m]}{\lplusl \out{a}[m]} & \text{ by S\(^1\) and  A\(^2\) for \(\conn\).}\\
	\lmidr \lplusl \out{a}[m] & \sdep && \lmidr \lplusr b[n] & \text{ by P\(^1\) and  C\(^2\) for \(\sdep\).} \\
	\lmidl a[k] 		& \ind  && \lmidr \lplusr b[n] & \text{ by P\(^2_k\) for \(\ind\).}
\end{alignedat}\]
\end{example}

We first prove that our notion of connectedness of labels is correct \wrt to our notion of connected transitions (\autoref{def:transitions-paths}):

\begin{restatable}{proposition}{propconnectednessadequacy}
	\label{prop:connectednessadequacy}
	\begin{enumerate}
		\item If \(t_1 : X_1 \pr{fb}[\theta_1] X'_1\) and \(t_2 : X_2 \pr{fb}[\theta_2] X'_2\) are connected  then \(\theta_1 \conn \theta_2\). \label{prop:connectednessadequacy-1}
		\item If \(\theta_1 \conn \theta_2\), then there exist \(t_1 : X_1 \pr{fb}[\theta_1] X'_1\) and \(t_2 : X_2 \pr{fb}[\theta_2] X'_2\) such that \(t_1\) and \(t_2\) are connected.  \label{prop:connectednessadequacy-2}
	\end{enumerate}
\end{restatable}

\begin{remark}
	Note that the converse of \itemreflipics{prop:connectednessadequacy-1} does not hold, as \eg \(a[k] \conn \lmidr b[m]\), but \(a \pr{f}[a][k] a[k]\) and \(\nil \Par b \pr{f}[\lmidr b][m] \nil \Par b[m]\) are not connected.
	However, \eg  \(t_1: a.(\nil \Par b) \pr{f}[a][k] a[k]. (\nil \Par b)\) and \(a[k]. (\nil \Par b) \pr{f}[\lmidr b][m] a[k].(\nil \Par b[m])\) are connected, illustrating \itemreflipics{prop:connectednessadequacy-2}.
\end{remark}

\begin{figure}[ht]
	\begin{tcolorbox}[title = {Connectivity Relation}, fontupper=\linespread{.9}\small, sidebyside, sidebyside align=top]
		\begin{tcolorbox}[adjusted title=Action]
			\makebox[.4\textwidth][c]{
				\begin{prooftree}
					\hypo{}
					\infer[]1[A\(^1\)]{\alpha[k] \conn \theta}
				\end{prooftree}
			}
			\makebox[.4\textwidth][c]{
				\begin{prooftree}
					\hypo{\theta \text{ is not a prefix}}
					\infer[]1[A\(^2\)]{\theta \conn \alpha[k]}
				\end{prooftree}
			}
		\end{tcolorbox}
		\begin{tcolorbox}[adjusted title=Parallel]
			\makebox[.4\textwidth][c]{
				\begin{prooftree}
					\hypo{\theta \conn \theta'}
					\infer[]1[P\(^1\)]{\lmidd \theta \conn \lmidd \theta'}
				\end{prooftree}
			}
			\makebox[.4\textwidth][c]{
				\begin{prooftree}
					\hypo{}
					\infer[]1[P\(^2\)]{\lmidd \theta \conn \lmidod \theta'}
				\end{prooftree}
			}
		\end{tcolorbox}
		\tcblower
		\begin{tcolorbox}[adjusted title=Choice]
			\makebox[.4\textwidth][c]{
				\begin{prooftree}
					\hypo{\theta \conn \theta'}
					\infer[]1[C\(^1\)]{\lplusd \theta \conn \lplusd \theta'}
				\end{prooftree}
			}
			\makebox[.4\textwidth][c]{
				\begin{prooftree}
					\hypo{}
					\infer[]1[C\(^2\)]{\lplusd \theta \conn \lplusod \theta'}
				\end{prooftree}
			}
		\end{tcolorbox}
		
		\begin{tcolorbox}[adjusted title=Synchronisation]
			\begin{prooftree}
				\hypo{\theta \conn \theta_{\D}}
				\infer[]1[S\(^1\)]{\lmidd \theta  \conn \cpair{\theta_{\L}}{\theta_{\R}}}
			\end{prooftree}
			\hfill
			\begin{prooftree}
				\hypo{\theta_{\D} \conn \theta}
				\infer[]1[S\(^2\)]{\cpair{\theta_{\L}}{\theta_{\R}} \conn \lmidd \theta}
			\end{prooftree}
			\\[.8em]
			\begin{prooftree}
				\hypo{\theta_1 \conn \theta'_1}
				\hypo{\theta_2 \conn \theta'_2}
				\infer[]2[S\(^3\)]{\cpair {\theta_1} {\theta_2} \conn \cpair {\theta'_1} {\theta'_2}}
			\end{prooftree}
		\end{tcolorbox}
	\end{tcolorbox}
	\begin{multicols}{2}
		\begin{tcolorbox}[title = {Dependence Relation}, fontupper=\linespread{.9}\small]%
			\begin{tcolorbox}[adjusted title=Action]
				\raisebox{-1.25em}{ %
					\makebox[.4\textwidth][c]{
						\begin{prooftree}
							\hypo{}
							\infer[]1[A\(^1\)]{\alpha[k] \sdep \theta}
						\end{prooftree}
					}
				}
				\raisebox{-1.25em}{ %
					\makebox[.4\textwidth][c]{
						\begin{prooftree}
							\hypo{\theta \text{ is not a prefix}}
							\infer[]1[A\(^2\)]{\theta \sdep \alpha[k]}
						\end{prooftree}
					}
				}
			\end{tcolorbox}
			\begin{tcolorbox}[adjusted title=Choice]
				\makebox[.4\textwidth][c]{
					\begin{prooftree}
						\hypo{\theta \sdep \theta'}
						\infer[]1[C\(^1\)]{\lplusd \theta \sdep \lplusd \theta'}
					\end{prooftree}
				}
				\makebox[.4\textwidth][c]{
					\begin{prooftree}
						\hypo{}
						\infer[]1[C\(^2\)]{\lplusd \theta \sdep \lplusod \theta'}
					\end{prooftree}
				}
			\end{tcolorbox}
			\begin{tcolorbox}[adjusted title=Parallel]
				\makebox[.4\textwidth][c]{
					\begin{prooftree}
						\hypo{\theta \sdep \theta'}
						\infer[]1[P\(^1\)]{\lmidd \theta \sdep\ \lmidd \theta'}
					\end{prooftree}
				}
				\makebox[.4\textwidth][c]{
					\begin{prooftree}
						\hypo{\kay{\theta} = \kay{\theta'}}
						\infer[]1[P\(^2_k\)]{\lmidd \theta \sdep\ \lmidod \theta'}
					\end{prooftree}
				}
			\end{tcolorbox}
			\begin{tcolorbox}[adjusted title=Synchronisation]
				\begin{prooftree}
					\hypo{\theta \sdep \theta_{\D}}
					\infer[]1[S\(^1\)]{\lmidd \theta  \sdep \cpair {\theta_{\L}} {\theta_{\R}}}
				\end{prooftree}
				\hfill
				\begin{prooftree}
					\hypo{\theta_{\D} \sdep \theta}
					\infer[]1[S\(^2\)]{\cpair {\theta_{\L}} {\theta_{\R}} \sdep\ \lmidd \theta}
				\end{prooftree}
				\\[.8em]
				\begin{prooftree}
					\hypo{\theta_i \sdep \theta'_i}
					\hypo{\theta_j \conn \theta'_j}
					\hypo{i, j \in \{1, 2\}, i \neq j}
					\infer[]3[S\(^3\)]{\cpair {\theta_1} {\theta_2} \sdep \cpair {\theta'_1} {\theta'_2}}
				\end{prooftree}
			\end{tcolorbox}
		\end{tcolorbox}
		\begin{tcolorbox}[title = {Independence Relation}, fontupper=\linespread{.9}\small]%
			\begin{tcolorbox}[adjusted title = {Action}]
				\vbox to 23
				pt {\vfil
					\hfill \emph{(empty)} \hfill~
					\vfil
				}
			\end{tcolorbox}
			\begin{tcolorbox}[adjusted title=Choice]
				\makebox[.4\textwidth][c]{
					\begin{prooftree}
						\hypo{\theta \ind  \theta'}
						\infer[]1[C\(^1\)]{\lplusd \theta \ind  \lplusd \theta'}
					\end{prooftree}
				}
				\makebox[.4\textwidth][c]{ %
				}
			\end{tcolorbox}
			\begin{tcolorbox}[adjusted title=Parallel]
				\makebox[.4\textwidth][c]{
					\begin{prooftree}
						\hypo{\theta \ind  \theta'}
						\infer[]1[P\(^1\)]{\lmidd \theta \ind\  \lmidd \theta'}
					\end{prooftree}
				}
				\makebox[.4\textwidth][c]{
					\begin{prooftree}
						\hypo{\kay{\theta} \neq \kay{\theta'}}
						\infer[]1[P\(^2_k\)]{\lmidd \theta \ind \ \lmidod \theta'}
					\end{prooftree}
				}
			\end{tcolorbox}
			\begin{tcolorbox}[adjusted title=Synchronisation]
				\begin{prooftree}
					\hypo{\theta \ind  \theta_{\D}}
					\infer[]1[S\(^1\)]{\lmidd \theta  \ind  \cpair{\theta_{\L}}{\theta_{\R}}}
				\end{prooftree}
				\hfill
				\begin{prooftree}
					\hypo{\theta_{\D} \ind \theta}
					\infer[]1[S\(^2\)]{\cpair{\theta_{\L}}{\theta_{\R}} \ind\ \lmidd \theta}
				\end{prooftree}
				\\[.8em]
				\begin{prooftree}
					\hypo{\theta_1 \ind  \theta'_1}
					\hypo{\theta_2 \ind  \theta'_2}
					\infer[]2[S\(^3\)]{\cpair {\theta_1} {\theta_2} \ind  \cpair {\theta'_1} {\theta'_2}}
				\end{prooftree}
			\end{tcolorbox}
		\end{tcolorbox}
	\end{multicols}
	\caption{Relations on proof labels%
	}
	\label{fig:pccsk-relations}
\end{figure}

The notion of connectedness is now leveraged to reduce our search space on proof keyed labels: for example, \(\lplusl a[k]\) and \(\lmidr b[m]\) are neither dependent nor independent, but they also cannot belong to connected transitions (since no process can have both \(+\) and \(\Par\) at top level).

\begin{restatable}[Complementarity on labels]{theorem}{thmcomplementarity}
\label{thm:complementarity}
For all \(\theta_1\), \(\theta_2\),
\begin{enumerate}
	\item If $\theta_1 \ind \theta_2$ then $\theta_1 \conn \theta_2$.
	\item If $\theta_1 \sdep \theta_2$ then $\theta_1 \conn \theta_2$.
	\item If $\theta_1 \conn \theta_2$ then either $\theta_1 \ind \theta_2$ or $\theta_1 \sdep \theta_2$, but not both.
\end{enumerate}
\end{restatable}

The relations $\ind$ and $\sdep$ from \autoref{def:relation-proof-labels} are easily extended to \pccsk and \ccsk's transitions by using the bijection \(\base{\cdot} = {(\prov{\cdot}})^{-1}\) from \pccsk to \ccsk defined in \autoref{app:sub-systems}.

\begin{definition}[Relations on \pccsk and \ccsk transitions]
	\label{def:ind-on-ccsk}
	For transitions \(t_1\), \(t_2\) in \pccsk (\resp \ccsk), we let \(t_1 \ind t_2\) iff $t_1$ and $t_2$ are connected and \(\lblof{t_1} \ind \lblof{t_2}\) (\resp \(t_1 \ind t_2\) iff  $t_1$ and $t_2$ are connected and \(\lblof{\prov{t_1}} \ind \lblof{\prov{t_2}}\)), and correspondingly for \(\sdep\).
\end{definition}

We easily deduce from \autoref{thm:complementarity}:
\begin{proposition}[Complementarity on transitions]
\label{prop:complementarity transitions}
If $t_1$ and $t_2$ are connected then exactly one of $t_1 \ind t_2$ and $t_1 \sdep t_2$ holds.
\end{proposition}

	\section{Basics of the Axiomatic Approach}
\label{ssec:reminders-axiom}

We shall now recall the basics of the axiomatic approach~\cite{LanesePU20,LPU24}. Once a model of computation (presented as a labelled transition system with independence~\cite{SNW96}) is reversed, for example a process calculus, it is challenging to prove that it satisfies desired properties such as
\emph{causal consistency}~\cite{DK04} or \emph{causal safety} and \emph{causal liveness}~\cite{LanesePU20,LPU24}.
The axiomatic approach allows us to obtain such properties---among others---from simpler axioms.
We recall the basic axioms, introduce
\emph{polychotomies} and the conditions under which these properties hold (\autoref{prop:poly}).
We conclude the section by showing that the independence relation for reversible calculi satisfying appropriate axioms is unique (\autoref{thm:uniqueness}).

\begin{definition}[LTSI~\protect{\cite[Defs 2.1--2.3]{LPU24}}]
	\label{def:ltsi}
	Let $\Proc$ be a set of \emph{processes}, ranged over by $P,Q,\ldots$, and $\Lab$ a set of \emph{labels}, ranged over by $a,b,\ldots$.
	A \emph{combined LTS} is a forward LTS $(\Proc,\Lab,\r{f})$ together with a backward LTS $(\Proc,\Lab,\r{b})$ satisfying the Loop Lemma: $P \r{f}[a] Q$ iff $Q \r{b}[a] P$.
	To refer to transitions which may be either forward or backward, we introduce backward labels $\rev{\Lab} = \{\rev{a} : a \in \Lab\}$, and let $\alpha,\beta$ range over \emph{directed labels}, \ie members of the disjoint union $\Lab \union \rev{\Lab}$.
	Then $P \r{fb}[\alpha] Q$ denotes $P \r{f}[a] Q$ if $\alpha = a$ and $P \r{b}[a] Q$ if $\alpha = \rev{a}$.
	We let $\rev{\rev{a}} = a$ and define the \emph{inverse} $\rev{t}$ of a transition $t: P \r{fb}[\alpha] Q$ to be 
	$\rev{t}: Q \r{fb}[\rev{\alpha}] P$.
	The \emph{underlying label} $\und\alpha$ is defined as $\und{a} = \und{\rev{a}} = a$. %
	
	We say that $(\Proc,\Lab,\r{fb},\ind)$ is a \emph{labelled transition system with independence} (LTSI) if $(\Proc,\Lab,\r{fb})$ is a combined LTS and $\ind$
	is an irreflexive symmetric binary relation on transitions---the \emph{independence} relation.
\end{definition}

\begin{remark}
	Observe that the notation differs from \pccsk's notation, since $X \pr{fb}[\theta] X'$ can mean either $X \pr{f}[\theta] X'$ or $X \pr{b}[\theta] X'$, but $P \r{fb}[\alpha] Q$ means $P \r{f}[a] Q$ or $P \r{b}[a] Q$ depending on $\alpha$.
	This added precision is more consistent with previous work, and will simplify some notations, but we import from \autoref{sec:proved-lts} %
	the notions of (rooted) path, coinitial transitions, etc.
\end{remark}

\begin{figure}[h]
	\Comment{
		\begin{tabularx}{\textwidth}{>{\bfseries}l X}
			\multicolumn{2}{l}{\normalfont{\emph{An LTSI satisfies} \textbf{Property}\Comment{\dots} \emph{if the condition on the right holds:\Comment{\ldots}}}}\\ 
			SP & whenever $t:P \r{fb}[\alpha] Q$, $u:P \r{fb}[\beta] R$ with $t \ind u$ then there are cofinal transitions $u': Q \r{fb}[\beta] S$ and $t':R \r{fb}[\alpha] S$.\\
			BTI & whenever $t:P \r{b}[a] Q$ and $t': P \r{b}[b] Q'$ and $t \neq t'$ then $t \ind t'$.\\
			WF & there is no infinite reverse computation, \ie there are no $P_i$ (not necessarily distinct) such that $P_{i+1} \r{f}[a_i] P_i$ for all $i \in\mathbb{N}$.\\
			PCI & whenever $t:P \r{fb}[\alpha] Q$, $u:P \r{fb}[\beta] R$, $u': Q \r{fb}[\beta] S$ and $t':R \r{fb}[\alpha] S$ with $t \ind u$, then $u' \ind \rev{t}$.\\
			BLD & whenever $t: P \rtran a Q$ and $u: P \rtran a R$ are coinitial backward transitions with the same label then $t = u$.\\
			ID & whenever $t:P \r{fb}[\alpha] Q$, $u:P \r{fb}[\beta] R$, $u': Q \r{fb}[\beta] S$ and $t':R \r{fb}[\alpha] S$, with $Q \neq R$ if $t$ and $u$ have the same direction; $P \neq S$ otherwise; then $t \ind u$.\\
			NRE & for any rooted path $r$ and any event $e$ we have $\cte(r, e)\leq 1$.\\
			\hline
			IRE & whenever $t \sqeqt t' \ind u$ then $t \ind u$. \\
			RPI & whenever $t \ind t'$ then $\rev{t} \ind t'$.\\
		\end{tabularx}	
		\caption{Main properties studied by the axiomatic approach. %
			Properties above the line hold for pre-reversible LTSIs.}
	}
	
	\begin{tcolorbox}[title = {Basic Axioms}]
		\begin{tabularx}{\textwidth}{>{\bfseries}l X}
			SP & whenever $t:P \r{fb}[\alpha] Q$, $u:P \r{fb}[\beta] R$ with $t \ind u$ then there are cofinal transitions $u': Q \r{fb}[\beta] S$ and $t':R \r{fb}[\alpha] S$.\\
			BTI & whenever $t:P \r{b}[a] Q$ and $t': P \r{b}[b] Q'$ and $t \neq t'$ then $t \ind t'$.\\
			WF & there is no infinite reverse computation, \ie there are no $P_i$ (not necessarily distinct) such that $P_{i+1} \r{f}[a_i] P_i$ for all $i \in\mathbb{N}$.\\
			PCI & whenever $t:P \r{fb}[\alpha] Q$, $u:P \r{fb}[\beta] R$, $u': Q \r{fb}[\beta] S$ and $t':R \r{fb}[\alpha] S$ with $t \ind u$, then $u' \ind \rev{t}$.
		\end{tabularx}
	\end{tcolorbox}
	\begin{tcolorbox}[title = {Other Useful Properties}]
		\begin{tabularx}{\textwidth}{>{\bfseries}l X}
			ID & whenever $t:P \r{fb}[\alpha] Q$, $u:P \r{fb}[\beta] R$, $u': Q \r{fb}[\beta] S$ and $t':R \r{fb}[\alpha] S$, with $Q \neq R$ if $t$ and $u$ have the same direction; $P \neq S$ otherwise; then $t \ind u$.\\
			IRE & whenever $t \sqeqt t' \ind u$ then $t \ind u$.  \tikzmark{top} \\
			RPI & whenever $t \ind t'$ then $\rev{t} \ind t'$.\\
		\end{tabularx}
		
	\end{tcolorbox}
	\caption{%
		Main properties studied by the axiomatic approach.
		In the tables above, an LTSI satisfies \textbf{Property} if the condition on the right holds.
	}
	\label{fig:axiomatic-properties}
\end{figure}

\begin{definition}[Axiomatic properties and pre-reversible LTSI]
	\label{def:basic}
	\label{def:other-properties}
	\label{def:prerev}
	\autoref{fig:axiomatic-properties} presents
	\begin{description}
		\item[the basic axioms] \emph{Square property} (SP), \emph{Backward transitions are independent} (BTI), \emph{Well-founded} (WF) and \emph{Propagation of coinitial independence} (PCI)~\cite[Defs.\ 3.1, 4.2]{LPU24};
		\item[other useful properties] %
		\emph{Independence of diamonds} (ID)%
		, %
		\emph{Independence respects events} (IRE) and \emph{Reversing preserves independence} (RPI). %
	\end{description}
	An LTSI is \emph{pre-reversible} if it satisfies the basic axioms.%
\end{definition}

\Comment{
	\clem{Maybe we can get rid of the definition environment, to gain a bit of space and merge with the following paragraph? Or we could merge  \autoref{def:other-properties} with it.}
	\begin{definition}[Basic axioms~\protect{\cite[Def. 3.1 and 4.2]{LPU24}}]
		\label{def:basic}
		The \emph{Square property} (SP), \emph{Backward transitions are independent} (BTI), \emph{Well-founded} (WF) and \emph{Propagation of coinitial independence} (PCI) properties are given in %
		\autoref{fig:axiomatic-properties}.
	\end{definition}
}

SP and BTI are complementary: SP expresses soundness of a definition of independence, such as the one given in \autoref{sec:ind-complem},
while BTI expresses its completeness.
We save for \autoref{ssec:instantiating} the proof that \pccsk and \ccsk fulfill the requirements to leverage the axiomatic approach.

\Comment{
	\begin{definition}[Pre-reversible LTSI~\protect{\cite[Def. 4.3]{LPU24}}]
		\label{def:prerev}
		An LTSI is \emph{pre-reversible} if it satisfies SP, BTI, WF and PCI.
	\end{definition}
}

\begin{definition}[Event~\protect{\cite[Def. 4.1]{LPU24}}]
	\label{def:event-general}
	Consider a pre-reversible LTSI and let $\sqeqt$ be the smallest equivalence relation satisfying:
	if $t:P \r{fb}[\alpha] Q$, $u:P \r{fb}[\beta] R$,
	$u':Q \r{fb}[\beta] S$, $t':R \r{fb}[\alpha] S$,
	and $t \ind u$, then $t \sqeqt t'$ and $\rev{t} \sqeqt t'$.
	The equivalence classes of transitions, written $[t]$, are the \emph{events}.
		We say that an event is \emph{forward} if it is the equivalence class of a forward transition; similarly for reverse events.
		Given an event \(e = [t]\) we let \(\rev{e} = [\rev{t}]\).
\end{definition}

\begin{definition}[Counting events in path~\protect{\cite[Def. 4.11.]{LPU24}}]
	\label{def:count-events}
	Let $r$ be a path and $e$ be an event, %
	we define $\cte(r,e)$ %
	as follows, for \(tr\) the path made of \(t\) followed by \(r\):
	\begin{align*}
		\cte(\es,e) & = 0 &&& 
		\cte(tr,e) & =
		\begin{cases}
			\cte(r,e)+1 & \text{if } t \in e \\ %
			\cte(r,e) -1 & \text{if } t \in \rev e \\
			\cte(r,e) & \text{otherwise} %
		\end{cases}
	\end{align*}
\end{definition}

\begin{definition}[Relations on events, transitions~\protect{\cite[Defs 4.14, 4.23, 4.27]{LPU24}}]\label{def:event relations}
	Two events $e$, $e'$ are 
	\begin{itemize}
		\item \emph{Core independent}\footnote{Called \enquote{coinitially independent} in \cite{LPU24}, but this is confusing when $\coind$ is extended to (not necessarily coinitial) transitions.}, written $e \coind e'$, iff there are coinitial transitions \(t\), \(t'\) such that $[t] = e$, $[t'] = e'$ and $t \ind t'$.
	\end{itemize}
	Two forward events $e$, $e'$ are 
	\begin{itemize}
		\item \emph{Causally related}, written $e \leq e'$, iff for all rooted paths $r$, if $\cte(r,e') > 0$ then $\cte(r,e) > 0$.
		\item In \emph{conflict}, written $e \Cf e'$, iff there is no rooted path $r$ such that $\cte(r,e) >0$ and $\cte(r,e') > 0$.
	\end{itemize}
	We write $e < e'$, $e$ is a \emph{cause} of $e'$, if $e \leq e'$ and $e \neq e'$.
	We also extend those relations to transitions by letting $t_1 \coind t_2$ iff $[t_1] \coind [t_2]$, and similarly for forward transitions for causal ordering and conflict.
\end{definition}

\begin{proposition}[Polychotomies~\protect{\cite[Def. 4.28, Proposition 4.29]{LPU24}}]
	\label{prop:poly}
	Pre-reversible LTSIs satisfy \emph{polychotomy for forward events}, \ie for all forward events \(e\), \(e'\), exactly one of the following holds: %
	\begin{multicols}{5}
		\begin{enumerate}
			\item $e = e'$;
			\item $e < e'$;
			\item $e' < e$;
			\item $e \Cf e'$; or
			\item $e \coind e'$.
		\end{enumerate}
	\end{multicols}
	Being defined on events, $<$, $\Cf$ and $\coind$ on transitions are closed under $\sim$, and pre-reversible LTSIs also satisfy \emph{polychotomy for forward transitions}: for all transitions $t$, $t'$, exactly one of the following holds: 
	\begin{multicols}{5}
		\begin{enumerate}
			\item $t\sim t'$;
			\item  $t< t'$;
			\item $t'< t$;
			\item $t \Cf t'$; or
			\item $t\coind t'$.
		\end{enumerate}
	\end{multicols}
\end{proposition}
\begin{definition}[Immediate predecessor]\label{def:immed pred}
Let $e_1,e_2$ be forward events, $e_1$ is an \emph{immediate predecessor} of $e_2$ (written $e_1 \ip e_2$) if $e_1 < e_2$ and there is no event $e$ such that $e_1 < e < e_2$.
\end{definition}

\begin{definition}[Composable events]\label{def:compose events}
Let $e_1,e_2$ be events, $e_1$ is \emph{composable} with $e_2$ if there are transitions $t_1 \in e_1$ and $t_2 \in e_2$ such that $t_1$ is composable with $t_2$.
\end{definition}

The following result (proved in \autoref{app:axioms}) will be used in the proof of \autoref{thm:ordering-event-key}.
\begin{restatable}[Immediate predecessor is not compatible with core independence]{lemma}{immedpredcomp}\label{lem:immed pred composable}
Let $e_1,e_2$ be forward events in a pre-reversible LTSI satisfying \ire and RPI.  Then $e_1 \ip e_2$ iff $e_1$ is composable with $e_2$ and not $e_1 \coind e_2$.
\end{restatable}

We conclude this section by proving that under the mild condition of \enquote{admitting} pre-reversibility (\autoref{def:admit-prever}), any LTSI requiring its independence relation to satisfy SP, BTI and PCI accepts a unique independence relation, and that this relation further uniquely determines the notions of event equivalence, core independence, causal ordering and conflict. Additional material and proofs are in \autoref{app:axioms}. %
\begin{definition}[Admitting pre-reversibility]\label{def:admit-prever}
	A combined LTS $(\Proc,\Lab,\r{fb})$ \emph{admits pre-reversibility} if %
	there exists $\ind$ \st $(\Proc,\Lab,\r{fb},\ind)$ is a pre-reversible LTSI. %
\end{definition}	

\begin{restatable}[Uniqueness]{proposition}{uniqueness}\label{thm:uniqueness}
	If a combined LTS admits pre-reversibility %
		and we require any independence relation to satisfy SP, BTI and PCI,
	then the notions of %
	event equivalence, core independence $\coind$, causal ordering $\leq$ and conflict $\#$ %
	are uniquely determined.
\end{restatable}

We are not aware of any previous such uniqueness result.
Its novelty is that, instead of providing a definition of, for example, independence, and then \enquote{manually} proving that it satisfies various properties, we fix the properties (our basic axioms) and show that there can be only one independence 
relation satisfying them, and similarly for the causality and conflict relations. %

	\section{Properties of \texorpdfstring{\pccsk and \ccsk}{CCSK and PCCSK}} %
\label{sec:properties}

Properties of \pccsk and \ccsk may be divided into two classes:
\textcolor{lipicsGray}{\sffamily\bfseries\upshape\mathversion{bold}1.\@\xspace} those that are derivable from general axioms, which hold for \pccsk and \ccsk by \autoref{thm:axioms-hold} below, and
\textcolor{lipicsGray}{\sffamily\bfseries\upshape\mathversion{bold}2.\@\xspace}  those that depend in some way on the key structure, and do not follow from axioms alone.

\subsection{Instantiating the Axiomatic Approach to \texorpdfstring{\pccsk}{CCSKP} and \texorpdfstring{\ccsk}{CCSK}}
\label{ssec:instantiating}

We now show that instantiating the axiomatic approach to \pccsk and \ccsk, where the independence relation $\ind$ is as in 
\autoref{def:ind-on-ccsk},   produces
 pre-reversible LTSIs that also satisfy IRE and RPI. As a result of \autoref{thm:axioms-hold}, the independence relation of \pccsk and \ccsk is 
 \enquote{the only one} thanks to uniqueness (\autoref{thm:uniqueness}).

\Comment{%
	This is the first time that independence has been defined directly for \pccsk.  Independence was defined as the complement of dependence in~\cite{aubert2023c}.  The rules are adapted from those for CCS in~\cite{BC88,BC94}%
	, with the additional requirement that keys must differ.
	\iain{A natural question is how do we know that our definition here is equivalent to that of~\cite{aubert2023c}.  This could be shown using our rules and complementarity result, say in an appendix, or else for coinitial transitions it could be obtained via uniqueness, since the two notions both satisfy the axioms (after correcting for different keys).}}
\Comment{
\autoref{def:relation-proof-labels} is easily extended to \pccsk and \ccsk's transitions by using the bijection from \autoref{def:sub-systems}:

\begin{definition}[Relations on \pccsk and \ccsk transitions]
	\label{def:ind-on-ccsk}
	For all connected \(t_1\), \(t_2\) in \pccsk (\resp \ccsk), we let \(t_1 \ind t_2\) iff \(\lblof{t_1} \ind \lblof{t_2}\) (\resp \(t_1 \ind t_2\) iff \(\lblof{\prov{t_1}} \ind \lblof{\prov{t_2}}\)), and similarly for \(\sdep\).
\end{definition}
}

\begin{restatable}[The axiomatic approach is applicable to the LTSIs of \ccsk and \pccsk]{theorem}{axiomshold}\label{thm:axioms-hold}
	SP, BTI, WF, PCI, IRE and RPI hold for the LTSIs of \pccsk and \ccsk.
\end{restatable}
See \autoref{app:axioms-hold} for the proof.
SP, BTI, WF and PCI were already shown for \pccsk~\cite{aubert2023c}, using as an independence relation the complement of a dependence relation, rather than defining it directly as in this work\footnote{
The complementarity result from \autoref{ssec:complementarity} makes it \emph{almost} equivalent---the original definition of dependence was missing a case, as we explain in \autoref{rem:changes}.
Also, we require independent transitions to be connected, whereas in~\cite{aubert2023c} they had to be coinitial or composable.}.
IRE and RPI were already shown for \pccsk in~\cite{LPU24}.
We can transfer these results from \pccsk to \ccsk
using their bijection.

We can lift independence and dependence from transitions to events.
\begin{definition}[Relations on events]\label{def:ind sdep event}
	Let $e_1,e_2$ be events in the LTSI of \pccsk.
	\begin{enumerate}
		\item 
		$e_1,e_2$ are \emph{connected} if there are transitions $t_1 \in e_1$ and $t_2 \in e_2$ such that $t_1,t_2$ are connected;
		\item
		$e_1 \ind e_2$ if there are transitions $t_1 \in e_1$ and $t_2 \in e_2$ such that $t_1 \ind t_2$;
		\item
		$e_1 \sdep e_2$ if there are transitions $t_1 \in e_1$ and $t_2 \in e_2$ such that $\lblof{t_1} \sdep \lblof{t_2}$.
	\end{enumerate}
\end{definition}
Since the LTSI of \pccsk satisfies IRE by \autoref{thm:axioms-hold}, if $e_1 \ind e_2$ then $t_1 \ind t_2$ for any $t_1 \in e_1$, $t_2 \in e_2$;
similarly for $e_1 \sdep e_2$ using complementarity on connected transitions (\autoref{prop:complementarity transitions}).

\begin{remark}
Note that we are not entirely relying on axioms for LTSIs here, since we have defined dependence directly for \pccsk; if we wanted to be purely axiomatic to analyse the LTSI of some other reversible calculus we could instead have \emph{defined} dependence to be the complement of independence.
\end{remark}
The next result will be used in the proof of \autoref{thm:ordering-event-key}.
See \autoref{app:ind sdep coind} for the proof.
\begin{restatable}[Complementarity for events]{lemma}{indsdepcoind}\label{lem:ind sdep coind}
	Let $e_1,e_2$ be connected events in the LTSI of \pccsk.
	\begin{enumerate}
		\item 
		exactly one of $e_1 \ind e_2$ and $e_1 \sdep e_2$ holds;
		\item
		if $e_1 \coind e_2$ then $e_1 \ind e_2$;
		\item
		if $e_1,e_2$ are composable and $e_1 \ind e_2$ then $e_1 \coind e_2$.
	\end{enumerate}
\end{restatable}

The following example shows that the independence relation 
$\ind$ based on proof labels for \pccsk is more general than core independence $\coind$
for non-composable events.
\begin{example}\label{ex:ind-causal}
	Consider $(a.b\Par \out{b}.c) \bs  b$. Executing $a$ and then $c$ (after the synchronisation on $b$) produces transitions $t_a$ and $t_c$ respectively:
	\begin{align*}
		t_a : (a.b\Par \out{b}.c) \bs  b & \pr{f}[\lmidl a][k] (a[k].b \Par \out{b}.c) \bs  b \\ %
		t_c  : (a[k].b[l]\Par \out{b}[l].c) \bs  b & \pr{f}[\lmidr c][n] (a[k].b[l]\Par \out{b}[l].c[n]) \bs  b
	\end{align*}
	The labels $\lmidl a[k], \lmidr c[n]$ of $t_a, t_c$ respectively are independent by P\(^2_k\) (since producing \(t_c\) required to have $k \neq n$). However, polychotomy (\autoref{prop:poly}) yields
	$[t_a] \not\coind [t_c]$,
	since $c$ causally depends on $a$: $[t_a]<[t_c]$\footnote{This can also be observed considering that \(t_a\) and \(t_c\) are neither coinitial nor event equivalent to coinitial transitions.}.
	\Comment{\begin{align*}
			t_a &: (a.b\Par \out{b}.c) \bs  b \pr{f}[\lmidl a][k] (a[k].b \Par \out{b}.c)  \bs  b &&& t_a \ind t_c \text{ by P\(^2_k\) provided $k \neq l$}\\
			t_c &: (a[k].b[m]\Par \out{b}[m].c) \bs  b \pr{f}[\lmidr c][l] (a[k].b[m]\Par \out{b}[m].c[l]) \bs  b &&& t_a \not\coind t_c \text{ but } t_a<t_c
		\end{align*}
		Executing $a,c$ produces transitions $t_a, t_c$ with labels $ \lmidl a[k], \lmidr c[l]$, respectively,  which are independent by P\(^2_k\) provided $k \neq l$. However, $t_a \not\coind t_c$ since they are not coinitial, but clearly $c$ causally depends on $a$: $t_a<t_c$.}
\end{example}

\Comment{
\subsection{Uniqueness From an Axiomatic Perspective}
\label{ssec:uniqueness}

\clem{This needs to become a simple comment. Do we keep the proofs in appendix?}

We now prove that under the mild condition of \enquote{admitting} pre-reversibility (\autoref{def:admit-prever}), any LTSI requiring its independence relation to satisfy SP, BTI and PCI accepts a unique independence relation, and that this relation further uniquely determines the notions of event equivalence, coinitial independence, causal ordering and conflict (\autoref{thm:uniqueness}).

\begin{lemma}[Non-degenerate diamond~\protect{\cite[Lem. 4.7]{LPU24}}]
	\label{lem:non-degenerate}
	If an LTSI is pre-reversible and we have a diamond
	$t:P \r{fb}[\alpha] Q$, $u:P \r{fb}[\beta] R$ with $t \ind u$
	and cofinal transitions $u': Q \r{fb}[\beta] S$, %
	$t': R \r{fb}[\alpha] S$,
	then the diamond is \emph{non-degenerate},
	meaning that $P,Q,R,S$ are distinct processes.
\end{lemma}

\begin{lemma}[\protect{\cite[Prop. 4.10]{LPU24}}]\label{prop:ID}
	If an LTSI satisfies BTI and PCI then it satisfies ID.
\end{lemma}

\begin{proposition}[Uniqueness of coinitial independence]
	\label{prop:prerev-coinitial-unique}
	Suppose that $(\Proc,\Lab,\r{fb},\ind_1)$ and $(\Proc,\Lab,\r{fb},\ind_2)$ are two pre-reversible LTSIs with the same underlying combined LTS. Then $\ind_1$ and $\ind_2$ agree on coinitial transitions. 
\end{proposition}
\begin{proof}
	Suppose that $t_1$ and $t_2$ are coinitial transitions with $t_1 \ind_1 t_2$.
	By SP for $\ind_1$ we get a square.
	This is non-degenerate (all states are distinct) by \autoref{lem:non-degenerate}.
	We deduce that $t_1 \ind_2 t_2$ using ID which holds by \autoref{prop:ID}.
	By symmetry we deduce the result.
\end{proof}

\Comment{
	\begin{definition}\label{def:admit-prever}
		A combined LTS $(\Proc,\Lab,\r{fb})$ \emph{admits pre\hyph reversibility} if it satisfies WF, and an independence relation $\ind$ can be added to form an LTSI $(\Proc,\Lab,\r{fb},\ind)$ satisfying SP, BTI and PCI.
	\end{definition}
}
\begin{definition}[Admitting pre-reversibility]\label{def:admit-prever}
	A combined LTS $(\Proc,\Lab,\r{fb})$ \emph{admits pre-reversibility} if %
	there exists $\ind$ \st $(\Proc,\Lab,\r{fb},\ind)$ is a pre-reversible LTSI. %
\end{definition}

\begin{theorem}[Uniquenesses]\label{thm:uniqueness}
	If a combined LTS admits pre-reversibility %
		and we require any independence relation to satisfy SP, BTI and PCI,
	then the notions of %
	event equivalence, coinitial independence $\coind$, causal ordering $\leq$ and conflict $\#$ %
	are uniquely determined.
\end{theorem}
\begin{proof}
	The independence relation $\ind$ is determined for coinitial transitions
	by \autoref{prop:prerev-coinitial-unique}.  %
	From \autoref{def:event-general}, event equivalence \clem{R3 found this term confusing, as it is an equivalence on transitions. I would support a renaming.} $\eveqt$ %
	is determined by $\ind$ on coinitial transitions.  Furthermore, it is clear from \autoref{def:event relations} that $\coind$, $\leq$
	and $\Cf$ are then uniquely determined.
\end{proof}

We are not aware of any previous such uniqueness result.
Its novelty is that, instead of providing a definition of, for example, independence, and then \enquote{manually} proving that it satisfies various properties, we fix the properties (our basic axioms) and show that there can be only one independence 
relation satisfying them, and similarly for the causality and conflict relations. %
We now apply this result to \pccsk after endowing it with an independence relation. %

}

\subsection{Key-Based Properties of \texorpdfstring{\pccsk and \ccsk}{CCSK and PCCSK}}
\label{ssec:key-props-ccsk}
The results in this section depend in some way on the key structure of  \pccsk and \ccsk.
A basic design decision in both systems
is that fresh keys are chosen when computing forwards, so that past events will have different keys.
We shall see that we can decide whether transitions belong to the same event simply using keys (\autoref{prop:event_coincide}). Moreover, it is possible to decide whether one event caused another by purely syntactic means (\autoref{thm:ordering-event-key}).
Proofs and additional material for this section are in Sections~\ref{app:key-based properties} and \ref{app:ordering-event-key} (for \autoref{thm:ordering-event-key}).

\begin{restatable}[Independence implies different keys]{lemma}{inddiffkeys}\label{lem:ind-diff-keys}
For both \pccsk and \ccsk, if $t_1,t_2$ are %
 transitions such that $t_1 \ind t_2$, then $t_1$ and $t_2$ have different keys.
\end{restatable}

\begin{restatable}[Backward key determinism]{lemma}{bwdkeydet}\label{lem:bwd_key_det}
For both \pccsk and \ccsk, if $t_1, t_2$ are both backward transitions and $\key{t_1} = \key{t_2}$, then $t_1 = t_2$.
\end{restatable}

By definition, for any LTSI if $t_1 \eveqt t_2$ then $t_1$ and $t_2$ must have the same labels.
However, the converse is false.  In \pccsk consider $t_1: a[m].b \pr{f}[b][k] a[m].b[k]$ and $t_2: a[n].b \pr{f}[b][k] a[n].b[k]$.  These transitions have the same proof labels, but are not the same event, since they are caused by different events.  However, we can show that if \pccsk or \ccsk transitions $t_1,t_2$ have the same key then $t_1 \eveqt t_2$, provided that there is a path connecting the target processes which does not use their common key.
The key-based definition of events for \pccsk and \ccsk below is simpler than in the general axiomatic approach~\cite{LPU24}---recalled in \autoref{def:event-general}---, as well as the earlier definition of~\cite{PU07a}.

\begin{definition}[Event key equivalence]
	\label{def:event-key}
	Two forward \pccsk transitions with the same key $t_1:X_1 \pr{f}[\theta_1][k] X'_1$ and $t_2:X_2 \pr{f}[\theta_2][k] X'_2$ are \emph{event key equivalent} ($t_1 \evkeqt t_2$) if there is a path $r: X'_1 \pr{bf}^* X'_2$ such that $k \notin \keys{r}$.
	We extend to all transitions (forward or backward) by letting $t_1 \evkeqt t_2$ iff $\fwdof {t_1} \evkeqt \fwdof {t_2}$.
	Similarly for \ccsk.
\end{definition}

\Comment{
\begin{restatable}{proposition}{eventcoincide}\label{prop:event_coincide}
The two definitions of event equivalence coincide: $t_1 \eveqt t_2$ (\autoref{def:event-general}) iff $t_1 \evkeqt t_2$ (\autoref{def:event-key}). \iain{need to clarify \pccsk vs \ccsk---holds for both.}
\end{restatable}
}

\begin{restatable}[Event equivalences coincide]{proposition}{eventcoincide}\label{prop:event_coincide}
For %
 \pccsk and \ccsk, $t_1 \eveqt t_2$ %
 iff $t_1 \evkeqt t_2$.%
\end{restatable}

Besides being simpler to work with than \autoref{def:event-general}, a crucial difference between the two definitions is that \autoref{def:event-key} does not make any use of independence.  We could build on this independence-free notion of event to obtain causation and conflict between events (\autoref{def:event relations}) without using independence.
We can also formulate coinitial independence using keys.
We do this for \ccsk, since the proof labels are not relevant.
\begin{definition}[Key independence for \ccsk]\label{def:key ind}
\begin{enumerate}
\item
Whenever $t:P \r{fb}[\alpha][m] Q$, $u:P \r{fb}[\beta][n] R$, $u': Q \r{fb}[\beta][n] S$ and $t':R \r{fb}[\alpha][m] S$, with $m \neq n$ then $t,u$ are \emph{directly key independent};
\item Connected \ccsk transitions $t,u$ are \emph{key independent} if $t \evkeqt t',u \evkeqt u'$ with $t',u'$ directly key independent.
\end{enumerate}
\end{definition}
Key independence corresponds to coinitial independence as defined using the axiomatic approach and independence on proof labels.
\begin{restatable}[Coinitial independences coincide]{lemma}{directkeyind}\label{prop:direct key ind}
For all coinitial transitions $t,u$ in \ccsk, $t,u$ are directly key independent iff $t \ind u$.
\end{restatable}

\begin{restatable}[Independences coincide]{proposition}{keyindcoind}\label{prop:key ind coind}
For all transitions $t,u$ in \ccsk, $t,u$ are key independent iff $t \coind u$.
\end{restatable}

Note that since \ccsk satisfies IRE (\autoref{thm:axioms-hold}),
$t \coind u$ implies $t \ind u$.  The converse does not hold by \autoref{ex:ind-causal}, adapted from \pccsk to \ccsk.
Thus key independence is strictly finer than independence via proof labels.

We conclude this section by showing that the causal ordering on past events (using \autoref{def:event relations}) can be equivalently computed by a syntactic ordering on keys.

\begin{definition}[Partial order on keys~\protect{\cite[Def. 3.1]{LaneseP21}}]
	\label{def:partial-order-on-keys}
	Given a process \(X\), its \emph{partial order on \(\keys{X}\)}, written \(\leq_{X}\), is the reflexive and transitive closure of \(\ord{X}\):%
	\begin{align*}
		\ord{\nil} & = \emptyset  & \ord{X + Y} &= \ord{X} \cup \ord{Y} \\
		\ord{\alpha.X} & = \ord{X} & \ord{X \Par  Y} &= \ord{X} \cup \ord{Y} \\
		\ord{\alpha[n].X} & =  \ord{X} \cup \{n < k \setst k \in \keys{X}\} &   \ord{X \bs \lambda} & = \ord{X} 
	\end{align*}
\end{definition}

\begin{restatable}[Events in a process]{definition}{eventsprocess}\label{def:events process}
	We let \(\events\) be the set of all forward events, and given %
	a reachable process \(X\), we let 
	\(\ev{X} = \{e \in \events \setst \exists \text{ rooted path } r \text{ to } X, \cte(r,e) > 0\}\).
\end{restatable}

\begin{restatable}[Orderings coincide]{theorem}{thmorderings}
	\label{thm:ordering-event-key}
	For any process \(X\), if $e_1,e_2 \in \ev{X}$ we have: $e_1 \leq e_2$ iff $\key {e_1} \leq_X \key {e_2}$.
\end{restatable}

	\section{KP and DP Bisimulations}%
\label{sec:bisimulations}
Our aim is to define bisimulations for \pccsk which have the same distinguishing power as History Preserving (HP) 
bisimulation~\cite{Rabinovich1988}\footnote{%
	Proving that they actually coincide would require us to develop a formal correspondence between behaviour 
	structures~\cite{Rabinovich1988} and \pccsk, but bisimulations defined on \eg automata~\cite{FAHRENBERG2013165}, process graphs~\cite{vG96} or Petri nets~\cite{Vogler1991} are routinely named HP %
using informal arguments.
}. Since we define them for standard processes, our bisimulations can be seen as for \ccs or \pccs.
In van Glabbeek and Goltz's definition of HP bisimulation on configuration structures in~\cite{vGG01}, events of bisimilar structures have 
matching labels and there is a bijection between causal orders on events in the matching structures.  
Here we show how to define syntactically causal orders and bijections between them using our new results presented so far.

\Comment{Two natural questions that emerge at this point are whether causal order--as captured through keys, thanks to \autoref{thm:ordering-event-key}--and dependence (or independence) are powerful enough to define interesting bisimulations, namely, History Preserving ones.
The first section motivates the reasons why dependence is \enquote{enough} \emph{provided the focus is on adjacent transitions}.
}

\subsection{Pinpointing the Relevant Relations}
\label{ssec:dep-is-enough}
We have seen in \autoref{thm:ordering-event-key} that causal order between events of any reachable \pccsk process can be equally represented by the order among their keys. This gives rise to our first new bisimulation relation called Key-Preserving (KP) bisimulation (\autoref{def:HPbisim}). 

A question arises if we can use the dependence relation $\sdep$ or the independence relation $\ind$ to capture the causal order between events.  We have seen in \autoref{ex:ind-causal} that some causally dependent connected transitions, which are not composable, have independent labels, so simply using 
$\sdep$ would not be sufficient to represent $<$. The example below additionally shows that $\tsdep$, the transitive closure of $\sdep$, captures 
more than just the causal order. %
Hence, $\tsdep$ cannot be used instead of
$\leq$ on connected forward transitions in a potential characterisation of HP bisimulation.  

\begin{figure}
\centering
	\begin{tikzpicture}[
	x={(1, 0)},  %
	y={(0, 1)}, %
	baseline,
	anchor=base
	]
	\node (z) at (-1, 0){};	
	\node (a) at (0, 0){};
	\node (b) at (1, -1){};
	\node (c) at (1, 1){};
	\node (d) at (2, 0){};
	\node (e) at (3.41, 0){};
	\node (f) at (2.41, -1){};
	
	\draw[->] (z) -- node[above, pos=.3]{\(t\)} (a);
	\draw[->] (a) -- node[below, pos=.3]{\(t_2\)} (b);
	\draw[->] (a) -- node[above, pos=.3]{\(t_1\)} (c);
	\draw[->] (b) -- node[above, pos=.3]{\(t_1'\)} (d);
	\draw[->] (c) -- node[above, pos=.6]{\(t_2'\)} (d);
	\draw[->] (d) -- node[above, pos=.5]{\(t_3'\)} (e);
	\draw[->] (b) -- node[below, pos=.5]{\(t_3\)} (f);
	\draw[->] (f) -- node[above, pos=.3]{\(t_1''\)} (e);	
\end{tikzpicture}
\caption{LTSI for \(a. (b \Par c.d)\) in \autoref{conc-tran-clo}.}
\label{fig:LTSIs}
\end{figure}

\begin{restatable}{example}{conctranclo}
	\label{conc-tran-clo}
	In \autoref{fig:LTSIs}, letting \(t \sdep t_1\), \(t \sdep t_2\), \(t_2 \sdep t_3\) and \(t_2' \sdep t_3'\) and all other pairs of transitions have independent labels %
	gives a pre-reversible (and LLG) LTSI.
	Clearly $t_1, t_2$ are adjacent, $t_1\coind t_2$ and $t_1\tsdep  t_2$ since they depend on $t$ and \(\sdep\) is symmetric (\autoref{rem:ind-irreflexivity}).
	Correspondingly, $t_1\coind t_3$ and $t_1\tsdep  t_3$ but they are not adjacent.
\end{restatable}

Recall that we have defined the immediate predecessor order $\ip$ on events (\autoref{def:immed pred}), which when transitively closed gives us
back the original causal order $<$. This means that $\ip$ is a concise representation of our causal order. Noting 
that $e_1 < e_2$ iff $e_1 \ip e_2$ for composable forward events, Lemmas~\ref{lem:immed pred composable} and \ref{lem:ind sdep coind} show that causal order between composed events is precisely represented by dependence of their proof labels:

\begin{restatable}{lemma}{conntrandepiscausality}
	\label{lem:im-fcau-iff-dep}
	Let $e_1, e_2$ be composable forward events, then $e_1 < e_2$ iff $e_1 \sdep e_2$.
	\Comment{
	Then  \begin{enumerate}
		\item $t_1 < t_2$ iff $t_1 \sdep t_2$; \iain{this will follow from \autoref{lem:immed pred composable} and \autoref{lem:ind sdep coind}, noting that $t_1 < t_2$ iff $t_1 \ip t_2$ for composable forward transitions} \irek{Noted below.}
		
		\item $[t_1] < [t_2]$ iff $[t_1] \sdep [t_2]$.
	\end{enumerate}
	} 
\end{restatable}
This allows us to define an alternative version of HP bisimulation, called Dependence-Preserving (DP) bisimulation,
where causal order is expressed concisely by dependence $\sdep$ between events and their immediate predecessors.

\Comment{
\iain{Looking at Lemma 6.1 - do we need the case for transitions?  It seems not, apart from the proof of Lemma E.4 which I think we will be dropping. \clem{Yes, E.4 can go.}  It follows more cleanly from 4.11 and 5.2 if we just state it for forward events.}
\clem{Is there any other technical development that needs to be there? We can include some discussion, but it seems that this is all we need.}
\clem{Moving FR's definition here would allow to better motivate our choice / problem I think.}
}

\subsection{Bisimulations}
\label{ssec:bisim}

\Comment{
Only forward transitions are used and syntactic conditions are placed on the matched processes.  
We follow the style of definition as in~\cite{vG96,vGG01}, where a causal order-preserving bijection between the matched systems is used instead of behaviour structures~\cite{Rabinovich1988}.
Our bijections are defined syntactically: 
bijection of KP bisimulation (\autoref{def:HPbisim}) preserves the ordering on keys, and the bijection of DP bisimulation (\autoref{def:Dsim}) is defined using dependence (and implicitly causal order%
) relation on maximal 
events (\autoref{def:maximal-event}).
}

We first define label- and order-preserving bijections:

\begin{definition}[Label- and order-preserving bijection]\label{def:ord-pres}
Given \pccsk processes \(X\), \(Y\), a bijection \(f: \ev{X} \to \ev{Y}\) is \emph{label preserving} if \(\forall e \in \ev{X}\), \(\labl{e} = \labl{f(e)}\), with \(\labl{e}\) defined as \(\labl{t}\) for \(t \in e\).
It is \emph{order preserving} %
if \(\key{e} \leq_{X} \key{e'} \iff \key{f(e)} \leq_{Y} \key{f(e')}\).
\end{definition}
We observe that any order-preserving bijection also preserves causal order by \autoref{thm:ordering-event-key}.

\begin{definition}[KP bisimulation (inspired by \protect{\cite[Sect. 3]{PhillipsU12}})]
	\label{def:HPbisim}
	Let \(X\) and \(Y\) be standard  \pccsk processes.
	A relation \(\Rel{r}[KP] \subseteq \kprocset \times \kprocset \times (\events \to \events)\) is a \emph{Key-Preserving (KP) bisimulation between \(X\) and \(Y\)} if %
	 \((X, Y, \emptyset) \in \Rel{r}[KP]\), %

	and if whenever \((X', Y', f') \in \Rel{r}[KP]\), then \(f'\) is a label- and order-preserving bijection from \(\ev{X'}\) to \(\ev{Y'}\) and
	\begin{align}
	\forall t:  X' \pr{f}[\theta] X'' \Rightarrow \; & 
	\exists t': Y' \pr{f}[\theta']  Y''
	\text{ and } (X'', Y'', f'\cup\{[t] \mapsto [t']\}) \in \Rel{r}[KP];  \label{cond1-kp}\\
	\forall t':  Y' \pr{f}[\theta]  Y'' \Rightarrow \; &  
	\exists t: X' \pr{f}[\theta']  X'' \text{ and } (X'', Y'', f'\cup\{[t] \mapsto [t']\})\in \Rel{r}[KP]. \label{cond2-kp}
\end{align}
	Given any \pccsk processes $X, Y$, if there is a KP bisimulation \(\Rel{r}[KP]\) 
between \(\orig{X}\) and \(\orig{Y}\), 
	such that \((X, Y, f)\in \Rel{r}[KP]\) for some $f$, 
	then we say that $X$ and $Y$ are KP bisimilar, written $X \Rel{b}[KP] Y$.
\end{definition}
\begin{remark}
If we only consider KP bisimilarity $\Rel{b}[KP]$ on standard processes (as in \autoref{thm:bisim}), then we can take 
the mappings $f$ to be the identity on keys, as in FR bisimilarity.
\end{remark}

\Comment{
\iain{Note that a KP bisimulation between $X$ and $Y$ does not have to include a triple $(X,Y,f)$.
Not sure what is intended but a definition of KP bisimulation between $X$ and $Y$ could require minimal and that $(X,Y,f)$ belongs to it for some $f$.}\irek{The definition has been changed following Iain's suggestions. I have also added ``reachable'' because otherwise 
$\orig{X}$ may not be defined. I am not sure we require the ``minimal'' requirement.
}
\clem{The counter-example in the previous definition, as Iain pointed out, that you could have \(a[m].b[n].c, b[m].a[n].c, \{a[m] \mapsto b[m], b[n] \mapsto a[n]\}\) \enquote{sneak in} a KP bisim, and conclude that \(a[m].b[n].c\) and \(b[m].a[n].c\) are KP, while they should not. But requiring the function to be label-preserving seem to fix that, so we may be able to do without this requirement indeed.}
\irek{
 Could we please discuss and agree this definition before moving to DP?}\clem{Yes of course, less editing.}
 \iain{Perhaps it is best to omit `minimal' in the definition of KP bisimulation.}
}

The bijection above %
is constructed in a step by step fashion starting from a pair of standard processes and the empty mapping, ensuring that it preserves the labels and the order on keys between the matched events. Such a construction provides \emph{KP-grounded triples} (\autoref{def:grounded-triples}). %
Note that only action labels and keys are used to define KP bisimulation; the proof part of labels will be used to formulate the next bisimulation.

\begin{definition}[KP-grounded triples]
	\label{def:grounded-triples}
	Let $\Rel{r}[KP]$ be a KP bisimulation between standard $X$ and $Y$. A triple \((X', Y', f) \in \Rel{r}[KP]\), is 
	\emph{KP-grounded (for $\Rel{r}[KP]$)} if either 
	\(X'=X\), \(Y'=Y\) and \(f = \emptyset\), or if there exists a KP-grounded triple \((X'', Y'', f')\) such that  
	\((X', Y', f)\) was obtained from it using either (\ref{cond1-kp}) or (\ref{cond2-kp}) from \autoref{def:HPbisim}.
\end{definition}

Note that given a KP bisimulation between standard $X, Y$, any triple $(X', Y', f)$ derived from $(X,Y,\emptyset)$ using \autoref{def:HPbisim} is KP-grounded for that KP bisimulation, something that will prove useful when studying the corresponding
notion of DP-grounded triples below.

\begin{example}
Consider $a$ and $a+a$.
Although their transitions have different proof keyed labels, \(\Rel{r}[KP]\) only matches their labels  
(in both cases $a$). As there are no causal dependencies in the executed processes, %
there is an empty order on keys. %
There are multiple minimal KP bisimulations for $a$ and $a+a$, where label-preserving $f$s are omitted:
for each $m,n,h,k \in \keyset$ we have
$\{(a,a+a),(a[m],a[n]+a),(a[h],a+a[k])\}$. 
We also note that union of these KP bisimulations (with $f$s omitted), namely $\{(a,a+a)\} \union \{(a[m],a[n]+a):m,n \in \keyset\} \union \{(a[h],a+a[k]):h,k \in\keyset\}$ is also a KP bisimulation for $a$ and $a+a$. Overall, $a\Rel{b}[KP] a+a$. 

\end{example}

\begin{example}
	\label{ex:SB}
Consider $P=a\mid a$ and $Q=a\Par a + a.a$. Although transitions of \(P\), \(Q\) have matching labels and are step bisimilar~\cite[Def. 7.3]{vGG01}, they have different causal behaviour. When $Q$ executes to $a\Par a + a[l'].a[k']$, for some
keys $l', k'$, it has two causally ordered $a$ events with the ordering $l'<k'$. Executing $P$ to $a[l']\mid a[k']$ we get 
label-preserving bijections between the events of thus executed $P$ and $Q$: $\{(a[l], a[l']), (a[k], a[k'])\}$ for any
$l, l', k, k'$. However, none of such bijections is order preserving because $a[l]$ and $a[k]$ are 
core independent whereas $a[l']$ and $a[k']$ are causally ordered (or have ordered keys).
Hence, $P\not  \Rel{b}[KP] Q$. 
\end{example}

\begin{example}
	\label{ex:not-kp}
	Processes $X = a[m].b[n].c$ and $Y = b[m].a[n].c$ are not KP bisimilar because, although there is an order-preserving mapping $f$ with $f(a[m]) = b[m]$, $f(b[n]) = a[n]$, the mapping is not label preserving. Also, $X$ and $Y'= b[n].a[m].c$ are not KP bisimilar since although there is a label-preserving $g(a[m]) = a[n]$, $g(b[n]) = b[m]$, it is not order preserving.  
\Comment{	
	the events with the same keys together is label- and order-preserving \irek{I don't think it is label preserving. Will do this later.}, any KP relation containing \(X\) and \(Y\) will have to construct the mapping starting from \(\orig{X}\) and \(\orig{Y}\), and will not be able to construct a label-preserving mapping between the events corresponding to \(a[m]\) in \(X\) and to \(b[m]\) in \(Y\).
	}
\end{example}
\begin{definition}
	\label{def:maximal-event}
	Given a process \(X\), an event \(e \in \ev{X}\) is \emph{maximal},  if \(\forall e' \in \ev{X}, \neg (e < e')\).
	We write \(\Max{\ev{X}}\) for the set of maximal events in \(\ev{X}\).
\end{definition}

\begin{definition}[DP bisimulation (inspired by \protect{\cite[Sect. 4.3]{bisim-applied-pi}})]
	\label{def:Dsim}	
	Let \(X\) and \(Y\) be standard \pccsk processes.
	A relation $\Rel{r}[DP] \subseteq \kprocset \times \kprocset \times (\events \to \events)$ is a \emph{Dependence-Preserving (DP) bisimulation between \(X\) and \(Y\)} if 
	 \((X, Y, f) \in \Rel{r}[DP]\), %

	 and if whenever $(X', Y', f') \in \Rel{r}[DP]$, then \(f'\) is a label-preserving bijection between \(\ev{X'}\) and \(\ev{Y'}\) and
	\begin{align*}
		\forall t: X' \pr{f}[\theta] X'' \implies & \exists t': Y' \pr{f}[\theta'] Y'' \text{ and } \forall e \in \Max{\ev{X'}},  e \sdep [t] \iff f'(e) \sdep [t']\\
		& \text{ and } (X'', Y'', f' \cup \{[t] \mapsto [t']\})%
		\in  \Rel{r}[DP]\\
 		\forall t': Y' \pr{f}[\theta'] Y'' \implies & \exists t: X' \pr{f}[\theta] X'' \text{ and } \forall e \in \Max{\ev{Y'}}, e \sdep [t'] \iff f'^{-1}(e) \sdep [t] \\
	     & \text{ and } (X'', Y'', f' \cup \{[t] \mapsto [t']\})%
	     \in \Rel{r}[DP]
	\end{align*}
	
	Given any \pccsk processes $X$ and $Y$, if there exists a DP bisimulation $\Rel{r}[DP]$ between $\orig{X}$ and $\orig{Y}$ and $(X,Y, f)\in\Rel{r}[DP]$ for some $f$, 
	then we say that \(X\) and \(Y\) are DP bisimilar, written $X \Rel{b}[DP] Y$.
\end{definition}
\Comment{
\iain{Consider $a[m].b[n].c[k]+b.a.c$ and $a.b.c+b[n].a[m].c[k]$.
These are not KP bisimilar, though the origin processes are.
However we can use the identity mapping on events to deduce they are DP bisimilar.
Requiring that maximal events map to maximal events does not help here.}

\irek{Very good example. Given $X,Y$ we should start constructing a bisimulation from $(\orig{X}, \orig{Y},\emptyset)$ using the conditions on lines 441 to 446. Without requiring that $(X,Y,f)$ belongs to $\Rel{r}[DP]$ to start with. Once the bisimulation is generated 
and $(X,Y,f)$ belongs to it, then $X \Rel{b}[DP] Y$. So, instead of 447-448, we should have
 
If there exists a DP bisimulation $\Rel{r}[DP]$ between $X$ and $Y$ such that $(X,Y, f)\in\Rel{r}[DP]$ for some label preserving $f$, then we say that \(X\) and \(Y\) are DP bisimilar, written $X \Rel{b}[DP] Y$.

With this change, in Iain's example a DP bisimulation for the two processes would not contain the tuple with the processes. To be clear: would this definition work better:

Let \(X\) and \(Y\) be \pccsk processes.
	A relation $\Rel{r}[DP] \subseteq \kprocset \times \kprocset \times (\events \to \events)$ is a \emph{dependence preserving (DP) bisimulation between \(X\) and \(Y\)} if 

	 $(\orig{X}, \orig{Y}, \emptyset) \in  \Rel{r}[DP]$, and if whenever $(X', Y', f') \in \Rel{r}[DP]$, then \(f'\) is a label preserving bijection between \(\ev{X'}\) and \(\ev{Y'}\) and
	\begin{align*}
		\forall t: X' \pr{f}[\theta] X'' \implies & \exists t': Y' \pr{f}[\theta'] Y'' \text{ and } \forall e \in \Max{\ev{X'}},  e \sdep [t] \iff f'(e) \sdep [t']\\
		& \text{ and } (X'', Y'', f' \cup \{[t] \mapsto [t']\})%
		\in  \Rel{r}[DP]\\
 		\forall t': Y' \pr{f}[\theta'] Y'' \implies & \exists t: X' \pr{f}[\theta] X'' \text{ and } \forall e \in \Max{\ev{Y'}}, e \sdep [t'] \iff f'^{-1}(e) \sdep [t] \\
	     & \text{ and } (X'', Y'', f' \cup \{[t] \mapsto [t']\})%
	     \in \Rel{r}[DP]
	\end{align*}
	
	If there exists a DP bisimulation $\Rel{r}[DP]$ between $X$ and $Y$ and $(X,Y, f)\in\Rel{r}[DP]$ for some label preserving $f$, 
	then we say that \(X\) and \(Y\) are DP bisimilar, written $X \Rel{b}[DP] Y$.

}
\iain{I don't think that change to the definition makes any difference.  Whatever DP bisimulation one takes between $a.b.c+b.a.c$ and $a.b.c+b.a.c$, one can always add $(a[m].b[n].c[k]+b.a.c, a.b.c+b[n].a[m].c[k], f)$ to get them bisimilar.}
}
\Comment{
\iain{Unlike for KP bisimulation, no requirement that $(X,Y,f)$ belongs to the relation (for some $f$).
Also, no requirement that $f$ maps maximal events to maximal events, unless that can be deduced.}
	
\iain{Consider $X = a[m].b[n].c$ and $Y = b[n].a[m].c$ with $f$ the identity mapping on events (not order-preserving).} \clem{The definition may not be clear, but the point is that DP and KP have to start from the origin process, so the mapping you are describing cannot be obtained since we first need to associate \(a[m]\) and \(b[n]\). We can't "twist" the function like you did.}
\iain{
Then we can extend to $a[m].b[n].c[k]$ and $b[n].a[m].c[k]$ since $b[n] \sdep c[k]$ and $a[m] \sdep c[k]$.
Does this mean that $X$ and $Y$ are DP bisimilar?
They seem not to be KP bisimilar, since $f$ does not preserve the ordering on keys.

Speculation: $\Max{\ev{X}}$ is relying on an ordering (axiomatic or keys?).  Could we reduce reliance on orderings by adding $M_X,M_Y$ as parameters for the maximal sets.  Then $f$ would be a bijection from $M_X$ to $M_Y$.  If the definition was inductive, rather than co-inductive, one might be able to build up an order-preserving bijection by forward computation just using $\sdep$.}
}

In contrast to \(\Rel{r}[KP]\), the function $f$ above %
is constructed using information about dependence between
proof keyed labels. %
When the matching transitions \(t\), \(t'\) are performed by \(X\), \(Y\), respectively, where
$(X, Y, f) \in \Rel{r}[DP]$, we require that, for all maximal events $e, f(e)$ in \(X'\), \(Y'\), respectively,
either both $e, [t]$ and
$f(e), [t']$ have dependent proof keyed labels or, by complementarity,  both have independent proof keyed labels.

We define \emph{DP-grounded triples} in the corresponding way to that in \autoref{def:grounded-triples} but using conditions of
\autoref{def:Dsim} instead. As a result a mapping in a DP-grounded triple maps a maximal event to a maximal event (as proven in \autoref{app:bisimulation} with \autoref{lem:max-events}).
Note that any triple obtained from standard processes is DP-grounded, and that this allows to discard degenerate cases such as the one illustrated with \autoref{ex:DP-not-KP} below.
  
\begin{example}
Consider $P$, $Q$ from \autoref{ex:SB} and transitions $Q\pr{f}[\lplusr a[k]] Q' \pr{f}[\lplusr a[l]] Q''$ for some
$Q',Q''$. The matching transitions from $P$ are $P\pr{f}[\lmidl  a[k']] P' \pr{f}[\lmidr a[l']] P''$, respectively, 
for some $P',P''$.
Although the transitions of $Q$ have dependent labels ($\lmidl  a[k] \sdep \lmidr a[l]$), the labels of the corresponding
transitions from $P$ are independent: $ \lmidl  a[k'] \ind \lmidr a[l']$.
Hence, $P\not  \Rel{b}[DP] Q$. 
\end{example}

\begin{example}
	Consider $X, Y'$ from \autoref{ex:not-kp} for which we have a label preserving mapping $g$, which can be extended to include
	$g(c[k])=c[k]$ when actions $c$ are performed. Moreover, transitions with labels $b[n]$ and  $c[k]$ of $X$ have dependent labels, and correspondingly for transitions with labels $a[m]$ and $c[k]$ of $Y'$. However, $X$ and $Y'$ are not DP bisimilar
	since we cannot construct a DP bisimulation relation from $\orig{X}$ and $\orig{Y'}$, which are standard,  because their initial actions do not match.
\end{example}

\begin{example}
	\label{ex:DP-not-KP}
Consider non-standard processes $R=a[m].b[n].c[k]+b.a.c$ and $S=a.b.c+b[n].a[m].c[k]$ which cannot compute forwards.
These are not KP bisimilar because, although there is a label-preserving bijection between their events, the 
bijection is not order preserving. However, the origin processes of $R, S$ are equal (to $a.b.c +b.a.c$), so are KP bisimilar.
As for DP bisimulation, using the label-preserving mapping on events, we deduce that $R, S$ are DP bisimilar. We note that such a mapping cannot be constructed starting from the origin processes of $R, S$ using the conditions of \autoref{def:Dsim}
(since $a[m]$ does not match $b[n]$), and hence is not part of any DP-grounded triple.
\end{example}

\autoref{ex:DP-not-KP} shows that KP and DP bisimilarity differ on non-standard processes, where bisimilaity triples are not required to be DP-grounded.
Nevertheless, they coincide on standard \pccsk processes, which produce only KP- and DP-grounded triples.

\begin{restatable}{thm}{thmbisimkpdp}\label{thm:bisim}
	Let \(P\), \(Q\) be any standard \pccsk processes. Then $P \Rel{b}[KP] Q \iff P \Rel{b}[DP] Q$.
\end{restatable}

\begin{example} Consider the Absorption Law (AL)~\cite{Bed91,BC87,vGG01} given below.
\[(a\Par  (b + c)) + (a\Par b) + ((a + c) \Par b) = (a\Par  (b + c)) + ((a + c) \Par b)\]
The right hand side process is a subprocess of the left hand side process, so we only check if the transitions 
of $a\Par b$ can be matched in the right hand side process. 
The transition $a$ needs to be matched by $a$ in $(a + c) \Par b$.
The $b$ transition must be matched by $b$ in $a\Par  (b + c)$. %
Hence, all the transitions match, hence AL holds for strong bisimulation. 
Moreover, we note that there are no composed transitions with dependent labels, so all intermediate maximal events are followed 
by (if any) independent events. Hence, AL is  valid for KP and DP bisimulations. 
\end{example}

Both KP and DP bisimulations are strictly coarser than Hereditary History Preserving (HHP) bisimulation~\cite{Bed91,JNW96,NC95}
and FR bisimulation~\cite{PU07}, defined below on \pccsk, which coincide~\cite{PU07a}.

\begin{definition}[FR bisimulation \protect{~\cite[Def. 5.1]{PU07}}]
	\label{def:FR-bisim}
	Let \(X\) and \(Y\) be \pccsk processes. A relation $\Rel{r}[FR] \subseteq \kprocset \times \kprocset$ is a \emph{Forward-Reverse (FR) bisimulation} between $X$ and $Y$ if $X \Rel{r}[FR] Y$, and whenever $X' \Rel{r}[FR] Y'$, then
	\begin{align*}
	\forall X'', X' \pr{f}[\theta] X'' & \Rightarrow \exists Y'', Y' \pr{f}[\theta'] Y'', \labl{\theta} = \labl{\theta'}, \kay{\theta} = \kay{\theta'},\text{ and } X'' \Rel{r}[FR] Y'' \\
	\forall Y'', Y' \pr{f}[\theta] Y'' & \Rightarrow  \exists X'', X' \pr{f}[\theta'] X'', \labl{\theta} = \labl{\theta'}, \kay{\theta} = \kay{\theta'},\text{ and } X'' \Rel{r}[FR] Y'' \\
	\forall X'', X' \pr{b}[\theta] X'' & \Rightarrow  \exists Y'', Y' \pr{b}[\theta'] Y'', \labl{\theta} = \labl{\theta'}, \kay{\theta} = \kay{\theta'},\text{ and } X'' \Rel{r}[FR] Y'' \\
	\forall Y'', Y' \pr{b}[\theta] Y'' & \Rightarrow \exists X'', X' \pr{b}[\theta'] X'', \labl{\theta} = \labl{\theta'}, \kay{\theta} = \kay{\theta'},\text{ and } X'' \Rel{r}[FR] Y''
\end{align*}
	If there exists an FR bisimulation between $X$ and $Y$, we say that \emph{$X$ and $Y$ are FR bisimilar}, written $X \Rel{b}[FR] Y$.
\end{definition}

AL does not hold for FR bisimulation and HHP bisimulation.
Overall, FR %
bisimulation is  strictly finer than KP and DP bisimulations.

\begin{restatable}{proposition}{propfrimplieskp}
	\label{prop:frimplieskp}
Let \(P, Q\) be any standard \pccsk processes. Then $P \Rel{b}[FR] Q \Rightarrow P \Rel{b}[KP] Q$.
\end{restatable}

	\section{Conclusion}
\label{sec:conclusion}
Our work casts an interesting light not only on the complementarity of dependence and independence, but also on the benefits of the axiomatic approach.
In addition, since KP bisimulations are defined using only keys, they provide a lightweight characterisation of HP for \ccs processes that could trigger in turn efficient algorithms~\cite{Froschle2005}.
Last but not least, our bisimulations generate interesting question about the role of keys in reversibility.
The way they permanently store information allows us to retrieve the causal structure---an element that is lost when constructing a DP bisimulation.
In particular, the status of non-standard processes is of interest: while \(a[n]\) and \(a[m]\) are not FR bisimilar (since their keys are different when backtracking), they are KP bisimilar as a mapping preserving the order on keys can be constructed by the bisimulation game starting from their respective origins.
We aim to leverage this to construct an alternative definition of HHP bisimulation, benefiting from our key- or dependence-based techniques, and compare it to HHP bisimulation in~\cite{NC95,JurdzinskiN00}.

	\bibliography{bib/axrev.bib}
	\endgroup
	\clearpage %
	\begingroup
	\let\clearpage\relax
	\appendix 
	\section{Section~\ref{sec:proved-lts}: Bijection between \texorpdfstring{\pccsk}{PCCSK} and \texorpdfstring{\ccsk}{CCSK}}
\label{app:sub-systems}

This section defines \ccsk~\cite{PU07} %
and proves that \pccsk and \ccsk's transitions are in bijection (\autoref{def:ccsk-pccsk-bijection}).

\begin{definition}[LTS for \ccsk]
	\label{def:ccsk}
		The set for processes for \ccsk %
	is 
	$\kprocset$,%
	and the set of labels is $\kplabelset$ with the proved part removed, namely $\klabelset$.
	The \emph{forward transition relation for \ccsk}, \(\r{f}[\alpha][k]\), is given in \autoref{fig:ltsrulesccskfw}---with \(\keysop\) and \(\stdop\) as in \autoref{def:std}.
	The \emph{backward transition relation for \ccsk}, \(\r{b}[\alpha][k]\), is defined as the symmetric of \(\r{f}[\alpha][k]\)~\cite[Figure 2]{PU07, LaneseP21}.
	The \emph{combined transition relation for \ccsk}, written as \(\r{fb}[\alpha][k]\), is defined as the union of \(\r{f}[\alpha][k]\) and \(\r{b}[\alpha][k]\).
\end{definition}

\Comment{
	\begin{definition}[LTS for \ccsk with proof labels~\protect{\cite{Aub22,aubert2023c}}]
		The \emph{labelled transition system (LTS) for \ccsk with proof labels}, denoted by \pccsk, is 
		$(\kprocset, \kplabelset, \pr{fb}[\theta])$ where $\pr{fb}[\theta]$ is the union of transition relations generated by 
		the forward and backward rules given in \autoref{fig:provedltsrulesccskfw}.
		As usual, we let \(\pr{fb}^*\) be the reflexive transitive closure of \(\pr{fb}\), and similarly for other LTSes.
	\end{definition}
}

\begin{figure}
	\begin{tcolorbox}[title = {Action, Prefix and Restriction}]
		\vspace{.2em}
		\begin{prooftree}
			\hypo{}
			\infer[left label={\(\std{X}\)}]1[act]{\alpha. X \r{f}[\alpha][k]  \alpha[k].X}
		\end{prooftree}
		\hfill 
		\begin{prooftree}
			\hypo{X \r{f}[\beta][k] X'}
			\infer[left label={\(k \neq k'\)}]1[pre]{\alpha[k']. X \r{f}[\beta][k] \alpha[k'].X'}
		\end{prooftree}
		\hfill 
		\begin{prooftree}
			\hypo{ X   \r{f}[\alpha][k]X '}
			\infer[left label={\(\alpha \notin \{\lambda, \out{\lambda}\}\)}]1[res]{X  \bs \lambda  \r{f}[\alpha][k]X ' \bs \lambda}
		\end{prooftree}
	\end{tcolorbox}
	
	\begin{tcolorbox}[adjusted title=Parallel]
		\vspace{.2em}
		\begin{prooftree}
			\hypo{X \r{f}[\alpha][k]X'}
			\infer[left label={\(k \notin \keys{Y}\)}]
			1[\(\lmidl\)]{X \Par   Y \r{f}[\alpha][k] X' \Par  Y}
		\end{prooftree}
		\hfill 
		\begin{prooftree}
			\hypo{Y \r{f}[\alpha][k]Y'}
			\infer[left label={\(k \notin \keys{X}\)}]
			1[\(\lmidr\)]{X \Par   Y  \r{f}[\alpha][k]X \Par  Y'}
		\end{prooftree}
		\hfill 
		\begin{prooftree}
			\hypo{X \r{f}[\lambda][k]  X'}
			\hypo{Y \r{f}[\out{\lambda}][k]  Y'}
			\infer%
			2[syn]{X \Par  Y \r{f}[\tau][k]  X' \Par  Y'}
		\end{prooftree}
	\end{tcolorbox}
	\begin{tcolorbox}[adjusted title=Sum]
		\vspace{.2em}
		\makebox[.4\textwidth][c]{
			\begin{prooftree}
				\hypo{X \r{f}[\alpha][k] X'}
				\infer[left label={\(\std{Y}\)}]1[\( \lplusl \)]{X + Y \r{f}[\alpha][k] X' + Y}
			\end{prooftree}
		}
		\makebox[.4\textwidth][c]{
			\begin{prooftree}
				\hypo{Y \r{f}[\alpha][k] Y'}
				\infer[left label={\(\std{X}\)}]1[\( \lplusr \)]{X + Y \r{f}[\alpha][k] X + Y'}
			\end{prooftree}
		}
	\end{tcolorbox}
	\caption{Forward transition rules for \ccsk. Backward rules are the symmetric versions of the forwards rules, thus are omitted.}
	\label{fig:ltsrulesccskfw}
\end{figure}

\Comment{
The LTSes generated by this definition coincide with the definitions in the literature, \eg for \ccsk~\cite[Figs.~2 and 4]{PU06}, \pccs~\cite[p.~440]{BC88} and \ccs without relabeling, identifier nor conditional~\cite[pp.~69--71]{milner80lncs}. %
We extend the notions of transitions, paths, standardness, etc., to \pccs, \ccsk and \ccs when applicable.
Note that the sets of labels \(\klabelset\) (\resp \(\labelset\), \(\plabelset\)) and of processes \(\kprocset\) (\resp \(\procset\)) for \ccsk (\resp \ccs, \pccs) were already introduced.
}

\begin{restatable}{lemma}{lemderivationuniquenesspccsk}
	\label{lem:derivation-uniqueness}
	For any transition \(t\)%
	, there exists exactly one derivation whose conclusion is $t$  in \pccsk (\resp \ccsk).%
\end{restatable}

\begin{proof}
	By induction on the label for \pccsk: %
	all the information is stored in the proof (keyed) labels, except for the application of the res\ and pre\ rules (or their reverse), but this information can be read off of the structure of the source of the transition. %
	For \ccsk, it suffices to notice that the source and target of the transition will differ only in the presence or absence of a key, from which the derivation can be uniquely obtained.
\end{proof}

\begin{remark}
	\label{rem:note-on-uniqueness}
	Note that \autoref{lem:derivation-uniqueness} would not hold as stated if our LTSes were using a structural congruence containing \eg \(P \Par Q \equiv Q \Par P\), as the processes \(P\) and \(Q\) could be swapped any even number of times in the derivations.
\end{remark}

\begin{definition}[Bijection between \ccsk and \pccsk]
	\label{def:ccsk-pccsk-bijection}
	We define the \emph{proof forgetful} (\(\base{\cdot}\)) and \emph{proof enrichment} (\(\prov{\cdot}\))  mappings between transitions
	\begin{align*}
		\base{\cdot}&: (X_1 \pr{bf}[\theta] X_2) \mapsto (X_1 \r{bf}[\ekay{\theta}] X_2) \tag{\pccsk to \ccsk} \\
		\prov{\cdot}&: (X_1 \r{bf}[\alpha][k] X_2) \mapsto (X_1 \pr{bf}[\theta] X_2) & \text{\st $\labl{\theta} = \alpha$, $\kay{\theta} = k$} \tag{\ccsk to \pccsk}
	\end{align*}
	as follows:
	
	\begin{description}
		\item[$\base{\cdot}$] is immediate: since the derivation of \(X_1 \pr{bf}[\theta] X_2\) in \pccsk is unique (\autoref{lem:derivation-uniqueness}), we can use in \ccsk the rule carrying the same name to obtain a derivation whose conclusion is the desired corresponding transition.
		\item[$\prov{\cdot}$] is also immediate, and is the inverse of \(\base{\cdot}\).
	\end{description}
\end{definition}

	\section{\autoref{sec:ind-complem}: Conservativity and Complementarity of the Independence and Dependence Relations}
\label{app:cons-compl}

\subsection{Conservativity Over Concurrency of CCS with Proof Labels}
\label{ssec:conservativity}

Our independence relation is inspired by the concurrency relation\footnote{%
	We prefer to avoid the term \enquote{concurrency}, in general reserved for transitions, but in this section both concurrency and independence describe the same notion.
} defined on \pccs in~\cite[Section 3]{BC94}.
This section reminds of \pccs and of its concurrency relation\footnote{The only differences with the aforementioned paper are that it records restrictions in the label, and admits a fixed point operator, but that does not impact our development here.}, defines a mapping between \pccsk and \pccs, states and proves the conservativity result over proof labels (\autoref{lem:conserv-labels}) and over transitions (Corollary~\ref{lem:conservativity} and \ref{lem:extension}).

\begin{definition}[Proved LTS for \ccs]
	\label{def:pccs}
	The \emph{proved LTS for \ccs} (\pccs) is $(\procset, \plabelset, \pr{c}[\theta])$ where $\pr{c}[\theta]$
	is given in \autoref{fig:ltsrules}. 
\end{definition}
Note that labels in \autoref{fig:ltsrules} are proof labels without keys, over which we quantify using \(\theta\) also--it will be clear from context if labels are keyed or not.
\begin{figure}
	\begin{tcolorbox}[title = Action and Restriction]
		\vspace{.2em}
		\makebox[.4\textwidth][c]{
			\begin{prooftree}
				\hypo{}
				\infer[]1[act]{\alpha. P \pr{c}[\alpha]  P}
		\end{prooftree}}
		\makebox[.4\textwidth][c]{
			\begin{prooftree}
				\hypo{P \pr{c}[\theta]P'}
				\infer[left label={\(\labl{\theta} \notin \{\lambda, \out{\lambda}\}\)}]1[res]{P \bs \lambda \pr{c}[\theta]P ' \bs \lambda}
			\end{prooftree}
		}
	\end{tcolorbox}
	
	\begin{tcolorbox}[adjusted title=Parallel Group]
		\makebox[.4\textwidth][c]{
			\begin{prooftree}
				\hypo{P \pr{c}[\theta]P'}
				\infer%
				1[\(\lmidl\)]{P \mid  Q \pr{c}[ \lmidl  \theta] P' \mid Q}
			\end{prooftree}
		}
		\makebox[.4\textwidth][c]{
			\begin{prooftree}
				\hypo{Q \pr{c}[\theta]Q'}
				\infer%
				1[\(\lmidr\)]{P \mid  Q  \pr{c}[ \lmidr \theta]P \mid Q'}
			\end{prooftree}
		}
		\\[1em]
		\makebox[.4\textwidth][c]{
			\begin{prooftree}
				\hypo{P \pr{c}[\upsilon_{1}\lambda]  P'}
				\hypo{Q \pr{c}[\upsilon_{2}\out{\lambda}]  Q'}
				\infer%
				2[syn]{P \mid Q \pr{c}[\cpair{\upsilon_{1}\lambda}{\upsilon_{2} \out{\lambda}}]  P' \mid Q'}
			\end{prooftree}
		}
	\end{tcolorbox}
	\begin{tcolorbox}[adjusted title=Sum Group]
		\makebox[.4\textwidth][c]{
			\begin{prooftree}
				\hypo{P \pr{c}[\theta] P'}
				\infer1[\( \lplusl \)]{P + Q \pr{c}[ \lplusl \theta] P'}
			\end{prooftree}
		}
		\makebox[.4\textwidth][c]{
			\begin{prooftree}
				\hypo{Q \pr{c}[\theta] Q'}
				\infer1[\( \lplusr \)]{P + Q  \pr{c}[ \lplusr \theta] Q'}
			\end{prooftree}
		}
	\end{tcolorbox}
	\caption{Transition rules with proof labels for \pccs.
	}
	\label{fig:ltsrules}
\end{figure}

\begin{definition}[Concurrency over proof labels{~\cite[p.~257]{BC94}.}]
		The \emph{concurrency relation \(\smile\) over proof labels} is the least symmetric relation that satisfies
		\begin{align}
			\lmidl \theta \smile \lmidr \theta' \span \tag{A1}\\
			\theta \smile \theta' & \Rightarrow \begin{dcases}
				\lmidl \theta \smile \cpair{\theta'}{\theta''}\\
				\lmidr \theta \smile \cpair{\theta''}{\theta'}
			\end{dcases} \tag{A2}\\
			\theta \smile \theta' & \Rightarrow \begin{dcases}
			\lmidd \theta \smile \lmidd \theta'\\
			\lplusd \theta \smile \lplusd \theta'\\
		\end{dcases} \tag{A3}\\
	\theta_{\L} \smile \theta_{\L}'  \text{ and } \theta_{\R} \smile \theta_{\R}' & \Rightarrow \cpair{\theta_{\L}}{\theta_{\R}} \smile  \cpair{\theta_{\L}'}{\theta_{\R}'} \tag{A4}
		\end{align}
\end{definition}

This relation is irreflexive and symmetric~\cite[p.~260]{BC94}, as is our independence relation \(\ind\) (\autoref{rem:ind-irreflexivity}): it suffices to note that S\(^1\) is the mirror version of S\(^2\). 
The relations actually coincide: 

\begin{lemma}[Conservativity over proof labels]
	\label{lem:conserv-labels}
	Let $\theta_1,\theta_2$ be proof labels and $m \neq n$ be keys.
	Then $\theta_1 \smile \theta_2$ iff $\theta_1[m] \ind \theta_2[n]$.
\end{lemma}

\begin{proof}
(\(\Rightarrow\)) The proof is by induction on the length of the derivation of \(\theta_1 \smile \theta_2\).
If it is of length \(1\), then it is A1, and since \(m \neq n\), we have $\theta_1[m] \ind \theta_2[n]$ by P\(^2_k\).
All the other cases amount to mapping A2 to S\(^1\), A3 to P\(^1\) or C\(^1\) depending on the operator considered, and A4 to S\(^3\).
No rule is mapped to S\(^2\), but its presence is needed to obtain closure by symmetry of \(\ind\), which is assumed for \(\smile\).

(\(\Leftarrow\)) Immediate by re-using the previous mapping, since A1 relaxes the condition on key in P\(^2_k\). If S\(^2\) is used, then the symmetric closure of \(\smile\) allows us to conclude.
\end{proof}

As with \(\ind\) we can define when transitions are concurrent (\autoref{def:conc-transitions}), and show that a very tight correspondence between independent \pccsk transitions and concurrent \pccs transitions.
However, this requires first to define mappings between their transitions, which in turn requires the following definitions and lemma:

\begin{definition}[Key removal from proof keyed labels]
	\label{def:key-removal}
	We define  \(\krop: \kplabelset \to \plabelset\) as:
	\begin{align*}
		\kr{\upsilon \alpha[k]} &= \upsilon \alpha & & & \kr{\upsilon \cpair{\upsilon_{1} \lambda[k]}{\upsilon_{2} \out{\lambda}[k]}} & = \upsilon \cpair{\upsilon_{1} \lambda}{\upsilon_{2} \out{\lambda}}
	\end{align*}
\end{definition}

\begin{definition}[Pruning function~%
	\protect{\cite[Def. 5.20]{PU07}}]
	\label{def:pruning-function}
	We define %
	$\prunop: \kprocset \to \procset$ %
	as: %
	\begin{align*}
		\prun{\nil} & = \nil  &&& \prun{X \Par  Y} &= \prun{X} \Par \prun{Y} \\
		\prun{\alpha.X} & = \alpha. \prun{X} &&& \prun{\alpha[k].X} & = \prun{X} \\
		\prun{X \bs \lambda} & = \prun{X}\bs \lambda \\
		\prun{X + Y} & = \span \span \span \span %
		\begin{dcases*}
			\prun{X} & If \(\std{Y}\) but \(\std{X}\) does not hold \\
			\prun{Y} & If \(\std{X}\) but \(\std{Y}\) does not hold\\
			\prun{X} + \prun{Y} & Otherwise \\
		\end{dcases*}
	\end{align*}
\end{definition}

As for \pccsk and \ccsk (\autoref{lem:derivation-uniqueness}), \pccs enjoys unique derivations:

\begin{restatable}{lemma}{lemderivationuniquenesspccs}
	\label{lem:derivation-uniqueness-pccs}
	For any transition \(t\), there exists exactly one derivation whose conclusion is $t$  in \pccs.
\end{restatable}

\begin{proof}
	By induction on the label: all the information is stored in the proof labels, except for the application of the res\ and pre\ rules, but this information can be read off of the structure of the source of the transition.
\end{proof}

\begin{definition}[Mapping between \pccs and \pccsk]
	\label{def:pccs-pccsk-bijection}
	We define the \emph{key forgetful} (\(\kbase{\cdot}\)) and \emph{key enrichment} (\(\kprov{\cdot}\))  mappings %
	\begin{align*}
		\kbase{\cdot} &: (X_1 \pr{f}[\theta] X_2) \mapsto (\prun{X_1} \pr{c}[\kr{\theta}] \prun{X_2})	\tag{\pccsk to \pccs}\\
		\kprov{\cdot} &: (P_1 \pr{c}[\theta] P_2, k) \mapsto (P_1 \pr{f}[\theta][k] X_2)  & \text{\st \(\prun{X_2} = P_2\)} \tag{\pccs to \pccsk}
	\end{align*}
	essentially as in \autoref{def:ccsk-pccsk-bijection}, additionally leveraging \autoref{lem:derivation-uniqueness-pccs}. 
\end{definition}

Note that the key enrichment function has to be given a key as an extra parameter, to avoid having to pick a fresh one.
Deterministic strategies to select keys have been explored~\cite{AubertM21} and could be leveraged to ensure that \(\kbase{\cdot}\) and \(\kprov{\cdot}\) are inverses, but this would complicate our definitions while not improving our results presented below.

\begin{definition}[Concurrent transitions]
	\label{def:conc-transitions}
	Two transitions \(t_0 : P \pr{c}[\theta] P_0\), \(t_1 : P \pr{c}[\theta] Q_1\) are \emph{concurrent}, denoted \(t_0 \smile t_1\), if and only if \(\theta_0 \smile \theta_1\).
\end{definition}

Finally, we have two immediate corollaries of \autoref{lem:conserv-labels}:

\begin{corollary}[$\ind$ extends $\smile$]
	\label{lem:conservativity}
	Given two \pccs transitions \(t_0\) and \(t_1\) and two keys \(m \neq n\), \(t_0 \smile t_1 \implies \kprov{(t_0, m)} \ind \kprov{(t_1, n)}\).
\end{corollary}

\Comment{
\irek{An alternative to \autoref{lem:conservativity} might be the following, which does not require $\ke$ maps:
\begin{lemma}[?]\label{lem:converse}
	\label{lem:?}
	Given two \pccsk transitions \(t_0, t_1\) with $\kay{t_0}\neq \kay{t_1}$, 
	\(\kf{t_0} \smile \kf{t_1} \implies t_0 \ind t_1\).
\end{lemma}
}

\clem{The quantification is over \pccsk transitions, which weakens a bit the result in my opinion. But I like the idea, so either way is fine by me.}
	
\iain{\autoref{lem:converse} is of course the converse of \autoref{lem:extension}.  Probably useful in addition to \autoref{lem:conservativity}.  I woult expect the converse of \autoref{lem:conservativity} also to hold; this should follow from \autoref{lem:extension} using $\forall m.\kf{\ke{t}^m} = t$.}
}

\begin{corollary}[$\ind$ is conservative over $\smile$]
	\label{lem:extension}
		Given two \pccsk transitions \(t_0\) and \(t_1\), \(t_0 \ind t_1 \implies \kbase{(t_0)} \smile \kbase{(t_1)}\).
\end{corollary}

\subsection{Proving the Complementarity of Dependence and Independence}%
\label{ssec:complementarity}

Proving \autoref{prop:connectednessadequacy} and \autoref{thm:complementarity} requires intermediate definitions and results:

\begin{definition}[Realisation]
	A process \(X\) \emph{realises the proof label \(\theta\)} if there exists \(X_1\) and \(X_2\) such that \(X \pr{fb}^* X_1 \pr{fb}[\theta] X_2\).
\end{definition}

\begin{proposition}
	\label{prop:realis}
	For every proof label \(\theta\), there exists a process that realises it, and we denote it \(\real{\theta}\).
\end{proposition}

\begin{proof}
	We prove it by induction on the size of \(\theta\):
	\begin{description}
		\item[\protect{\(\theta = \alpha[k]\)}] Then \(\alpha.\nil\) realises \(\theta\).
		\item[\(\theta = \lplusd \theta'\)] By induction hypothesis, \(\real{\theta'}\) realises \(\theta'\), and \(\nil + \real{\theta'}\) or \(\real{\theta'} + \nil\), depending on the value of \(\D\), will realise \(\theta\).
		\item[\(\theta = \lmidd \theta'\)] By induction hypothesis, \(\real{\theta'}\) realises \(\theta'\), and \(\nil \mid \real{\theta'}\) or \(\real{\theta'} \mid \nil\), depending on the value of \(\D\), will realise \(\theta\).
		\item[\(\theta  = \cpair{\theta_1}{\theta_2}\)] By induction hypothesis, \(\real{\theta_1}\) (\resp \(\real{\theta_2}\)) realises \(\theta_1\) (\resp \(\theta_2\)), so \(\real{\theta_1} \mid \real{\theta_2}\) realises \(\theta\). \qedhere
	\end{description}
\end{proof}

\begin{lemma}
	\label{lem:path-origin}
	For all reachable processes \(X\) and \(Y\), there exists a path \(X \pr{fb}^*Y \) iff \(\orig{X} = \orig{Y}\).
\end{lemma}

\begin{proof}
	\begin{description}
		\item[\(\Rightarrow\)] Informally, the key argument is that \(X\) and \(Y\) will diverge only in the name, presence or absence of keys, and that erasing them (\eg using the \(\tostdop\) function~\cite[p.~128]{LaneseP21}) will give the same standard process, which will be the origin of both.
		\item[\(\Leftarrow\)] It suffices to consider the path \(X \pr{fb}^* \orig{X} = \orig{Y} \pr{fb}^* Y\) which exists by definition of connectedness and \autoref{lem:loop_proved} applied to paths.\qedhere
	\end{description}
\end{proof}

From this lemma, it is easy to deduce the following:

\begin{corollary}
	\label{cor:connected-origin}
	If \(t_1 : X_1 \pr{fb}[\theta_1] X'_1\) and \(t_2 : X_2 \pr{fb}[\theta_2] X'_2\) are connected,  then \(\orig{X_1} = \orig{X_2}\).
\end{corollary}

\begin{proposition}
	\label{derivation-connectedness-unique}
	For all \(\theta_1\), \(\theta_2\), if a derivation of \(\theta_1 \conn \theta_2\) exists, then it is unique. 
\end{proposition}

\begin{proof}
	This follows easily by induction on the structure of \(\theta_1\) and \(\theta_2\), since no two conclusions overlap in the definition of the connectivity relation (\autoref{fig:pccsk-relations}).
\end{proof}

\propconnectednessadequacy*

\begin{proof}%
	We prove each equation separately, letting \(\D\), \(\E\) range over the \emph{directions} \(\L\)(eft) and \(\R\)(ight).
	\begin{description}
		\item[(\ref{prop:connectednessadequacy-1})]
		Since \(t_1\) and \(t_2\) are connected, we have by \autoref{cor:connected-origin} that \(\orig{X_1} = \orig{X_2} = P\),  and we reason by induction on \(P\):
			
		\begin{description}
			\item[$P = \nil$ ] Then no transitions are possible and this case is vacuously true.
			\item[$P = \alpha . Q$] Then there are three cases:
			\begin{itemize}
				\item If $X_1=P$ and  $\theta_1$ is of the form $\alpha[k]$, then by A\(^1\) we have $\theta_1 \conn \theta_2$.
				\item If $X_2=P$ and  $\theta_2$ is of the form $\alpha[k]$ and $\theta_1$ is not of the form $\beta[k']$ for some $\beta$ and \(k'\), then by A\(^2\) we have $\theta_1 \conn \theta_2$.
				\item Otherwise, there exist paths \(Q \pr{bf}^* X_1 \pr{bf}[\theta_1] X_1'\) and \(Q \pr{bf}^* X_2 \pr{bf}[\theta_2] X_2'\) and we use the induction hypothesis on \(Q\) to obtain the desired result.
			\end{itemize}
			\item[$P = Q \bs \lambda$] Then, there exist paths \(Q \pr{bf}^* X_1 \pr{bf}[\theta_1] X_1'\) and \(Q \pr{bf}^* X_2 \pr{bf}[\theta_2] X_2'\) and we use the induction hypothesis on \(Q\) to obtain the desired result.
			\item[$P = Q_{\L} + Q_{\R}$] Then it must be the case that \(\theta_1 = \lplusd \theta_1'\) and  \(\theta_2 = \lpluse \theta_2'\), and there are two cases:
			\begin{description}
				\item[\(\E = \D\)]  Then, there exist paths \(Q_{\D} \pr{bf}^* {X_1}_{\D} \pr{bf}[\theta_1'] {X_1'}_{\D}\) and \(Q_{\D} \pr{bf}^* {X_2}_{\D} \pr{bf}[\theta_2'] {X_2'}_{\D}\) for \(X_i = {X_i}_{\L} + {X_i}_{\R}\), \(X_i' = {X_i'}_{\L} + {X_i'}_{\R}\) for \(i \in \{1, 2\}\).
				By induction on \(Q_{\D}\), we have that \(\theta_1' \conn \theta_2'\) and \(\theta_1 = \lplusd \theta_1' \conn \lplusd \theta_2' = \theta_2\) follows by C\(^1\).
				\item[\(\E = \OD\)] Then by C\(^2\) we have that \(\theta_1 = \lplusd \theta_1' \conn \lpluse \theta_2' = \theta_2\).
			\end{description}
			\item[$P = Q_{\L} \Par Q_{\R}$] There are four cases, depending of the structure of \(\theta_1\) and \(\theta_2\):
			\begin{description}
				\item[\(\theta_1 = \lmidd \theta_1'\) and \(\theta_2 = \lmide \theta_2'\)] Then there are two cases, depending on \(\D\) and \(e\):
				\begin{description}
					\item[\(\E = \D\)] Then, there exist paths \(Q_{\D} \pr{bf}^* {X_1}_{\D} \pr{bf}[\theta_1'] {X_1'}_{\D}\) and \(Q_{\D} \pr{bf}^* {X_2}_{\D} \pr{bf}[\theta_2'] {X_2'}_{\D}\) for \(X_i = {X_i}_{\L} \Par {X_i}_{\R}\), \(X_i' = {X_i'}_{\L} \Par {X_i'}_{\R}\) for \(i \in {1, 2}\).
					By induction, we have that \(\theta'_1 \conn \theta'_2\) and by P\(^1\) we obtain the desired result.
					\item[\(\E = \OD\)] Then by P\(^2\) we have that \(\theta_1 = \lmidd \theta_1' \conn \lmide \theta_2' = \theta_2\).				
				\end{description}
				\item[\(\theta_1 = \lmidd \theta_1'\) and \(\theta_2 = \cpair{\theta_{\L}}{\theta_{\R}}\)]
				Then, there exist paths \(Q_{\D} \pr{bf}^* {X_1}_{\D} \pr{bf}[\theta_1'] {X_1'}_{\D}\) and  \(Q_{\D} \pr{bf}^* {X_2}_{\D} \pr{bf}[\theta_{\D}] {X_2'}_{\D}\) for \(X_i = {X_i}_{\L} \Par {X_i}_{\R}\), \(X_i' = {X_i'}_{\L} \Par {X_i'}_{\R}\) for \(i \in {1, 2}\).
				By induction, we have that \(\theta_1' \conn \theta_{\D}\), and by S\(^1\) we obtain the desired result.
				\item[\(\theta_1 = \cpair{\theta_{\L}}{\theta_{\R}}\) and \(\theta_2 = \lmidd \theta_2'\)] This case is nearly identical to the previous one, except that it uses S\(^2\) to obtain the desired result.
				\item[\(\theta_1 = \cpair{{\theta_1}_{\L}}{{\theta_1}_{\R}}\) and \(\theta_2 = \cpair{{\theta_2}_{\L}}{{\theta_2}_{\R}}\)]
				Then, there exist paths 
				\begin{align*}
					Q_{\D} \pr{bf}^* {X_1}_{\D} \pr{bf}[{\theta_1}_{\D}] {X_1'}_{\D} && \text{ and } && Q_{\D} \pr{bf}^* {X_2}_{\D} \pr{bf}[{\theta_2}_{\D}] {X_2'}_{\D}
				\end{align*}
				for \(X_i = {X_i}_{\L} \Par {X_i}_{\R}\), \(X_i' = {X_i'}_{\L} \Par {X_i'}_{\R}\) for \(i \in {1, 2}\).
				Using the induction hypothesis twice gives \({\theta_1}_{\D} \conn {\theta_2}_{\D}\), and we get the desired result by S\(^3\).\qedhere
			\end{description}
		\end{description}
		\item[(\ref{prop:connectednessadequacy-2})] 	
		We prove this by constructing a process \(X\) that realises both proof keyed labels, that is, than can reach a process \(X_1\) capable of performing a transition labelled \(\theta_1\) and a process \(X_2\) capable of performing a transition labelled \(\theta_2\).
		To do so, we leverage \autoref{prop:realis} and will use its \(\real{\theta}\) construction.
		
		We reason by induction on the last rule of the derivation of \(\theta_1 \conn \theta_2\), which we know to be unique by \autoref{derivation-connectedness-unique}.
		
		\begin{description}
			\item[A\(^1\)] Then \(\theta_1 = \alpha[k]\), \(t_1: \alpha.\real{\theta_2}\r{f}[\theta_1] \alpha[k].\real{\theta_2}\) and \(t_2: \alpha[k].\real{\theta_2} \r{f}[\theta_2] Y\), for some \(Y\), are connected since they are composable.
			\item[A\(^2\)] Then \(\theta_2 = \alpha[k]\), \(t_1: \alpha.\real{\theta_1} \r{f}[\theta_2] \alpha[k].\real{\theta_1}\) and \(t_2: \alpha[k].\real{\theta_1} \r{f}[\theta_1] Y\), for some \(Y\), are connected since they are composable.
			\item[C\(^1\)] Then \(\theta_1 = +_{\D} \theta_1'\), \(\theta_2 = +_{\D} \theta_2'\), and by the induction hypothesis there exists \(t_1': X_1' \pr{fb}[\theta_1'] Y_1'\) and \(t_2' : X_2' \pr{fb}[\theta_2'] Y_2'\) that are connected.
			We have that \(\orig{X_1'} = \orig{X_2'} = X'\) by \autoref{cor:connected-origin}, and it is immediate to observe that either \(X' + \nil\) or \(\nil + X'\) (depending on the value of \(\D\)) can realise both \(\theta_1\) and \(\theta_2\), and hence the transitions with labels \(\theta_1\) and \(\theta_2\) are connected.
			\item[C\(^2\)] Letting \(X = \real{\theta_1} + \real{\theta_2}\), it is obvious that \(X\) can perform two coinitial transitions with labels \(\theta_1\) and \(\theta_2\), that are hence connected.
			\item[P\(^1\)] This case is similar to C\(^1\).
			\item[P\(^2\)] This case is similar to C\(^2\).
			\item[S\(^1\)] Then \(\theta_1 = \lmidd \theta_1'\), \(\theta_2 = \cpair{\theta_{\L}}{\theta_{\R}}\), and by the induction hypothesis there exists  \(t_1': X_1' \pr{fb}[\theta_1'] Y_1'\) and \(t_2' : X_2' \pr{fb}[\theta_2'] Y_2'\) that are connected.
			We have that \(\orig{X_1'} = \orig{X_2'} = X'\) by \autoref{cor:connected-origin}, and it is immediate to observe that either \(X' \mid \real{\theta_{\R}}\) or \(\real{\theta_{\L}} \mid X'\) (depending on the value of \(\D\)) can realise both \(\theta_1\) and \(\theta_2\), and hence the transitions with labels \(\theta_1\) and \(\theta_2\) are connected.
			\item[S\(^2\)] This case is nearly identical to the previous one.
			\item[S\(^3\)] Then \(\theta_1 = \cpair{\theta_{\L}^1}{\theta_{\R}^1}\),  \(\theta_2 = \cpair{\theta_{\L}^2}{\theta_{\R}^2}\), and by the induction hypothesis there exists, for \(i \in \{1, 2\}\) and \(\D \in \{\L, \R\}\), four transitions 
			\[
			t_{\D}^i: X_{\D}^i \pr{fb}[\theta_{\D}^i] Y_{\D}^i
			\]
			such that \(t_{\D}^1\) and \(t_{\D}^2\) are connected.
			Hence, we know that \(\orig{X_{\D}^1} = \orig{X_{\D}^2}\) by \autoref{cor:connected-origin} and that \(\orig{X_{\L}^1} \Par \orig{X_{\R}^1}\) can realise both \(\theta_1\) and \(\theta_2\), and hence the transitions with labels \(\theta_1\) and \(\theta_2\) are connected. \qedhere
		\end{description}
	\end{description}
\end{proof}

\begin{proposition}
	\label{lem:ind-alpha}
	If $\theta_1 \ind \theta_2$ then neither $\theta_1$ nor $\theta_2$ is of the form $\alpha[k]$.
\end{proposition}

\begin{proof}
	This follows easily once observed that the only base case for \(\ind\) in \autoref{fig:pccsk-relations}, P\(^2_k\), requires \(\theta\) and \(\theta'\) to be prefixed by \(\lmidl\) and \(\lmidr\), respectively.
\end{proof}

Finally, we have all the elements to prove that the independence and dependence relations partition the connectedness relation:

\thmcomplementarity*

\begin{proof}
	\begin{enumerate}
		\item Any proof of $\theta_1 \ind \theta_2$ can be transformed into a proof of $\theta_1 \conn \theta_2$ by systematically replacing rules one-by-one.
		The only noticeable difference is that the condition on keys in P\(^2_k\) is absent in P\(^2\), but since we are relaxing a condition, this does not go in the way of the proof transformation.
		
		\item Similarly, any proof of \(\theta_1 \sdep \theta_2\) can be transformed into a proof of \(\theta_1 \conn \theta_2\) by systematically replacing rules one-by-one, and relaxing the condition on keys in P\(^2_k\).
		Note that the premises in S\(^3\) involves \(\sdep\) \emph{and} \(\conn\), hence requiring to transform only one of the two derivation sub-trees.
		
		\item We prove this by induction on the length of \emph{the} proof of \(\theta_1 \conn \theta_2\), %
		unique by \autoref{derivation-connectedness-unique}:
		\begin{description}
			\item[Length \(1\)] Then, the proof of \(\theta_1 \conn \theta_2\) is one of the following:
			\begin{description}
				\item[A\(^1\), A\(^2\) or C\(^2\)] Then, we can immediately obtain a proof of \(\theta_1 \sdep \theta_2\) using the same rule.
				By \autoref{lem:ind-alpha}, we know that no proof of \(\alpha[k] \ind \theta\) nor of \(\theta \ind \alpha[k]\) exist, and by inspection of the rules of \(\ind\), we can observe that no proof of \(+_{\D}\theta_1 \ind +_{\OD} \theta_2\) can exist.
				\item[P\(^2\)] Then, there are two cases:
				\begin{description}
					\item[\(\kay{\theta} = \kay{\theta'}\)] In this case,  we can obtain a proof of \(\theta_1 \sdep \theta_2\) using P\(^2_k\), but cannot derive \(\theta_1 \ind \theta_2\) since P\(^2_k\) cannot be used and no other rule has a conclusion of the form \(\lmidd \theta \ind \lmidod \theta'\). 
					\item[\(\kay{\theta_1} \neq \kay{\theta_2}\)]  In this case,  we can obtain a proof of \(\theta_1 \ind \theta_2\) using P\(^2_k\), but cannot derive \(\theta_1 \sdep \theta_2\) since P\(^2_k\) cannot be used and no other rule has a conclusion of the form \(\lmidd \theta \sdep \lmidod \theta'\).
				\end{description}
			\end{description}
			\item[Length \(>1\)]
			Then, the proof of \(\theta_1 \conn \theta_2\) terminates with one of the following rules:
			\begin{description}
				\item[C\(^1\)] Then \(\theta_1 = \lplusd \theta_1'\), \(\theta_2 = \lplusd \theta_2'\), and by the induction hypothesis there exists a proof of \(\theta_1' \sdep \theta_2'\) or of \(\theta_1' \ind \theta_2'\), but not of both.
				In both cases, we obtain the desired result by applying C\(^1\) to the proof obtained by induction, and no proof exists for the other relation since C\(^1\) is the only rule with a conclusion of this shape.
				\item[P\(^1\), S\(^1\) and S\(^2\)]  Those cases are similar to the previous one.
				\item[S\(^3\)] Then, \(\theta_1 = \cpair{\theta_{\L}}{\theta_{\R}}\), \(\theta_2 = \cpair{\theta_{\L}'}{\theta_{\R}'}\), and by induction, we have one of those cases:
				\begin{description}
					\item[\(\theta_{\L} \ind \theta_{\L}'\) and \(\theta_{\R} \ind \theta_{\R}'\) ] Then applying S\(^3\) gives that \(\theta_1 \ind \theta_2\), and no proof of \(\theta_1 \sdep \theta_2\) can exist since the only rule that could be applied to obtain this conclusion is S\(^3\) but its premises cannot be proven by induction.
					\item[\(\theta_{\L} \ind \theta_{\L}'\) and \(\theta_{\R} \sdep \theta_{\R}'\)]
					Since \(\theta_{\L} \ind \theta_{\L}'\), then \(\theta_{\L} \conn \theta_{\L}'\) by (1) of the current proposition, and \(\theta_1 \sdep \theta_2\) can be proven using S\(^3\).
					By induction, no proof of \(\theta_{\R} \ind \theta_{\R}'\) exists, and hence no proof of \(\theta_1 \ind \theta_2\) can exist since the only rule that could be applied to obtain this conclusion is S\(^3\) but its premises cannot be proven.
					\item[\(\theta_{\L} \sdep \theta_{\L}'\) and \(\theta_{\R} \ind \theta_{\R}'\)] This case is identical to the previous one.
					\item[\(\theta_{\L} \sdep \theta_{\L}'\) and \(\theta_{\R} \sdep \theta_{\R}'\)] Then \(\theta_1 \sdep \theta_2\) can be obtained using S\(^3\), since \(\theta_{\L} \sdep \theta_{\L}'\) implies \(\theta_{\L} \conn \theta_{\L}'\) by (2) of the current proposition.
					Since, by induction, neither \(\theta_{\L} \ind \theta_{\L}'\) nor \(\theta_{\R} \ind \theta_{\R}'\) can be proven, it is clear that \(\theta_1 \ind \theta_2\) cannot be proven either.\qedhere
				\end{description}
			\end{description}
		\end{description}
	\end{enumerate}
\end{proof}

\Comment{
\begin{notation}
\label{not:conn-transitions}
We write \(t_1 \conn t_2\) if \(t_1\) and \(t_2\) are connected, and \(t_1 \ind t_2\) (\resp \(t_1 \sdep t_2\)) if \(t_1 \conn t_2\) and \(\lblof{t_1} \ind \lblof{t_2}\) (\resp \(\lblof{t_1} \sdep \lblof{t_2}\)).
\end{notation}

Note that from \autoref{thm:complementarity} we see that if $t_1$ and $t_2$ are connected then either $t_1 \ind t_2$ or $t_1 \sdep t_2$, but not both; this is the reason why \autoref{fig:Revised-1} \enquote{partitions} connected transitions between independent and dependent ones.

\autoref{def:relation-proof-labels} and \autoref{not:conn-transitions} are easily extended to \ccsk's transitions by using the bijection from \autoref{def:sub-systems}:

\begin{definition}[Relations on \ccsk's transitions]
\label{def:relations-on-ccsk}
For all \(t_1\), \(t_2\) in \ccsk, we let
\begin{align*}
t_1 \conn t_2 \Leftrightarrow \prov{t_1} \conn \prov{t_2}  &&
t_1 \ind t_2 \Leftrightarrow \prov{t_1} \ind \prov{t_2} &&
t_1 \sdep t_2 \Leftrightarrow \prov{t_1} \sdep \prov{t_2}
\end{align*}
\end{definition}

\begin{remark}
\label{rem:relations-reflexivity}
It is easy to prove that \(\conn\) and \(\sdep\) are reflexive and symmetric, since A$^1$ is essentially the mirror of A$^2$, %
and that \(\ind\) is irreflexive and symmetric, since S\(^1\) is the mirror of S\(^2\).
\end{remark}

\Comment{
\noindent
\resizebox{.45\textwidth}{!}{\input{figures/transition-relations}}
}
}

	\section{\autoref{ssec:reminders-axiom}: Additional Material and Proofs}
\label{app:axioms}

\subsection{Proof of \autoref{lem:immed pred composable}}
\label{app:immed pred composable}
\immedpredcomp*
\begin{proof}%
($\Rightarrow$)
Suppose that $e_1 \ip e_2$.

We start by showing that $e_1$ is composable with $e_2$.
Let $r$ be a forward-only rooted path ending with a forward transition $t_2 \in e_2$.
Since $e_1 < e_2$ we must have a forward transition $t_1 \in e_1$ occurring before $t_2$ in $r$.
Thus $r = r't_1st_2$ for some $r',s$.
We proceed by induction on $\len s$.

If $\len s = 0$ then $e_1$ is composable with $e_2$ as required.
So suppose $\len s > 0$.
Clearly for any $t$ in $s$ we cannot have $t_1 < t < t_2$.
So transitions $t$ of $s$ fall into two groups:
\begin{enumerate}
\item $t_1 < t \not< t_2$
\item $t_1 \not< t$
\end{enumerate}
If a member $t$ of group 1 is immediately before a member $t'$ of group~2, then since $t_1 \not< t'$ we have $t \not< t'$ using transitivity of $<$.  Using polychotomy (\autoref{prop:poly}) we deduce $t \coind t'$, so that $t \ind t'$ using IRE.  Using RPI and SP, the two transitions can be swapped.
Thus if group 1 is non-empty then the last group 1 transition $t$ can be moved to immediately before $t_2$. Since $t \not< t_2$, the two transitions can be swapped and $\len s$ reduces.

So suppose that group 1 is empty.  But then the first transition $t$ of $s$ is in group 2, so that $t_1 \not< t$.  In this case $t_1$ and $t$ can be swapped, again reducing $\len s$.
Hence $e_1$ is composable with $e_2$.

Now we show not $e_1 \coind e_2$.  Since $e_1 \ip e_2$, \(e_1 < e_2\) and hence $e_1 \coind e_2$ cannot hold by polychotomy.

($\Leftarrow$)
Suppose that $e_1$ is composable with $e_2$ and not $e_1 \coind e_2$.
Let $r = st_1t_2$ be a rooted forward-only path with $t_1 \in e_1$ and $t_2 \in e_2$.
We can use polychotomy to deduce that $e_1 < e_2$.
Suppose forward event $e$ is such that $e_1 < e < e_2$.
Since $\cte(r,e_2) = 1$, we must have $\cte(r,e) = 1$.  But then there is $t \in e$ such that $t$ occurs in $r'$.  This contradicts $e_1 < e$.  Hence $e_1 \ip e_2$.
\end{proof}

\subsection{Proof of \autoref{thm:uniqueness}}
\begin{lemma}[Non-degenerate diamond~\protect{\cite[Lem. 4.7]{LPU24}}]
	\label{lem:non-degenerate}
	If an LTSI is pre-reversible and we have a diamond
	$t:P \r{fb}[\alpha] Q$, $u:P \r{fb}[\beta] R$ with $t \ind u$
	and cofinal transitions $u': Q \r{fb}[\beta] S$, %
	$t': R \r{fb}[\alpha] S$,
	then the diamond is \emph{non-degenerate},
	meaning that $P,Q,R,S$ are distinct processes.
\end{lemma}

\begin{lemma}[\protect{\cite[Prop. 4.10]{LPU24}}]\label{prop:ID}
	If an LTSI satisfies BTI and PCI then it satisfies ID.
\end{lemma}

\begin{proposition}[Uniqueness of coinitial independence]
	\label{prop:prerev-coinitial-unique}
	Suppose that $(\Proc,\Lab,\r{fb},\ind_1)$ and $(\Proc,\Lab,\r{fb},\ind_2)$ are two pre-reversible LTSIs with the same underlying combined LTS. Then $\ind_1$ and $\ind_2$ agree on coinitial transitions. 
\end{proposition}
\begin{proof}
	Suppose that $t_1$ and $t_2$ are coinitial transitions with $t_1 \ind_1 t_2$.
	By SP for $\ind_1$ we get a square.
	This is non-degenerate (all states are distinct) by \autoref{lem:non-degenerate}.
	We deduce that $t_1 \ind_2 t_2$ using ID which holds by \autoref{prop:ID}.
	By symmetry we deduce the result.
\end{proof}

\Comment{
	\begin{definition}\label{def:admit-prever}
		A combined LTS $(\Proc,\Lab,\r{fb})$ \emph{admits pre\hyph reversibility} if it satisfies WF, and an independence relation $\ind$ can be added to form an LTSI $(\Proc,\Lab,\r{fb},\ind)$ satisfying SP, BTI and PCI.
	\end{definition}
}

\uniqueness*

\begin{proof}
	The independence relation $\ind$ is determined for coinitial transitions
	by \autoref{prop:prerev-coinitial-unique}.  %
	From \autoref{def:event-general}, the equivalence on transitions %
	$\eveqt$ %
	is determined by $\ind$ on coinitial transitions.  Furthermore, it is clear from \autoref{def:event relations} that $\coind$, $\leq$
	and $\Cf$ are then uniquely determined.
\end{proof}

	\section{\autoref{sec:properties}: Additional Material and Proofs}
\subsection{Proof of \autoref{thm:axioms-hold}}
\label{app:axioms-hold}

Our goal here is to prove the following theorem:
\axiomshold*

We prove it \enquote{piecewise} below.

\begin{restatable}[\cite{Aub22,aubert2023c}]{proposition}{pccskSPBTIWF}
	\label{prop:pccsk-basic}
	The LTSI of $\pccsk$ satisfies SP, BTI and WF.
\end{restatable}

\begin{proof}
	SP~\cite[Theorem 2]{aubert2023c}, BTI~\cite[Lemma 10]{aubert2023c} and WF~\cite[Lemma 11]{aubert2023c} for \pccsk had already been proved, but a complete proof---adapted to our direct definition of independence---is given in \autoref{ssect:pccsk-basic} %
	for completeness.
\end{proof}
\begin{definition}[LLG~{\cf\cite[Def. 6.11]{LPU24}}]\label{def:LLG}
	An LTSI is \emph{locally label-generated (LLG)} if there is an irreflexive binary relation $I$ on $\Lab$ such that for any 
	transitions $t:P \r{fb}[\alpha] Q$ and $u:R \r{fb}[\beta] S$
	we have $t \ind u$ iff $t,u$ are connected and \(I(\und{\alpha}, \und{\beta})\).
\end{definition}
\begin{proposition}
	\label{prop:LLG}
	If an LTSI is LLG then it satisfies PCI, IRE and RPI.
\end{proposition}

\begin{proof}
	As the proof of \cite[Prop. 6.12]{LPU24}, given that PCI, IRE and RPI preserve connectedness of transitions.
\end{proof}
\begin{proposition}\label{prop:proved LLG}
The LTSI of \pccsk is LLG.
\end{proposition}
\begin{proof}
Immediate from \autoref{def:ind-on-ccsk}.
\end{proof}

\begin{proposition}\label{prop:pccsk IRE}
	The LTSI of $\pccsk$ satisfies PCI, IRE and RPI. %
\end{proposition}

\begin{proof}
	Since the LTSI of \pccsk is LLG by \autoref{prop:pccsk IRE}, it satisfies PCI, IRE and RPI by \autoref{prop:LLG}.
	This was essentially already observed in~\cite[Sect.~6.2]{LPU24}, using a slightly different definition of label-generated.
\end{proof}

To transfer axioms from \pccsk to \ccsk we use the bijection \(\base{\cdot} = {(\prov{\cdot}})^{-1}\) %
 of \autoref{def:ccsk-pccsk-bijection}.

\begin{proposition}\label{prop:ccsk basic}
	The LTSI of $\ccsk$ satisfies SP, BTI and WF.
\end{proposition}

\begin{proof}
	This follows essentially from \autoref{prop:pccsk-basic}:
	\begin{description}
		\item[SP] Let $t$ and $u$ be coinitial \ccsk transitions such that \(t \ind u\).  Then \(\prov{t}\) and \(\prov{u}\) are coinitial transitions in \pccsk with \(\prov{t} \ind \prov{u}\). By SP for \pccsk there exist cofinal transitions \(u'\) and \(t'\), from which we obtain the desired cofinal transitions \(\base{u'}\) and \(\base{t'}\).
		\item[BTI] Given coinitial backward transitions $t$, $t'$, BTI for \pccsk gives $\prov{t} \ind \prov{t'}$, which implies $t \ind t'$.
		\item[WF] The absence of infinite reverse computation follows from the finite number of keys in processes. \qedhere
	\end{description}
\end{proof}

The event equivalences obtained when instantiating \autoref{def:event-general} to \pccsk and \ccsk are preserved by their isomorphism.%

\begin{proposition}
\label{prop:eveqt-proved-base}
For all \(t, u\) in \ccsk (\resp \pccsk)
\begin{align*}
	t \eveqt u \implies \prov{t} \eveqt \prov{u} && \text{(\resp~ } t \eveqt u \implies \base{t} \eveqt \base{u}\text{)}
\end{align*}
\end{proposition}

\begin{proof}%
Being on opposite sides of a diamond as in \autoref{def:event-general} is trivially preserved by \(\prov{\cdot}\) and \(\base{\cdot}\), and so is \(\ind\) by \autoref{def:ind-on-ccsk}.
\end{proof}

\begin{proposition}\label{prop:ccsk axioms}
The LTSI of \ccsk satisfies PCI, IRE and RPI.
\end{proposition}

\begin{proof}
By transferring the axioms from \pccsk.  Note that we cannot use \autoref{prop:LLG}, since the LTSI of \ccsk is not LLG.

For PCI, we use mappings $\prov{\cdot}$ and $\base{\cdot}$, much as for the proof of SP (\autoref{prop:ccsk basic}).

For IRE, suppose $t \eveqt t' \ind u$.  Then $\prov t \eveqt \prov{t'} \ind \prov u$, using \autoref{prop:eveqt-proved-base}.  By IRE for \pccsk we get $\prov t \ind \prov u$, implying $t \ind u$ as required.

For RPI, we use mappings $\prov{\cdot}$ and $\base{\cdot}$, much as for PCI.
\end{proof}

We finally have all the elements in place to prove \autoref{thm:axioms-hold}:

\begin{proof}[Proof of \autoref{thm:axioms-hold}]
By Propositions~\ref{prop:pccsk-basic}, \ref{prop:pccsk IRE}, \ref{prop:ccsk basic} and \ref{prop:ccsk axioms}.
\end{proof}

\subsubsection{Proof of \autoref{prop:pccsk-basic}}
\label{ssect:pccsk-basic}

We now provide a complete proof of \autoref{prop:pccsk-basic}.
As in the first papers that proved this result~\cite{Aub22,aubert2023c}, the main challenge is to prove SP. \label{p:complete-proof}

\begin{remark}[Differences with original proof]
\label{rem:changes}
Our proof below of SP and BTI have some differences with the original proof. %
Indeed, independence (originally called \enquote{concurrency}~\cite[Deﬁnition 10]{aubert2023c}) was \emph{defined} by complementarity, instead of being \emph{proved} complementary, as we do now with \autoref{thm:complementarity}.
Providing a direct definition solves a minor problem with the original definition, which considered \eg
\begin{align*}
	a \Par b & \pr{f}[\lmidl a][m] a[m] \Par b && \text{ and } & a \Par b & \pr{f}[\lmidr b][m] a \Par b[m]
\end{align*}
independent, even if there cannot be cofinal transitions to \(a[m] \Par b[m]\), hence violating SP.
This adjustment primarily impacts the proof of BTI and of Equations~\ref{eq:sd}--\ref{eq:sp}, presented below.
\end{remark}

However, the proof of SP, presented p.~\pageref*{proof:pccskSPBTIWF}, requires the same three main ingredients:
\begin{enumerate}
\item %
The following three implications~\cite[Lemmas 7--9]{aubert2023c},
\begin{align}
	X \pr{f}[\theta_1] X_1 \pr{f}[\theta_2] Y \text{ with } \theta_1 \ind \theta_2 %
	& \implies \exists X_2, X \pr{f}[\theta_2] X_2 \pr{f}[\theta_1] Y \label{eq:sd}\\
	X \pr{f}[\theta_1] X_1 \pr{b}[\theta_2] Y\text{ with }\theta_1 \ind \theta_2 %
	& \implies \exists X_2, X \pr{b}[\theta_2] X_2 \pr{f}[\theta_1] Y \label{eq:rd}\\
	X \pr{b}[\theta_1] X_1 \pr{f}[\theta_2] Y\text{ with }\theta_1 \ind \theta_2 %
	& \implies \exists X_2, X \pr{f}[\theta_2] X_2 \pr{b}[\theta_1] Y \label{eq:sp}
\end{align}
which treats separately the various combinations of forward and backward transitions needed to facilitate the proof of SP.
Considering the changes discussed in \autoref{rem:changes}, we give the adjusted proofs pp.~\pageref*{proof:eq:sd}--\pageref*{proof:eq:sp}.
\item \label{item:def-projection} In turn, Equations~\ref{eq:sd}--\ref{eq:sp}, as well as BTI, require four functions on paths, \(\pi_{\D}\) and \(\rho_{\D}\), for \(\D \in \{\L, \R\}\), that projects a transition originating from two processes in parallel (\eg \(X \Par Y\)), or summed (\eg \(X + Y\)), respectively, onto its component on the \(\D\) side\footnote{%
	\Eg \(\pi_{\L} (a \Par X \pr{f}[\lmidl a][k] a[k] \Par X) = a \pr{f}[a][k] a[k])\) and \(\rho_{\R}(X + a[k] \pr{b}[\lplusr a][k] X + a) = a[k] \pr{b}[a][k] a\).	
}%
, and a lemma stating that this extraction preserves independence~\cite[Sect. 4.1]{aubert2023c}:
\begin{align}
	& \forall \D \in \{\L, \R\}, r: X_1 \pr{fb}[\theta_1] X_2 \pr{bf}[\theta_2] X_3, \theta_1 \ind \theta_2 \notag \\ 
	& \quad \pi_{\D}(r): \pi_{\D}(X_1) \pr{fb}[\pi_{\D}(\theta_1)] \pi_{\D}(X_2) \pr{bf}[\pi_{\D}(\theta_2)] \pi_{\D}(X_3) \text{ is defined} \notag \\
	& \quad \text{(\resp ~}\rho_{\D}(r): \rho_{\D}(X_1) \pr{fb}[\rho_{\D}(\theta_1)] \rho_{\D}(X_2) \pr{bf}[\rho_{\D}(\theta_2)] \rho_{\D}(X_3) \text{ is defined)} \notag \\
	& \qquad \implies \pi_{\D}(\theta_1) \ind \pi_{\D}(\theta_2) \text{ (\resp ~} \rho_{\D}(\theta_1) \ind \rho_{\D}(\theta_2) \text{)} \label{eq:extract}
\end{align}	
We use the exact same definition and lemma, but please the reader to note that in our case, proving \autoref{eq:extract} is immediate due to our direct definition of \(\ind\) in \autoref{fig:pccsk-relations}: C\(^{1}\) and P\(^{1}\) provide this lemma immediately.
\item Also, Equations~\ref{eq:sd}--\ref{eq:rd} require the definition of a \emph{removal function \(\rem_{k}^{\alpha}\)}~\cite[Def. 8]{aubert2023c} that removes occurrences of \(\alpha[k]\) and \(\out{\alpha}[k]\) in a process, along with a simple lemma~\cite[Lemma 3]{aubert2023c} proving that this function preserves derivability under some conditions on keys:
\begin{align}
	\forall X, \alpha, k, \theta, & \kay{\theta} \neq k \text{ and } k \notin \key{\rem_{k}^{\alpha}(X)} \implies \notag\\
	& (X \pr{fb}[\theta] Y  \iff \rem_{k}^{\alpha}(X) \pr{fb}[\theta] \rem_{k}^{\alpha}(Y)) \label{eq:rm-preserves}	
\end{align}
We use the exact same definition and proof of \autoref{eq:rm-preserves}.
\end{enumerate}

\begin{proof}[Proof of \autoref{eq:sd}]
\label{proof:eq:sd}
In short, the proof proceeds by induction on the length of the deduction for the derivation of \(X \pr{f}[\theta_1] X_1\), using \autoref{eq:extract} to enable the induction hypothesis if \(\theta_1\) is not a prefix.
The proof requires a particular care when \(X\) is not standard, more particularly if the last rule is pre, but \autoref{eq:rm-preserves} provides just what is needed to deal with this case.

The proof proceeds by induction on the length of the deduction for the derivation of \(X \pr{f}[\theta_1] X_1\).
\begin{description}
	\item[Length \(1\)]
	In this case, the derivation is a single application of act, and \(\theta_1\) is of the form \(\alpha[k]\). 
	But \(\alpha[k] \ind \theta_2\) cannot hold by \autoref{lem:ind-alpha}, so this case is vacuously true.
	\item[Length \(> 1\)]
	We proceed by case on the last rule.
	\begin{description}
		\item[pre] There exists \(\alpha\), \(k\), \(X'\) and \(X_1'\) \st  \(X = \alpha[k].X' \pr{f}[\theta_1] \alpha[k].X_1' = X_1\) and \(\kay{\theta_1} \neq k\).
		As \(\alpha[k].X_1' \pr{f}[\theta_2] Y\) we know that \(\kay{\theta_2} \neq k\)~\cite[Lemma 3.4]{LaneseP21}.
		Furthermore, since \(k\) occurs attached to \(\alpha\) in \(X_1\) and since \(X_1\) makes a \emph{forward} transition to reach \(Y\), \(k \notin \keys{\rem^{\alpha}_k(X_1)} \cup \keys{\rem^{\alpha}_k(Y)}\).
		Hence, we can apply \autoref{eq:rm-preserves} from left to right twice to obtain
		\[\rem^{\alpha}_k (\alpha[k].X') = X' \pr{f}[\theta_1] \rem^{\alpha}_k (\alpha[k].X'_1) = X'_1 \pr{f}[\theta_2] \rem^{\alpha}_k (Y)\]
		As \(\theta_1 \ind \theta_2\) by hypothesis, we can use the induction hypothesis to obtain that there exists \(X_2\) \st  \(X' \pr{f}[\theta_2] X_2 \pr{f}[\theta_1] \rem^{\alpha}_k (Y)\).
		Since \(\kay{\theta_2} \neq k\), we can append pre\ to the derivation of \(X' \pr{f}[\theta_2] X_2\) to obtain \(\alpha[k] . X' = X \pr{f}[\theta_2] \alpha[k]. X_2\).
		Using \autoref{eq:rm-preserves} again, but from right to left, we obtain that \(\rem^{\alpha}_k(\alpha[k].X_2) = X_2 \pr{f}[\theta_1] \rem^{\alpha}_k(Y)\) implies \(\alpha[k].X_2 \pr{f}[\theta_1] Y\), which concludes this case.
		\item[res] This is immediate by induction hypothesis.
		\item[\(\mathbf{\lmidl}\)]
		There exists \(X_{\L}\), \(X_{\R}\), \(\theta_1'\), \(X_{1_{\L}}\), and \(Y_{\L}\), \(Y_{\R}\) \st  \(X \pr{f}[\theta_1] X_1 \pr{f}[\theta_2] Y\) is
		\[X_{\L} \mid X_{\R} \pr{f}[\lmidl \theta_1'] X_{1_{\L}} \mid X_{\R} \pr{f}[\theta_2] Y_{\L} \mid Y_{\R}\text{.}\]
		Then, \(\pi_{\L}(X_{\L} \mid X_{\R} \pr{f}[\lmidl \theta_1'] X_{1_{\L}} \mid X_{\R}) = X_{\L} \pr{f}[\theta_1'] X_{1_{\L}}\) and the proof proceeds by case on \(\theta_2\):
		\begin{description}
			\item[\(\theta_2\) is \(\lmidr \theta_2'\)]
			Then \(X_{\R} \pr{f}[\theta_2'] Y_{\R}\), \(X_{1_{\L}} = Y_{\L}\) and the occurrences of the rules \(\lmidl\) and \(\lmidr\) can be swapped to obtain 
			\[X_{\L} \mid X_{\R} \pr{f}[\lmidr \theta_2'] X_{\L} \mid Y_{\R} \pr{f}[\lmidl\theta_1'] Y_{\L} \mid Y_{\R}\text{.}\]
			\item[\(\theta_2\) is \(\lmidl \theta_2'\)] Then, \(X_{\L} \pr{f}[\theta_1'] X_{1_{\L}} \pr{f}[\theta_2'] Y_{\L}\) and \(X_{\R} = Y_{\R}\).
			As \(\lmidl\theta_1' = \theta_1 \ind \theta_2 = \lmidl\theta_2'\), it is the case that \(\theta_1' \ind \theta_2'\) in \(X_{\L} \pr{f}[\theta_1'] X_{1_{\L}} \pr{f}[\theta_2'] Y_{\L}\) by \autoref{eq:rm-preserves}, and we can use the induction hypothesis  to obtain \(X_2\) \st  \(X_{\L} \pr{f}[\theta_2'] X_2 \pr{f}[\theta_1'] Y_{\L}\), from which it is immediate to obtain \(X_{\L} \mid X_{\R} \pr{f}[\lmidl \theta_2'] X_2 \mid X_{\R} \pr{f}[\lmidl\theta_1'] Y_{\L} \mid X_{\R} = Y_{\L} \mid Y_{\R}\).
			\item[\(\theta_2\) is \(\langle \lmidl {\theta_{2_{\L}}}, \lmidr {\theta_{2_{\R}}}\rangle\)] Since \(\lmidl\theta_1' = \theta_1 \ind \theta_2 = \langle \lmidl {\theta_{2_{\L}}}, \lmidr {\theta_{2_{\R}}}\rangle\), we have that \(\theta_1' \ind {\theta_{2_{\L}}}\) in \(X_{\L} \pr{f}[\theta_1'] X_{1_{\L}} \pr{f}[{\theta_{2_{\L}}}] Y_{\L}\) by \autoref{eq:rm-preserves}.
			Hence, we can use the induction hypothesis  to obtain \(X_{\L} \pr{f}[{\theta_{2_{\L}}}] X_2 \pr{f}[\theta_1'] Y_{\L}\).
			Since we also have that \(X_{\R} \pr{f}[{\theta_{2_{\R}}}] Y_{\R}\), we can compose both paths using first  syn, then \(\lmidl\) to obtain
			\[X_{\L} \mid X_{\R} \pr{f}[\langle \lmidl {\theta_{2_{\L}}}, \lmidr {\theta_{2_{\R}}}\rangle] X_2 \mid Y_{\R} \pr{f}[\lmidl\theta_1'] Y_{\L} \mid Y_{\R}\text{.}\]
		\end{description}
		\item[\(\lmidr\)] This is symmetric to \(\lmidl\).
		\item[syn]
		There exists \(X_{\L}\), \(X_{\R}\), \({\theta_1}_{\L}\), \({\theta_1}_{\L}\), \(X_{1_{\L}}\), \(X_{1_{\R}}\), \(Y_{\L}\) and \(Y_{\R}\) \st  \(X \pr{f}[\theta_1] X_1 \pr{f}[\theta_2] Y\) is
		\[X_{\L} \mid X_{\R} \pr{f}[\langle \lmidl {\theta_1}_{\L}, \lmidr {\theta_{1_{\R}}} \rangle] X_{1_{\L}} \mid X_{1_{\R}} \pr{f}[\theta_2] Y_{\L} \mid Y_{\R}\text{.}\]
		Then, 
		\begin{align*} 
			\pi_{\L}(X_{\L} \mid X_{\R} \pr{f}[\langle \lmidl {\theta_1}_{\L}, \lmidr {\theta_{1_{\R}}} \rangle] X_{1_{\L}} \mid X_{1_{\R}}) &=  X_{\L} \pr{f}[{\theta_1}_{\L}] X_{1_{\L}}\\
			\pi_{\R}(X_{\L} \mid X_{\R} \pr{f}[\langle \lmidl {\theta_1}_{\L}, \lmidr {\theta_{1_{\R}}} \rangle] X_{1_{\L}} \mid X_{1_{\R}}) &= X_{\R} \pr{f}[{\theta_{1_{\R}}}] X_{1_{\R}}
		\end{align*}
		and the proof proceeds by case on \(\theta_2\):
		\begin{description}
			\item[\(\theta_2\) is \(\lmidr {\theta_{2_{\R}}}\)]
			Then \(X_{1_{\R}} \pr{f}[{\theta_{2_{\R}}}] Y_{\R}\), \(X_{1_{\L}} = Y_{\L}\) and \(\langle \lmidl {\theta_1}_{\L}, \lmidr {\theta_{1_{\R}}} \rangle \ind \lmidr {\theta_{2_{\R}}}\). Then by \autoref{eq:extract} \(X_{\R} \pr{f}[{\theta_{1_{\R}}}] X_{1_{\R}} \pr{f}[{\theta_{2_{\R}}}]Y_{\R}\) and \({\theta_{1_{\R}}} \ind {\theta_{2_{\R}}}\).
			We can then use the induction hypothesis to obtain \(X_{\R} \pr{f}[{\theta_{2_{\R}}}] X_{2_{\R}} \pr{f}[{\theta_{1_{\R}}}] Y_{\R}\) from which it is immediate to obtain 
			\[X_{\L} \mid X_{\R} \pr{f}[\lmidr {\theta_{2_{\R}}}] X_{\L} \mid X_{2_{\R}}  \pr{f}[\langle \lmidl \theta_{2_{\L}}, \lmidr {\theta_{1_{\R}}} \rangle] X_{1_{\L}} \mid Y_{\R} = Y_{\L} \mid Y_{\R}\text{.}\]
			\item[\(\theta_2\) is \(\lmidl {\theta_{2_{\L}}}\)] This is symmetric to \(\lmidr {\theta_{2_{\R}}}\).
			\item[\(\theta_2\) is \(\langle \lmidl {\theta_{2_{\L}}}, \lmidr {\theta_{2_{\R}}} \rangle\)] This case is essentially a combination of the two previous cases.
			Since \(\langle \lmidl {\theta_1}_{\L}, \lmidr {\theta_{1_{\R}}} \rangle = \theta_1 \ind \theta_2 = \langle \lmidl {\theta_{2_{\L}}}, \lmidr {\theta_{2_{\R}}}\rangle\), \autoref{eq:extract} gives the two paths
			\begin{align*}
				X_{\L} \pr{f}[{\theta_1}_{\L}] X_{1_{\L}} \pr{f}[{\theta_{2_{\L}}}] Y_{\L}\qquad \text{and} \qquad 
				X_{\R} \pr{f}[{\theta_{1_{\R}}}] X_{1_{\R}} \pr{f}[{\theta_{2_{\R}}}] Y_{\R}
			\end{align*}
			and \({\theta_1}_{\L} \ind {\theta_{2_{\L}}}\) and \({\theta_{1_{\R}}} \ind {\theta_{2_{\R}}}\), respectively.
			By induction hypothesis, we obtain two paths 
			\begin{align*}
				X_{\L} \pr{f}[{\theta_{2_{\L}}}] X_{2_{\L}} \pr{f}[{\theta_1}_{\L}] Y_{\L}\qquad \text{and} \qquad 
				X_{\R} \pr{f}[{\theta_{2_{\R}}}] X_{2_{\R}} \pr{f}[{\theta_{1_{\R}}}] Y_{\R}
			\end{align*}
			that we can then combine using syn twice to obtain, as desired,
			\[X_{\L} \mid X_{\R} \pr{f}[\langle \lmidl {\theta_{2_{\L}}}, \lmidr {\theta_{2_{\R}}} \rangle] X_{2_{\L}} \mid X_{2_{\R}} \pr{f}[\langle \lmidl {\theta_1}_{\L}, \lmidr {\theta_{1_{\R}}} \rangle] Y_{\L} \mid Y_{\R}\text{.}\]
		\end{description}
		\item[\( \lplusl \)] 
		There exists \(X_{\L}\), \(X_{\R}\), \(\theta_1'\), \(\theta_2'\), \({X_1}_{\L}\), and \(Y_{\L}\) \st  \(X \pr{f}[\theta_1] X_1 \pr{f}[\theta_2] Y\) is
		\[X_{\L} + X_{\R} \pr{f}[ \lplusl \theta_1'] {X_1}_{\L} + X_{\R} \pr{f}[ \lplusl \theta_2'] Y_{\L} + X_{\R}\text{.}\]
		All transitions happen on the left side and \(X_{\R}\) remains unchanged as otherwise we could not sum two non-standard terms, so that \(\theta_2\) must be of the form \( \lplusl \theta_2'\).
		Then, we can use \autoref{eq:extract} to obtain
		\[X_{\L} \pr{f}[\theta_1'] {X_1}_{\L} \pr{f}[\theta_2'] Y_{\L}\]
		and \(\theta_1' \ind \theta_2'\).
		Hence we can use the induction hypothesis to obtain \(X_2\) \st \(X_{\L} \pr{f}[\theta_2'] X_2 \pr{f}[\theta_1'] Y_{\L}\).
		From this, it is easy to obtain 
		\[X_{\L} + X_{\R} \pr{f}[ \lplusl \theta_2'] X_2 + X_{\R} \pr{f}[ \lplusl \theta_1'] Y_{\L} + X_{\R} = Y_{\L} + Y_{\R}\text{.}\]
		\item[\( \lplusr \)] This is symmetric to \( \lplusl \).\qedhere
	\end{description}
\end{description}
\end{proof}

It is worth observing that in the proofs of Equations~\ref{eq:rd} and \ref{eq:sp} that follows, the cases of \(t;\Rev{t}: X \pr{f}[\theta_1] X_1 \pr{b}[\theta_1] X\), or of \(\Rev{t};t\) need not to be examined, since \(\theta_1 \ind \theta_1\) never holds since \(\ind\) is irreflexive.
The proofs essentially follows the proof of \autoref{eq:sd}, leveraging the fact that Equations~\ref{eq:extract} and \ref{eq:rm-preserves} hold for both directions: we only highlight the differences with the proof of \autoref{eq:sd} below.

\begin{proof}[Proof of \autoref{eq:rd}]
\label{proof:eq:rd}
The only case that diverges non-trivially with the proof of \autoref{eq:sd} is if the deduction for \(X \pr{f}[\theta_1] X_1\) have for last rule pre. 
In this case, 
\[\alpha[k].X' \pr{f}[\theta_1] \alpha[k].X'_1 \pr{b}[\theta_2] Y\text{,}\]
but we cannot deduce that \(\kay{\theta_2} \neq k\) immediately.
Using \autoref{lem:ind-diff-keys}, however, gives  that \(\kay{\theta_1} \neq \kay{\theta_2}\) since \(\theta_1 \ind \theta_2\), from which 
we can carry out the rest of the proof, using \autoref{eq:rm-preserves} as before.
\end{proof}

\begin{proof}[Proof of \autoref{eq:sp}] 
\label{proof:eq:sp}

The only case that diverges non-trivially with the proof of \autoref{eq:sd} is when summand operands are involved, \ie if the deduction for \(X \pr{b}[\theta_1] X_1\) have for last rule \(\rlplusl\) or \(\rlplusr\).
In the case of \(\rlplusl\) (the \(\rlplusr\) case is symmetric), there exists \(X_{\L}\), \(X_{\R}\), \({X_1}_{\L}\), and \(Y_{\L}\) \st  \(X  \pr{b}[\theta_1] X_1 \pr{f}[\theta_2] Y\) is

\[X_{\L} + X_{\R}  \pr{b}[\lplusl \theta_1'] {X_1}_{\L} + X_{\R} \pr{f}[\theta_2] Y_{\L} + Y_{\R}\text{.}\]

Then, \(\rho_{\L}(X_{\L} + X_{\R}  \pr{b}[\lplusl \theta_1'] {X_1}_{\L} + X_{\R}) = X_{\L}  \pr{b}[\theta_1'] {X_1}_{\L}\) and we proceed by case on \(\theta_2\):
\begin{description}
	\item[\(\theta_2\) is \( \lplusl \theta_2'\)]
	Then, \(\rho_{\L}({X_1}_{\L} + X_{\R} \pr{f}[\lplusl \theta_2'] Y_{\L} + Y_{\R}) = {X_1}_{\L} \pr{f}[\theta_2'] Y_{\L}\) and \(X_{\R} = Y_{\R}\).
	Since \( \lplusl \theta_1' \ind  \lplusl \theta_2'\), we can use \autoref{eq:extract} to obtain
	\[X_{\L}  \pr{b}[\theta_1'] {X_1}_{\L} \pr{f}[\theta_2'] Y_{\L} \]
	and \(\theta_1' \ind \theta_2'\), and by induction hypothesis there exists \(X_2\) such that
	\[X_{\L} \pr{f}[\theta_2'] X_2  \pr{b}[\theta_1'] Y_{\L} \]
	from which it is easy to obtain
	\[X_{\L} + X_{\R} \pr{f}[ \lplusl \theta_2'] X_2 + X_{\R}  \pr{b}[\lplusl \theta_1'] Y_{\L} + X_{\R} = Y_{\L} + Y_{\R}\text{.}\]	
	\item[\(\theta_2\) is \( \lplusr  \theta_2'\)]
	Since \(\lplusl \theta_1' \sdep  \lplusr \theta_2'\) by C\(^2\), and since the two transitions are obviously connected, by \autoref{thm:complementarity} it cannot be the case that \(\theta_1 \ind \theta_2\), so this case is vacuously true.\qedhere
\end{description} 
\end{proof}

\begin{proof}[Proof of \autoref{prop:pccsk-basic}]
\label{proof:pccskSPBTIWF}
We can now prove that the LTSI of $\pccsk$ satisfies SP, BTI and WF:
\begin{description}
	\item[SP]
	The proof is by case on the directions of the transitions, and always follows the same pattern: use \autoref{lem:loop_proved} to orient the transitions to meet the premises of Equation~\ref{eq:sd}, \ref{eq:rd} or \ref{eq:sp}, use the appropriate equation to obtain new paths, and finally use again \autoref{lem:loop_proved} to orient them as desired.
	\item[BTI] We have to prove that any two different coinitial backward transitions \(t_1: X\pr{b}[\theta_1]X_1\) and \(t_2:X\pr{b}[\theta_2]X_2\) are independent.
	The first important fact to note is that \(\kay{\theta_1} \neq \kay{\theta_2}\): by a simple inspection of the backward rules in \autoref{fig:provedltsrulesccskfw}, it is easy to observe that if a reachable process \(X\) can perform two different backward transitions, then their labels must have different keys.
	
	We then proceed by induction on the length of the deduction for the derivation of \(X \pr{b}[\theta_1] X_1\):
	\begin{description}
		\item[Length \(1\)]
		In this case, the derivation is a single application of \tRev{act}, and \(\theta_1\) is of the form \(\alpha[k]\), with \(X = \alpha[k].X'\) and \(\std{X'}\).
		Hence, \(X\) cannot perform two different transitions, and this case is vacuously true.
		\item[Length \(> 1\)]
		We proceed by case on the last rule.
		\begin{description}
			\item[\tRev{pre}] There exists \(\alpha\), \(k\), \(X'\) and \(X_1'\) \st  \(X = \alpha[k].X' \pr{b}[\theta_1] \alpha[k].X_1' = X_1\).
			Then, it must be the case that \(X'\pr{b}[\theta_1]X_1'\) and \(X'\) is not standard.
			Since \(X'\) is not standard, the last rule for the derivation of \(X \pr{b}[\theta_2] X_2\) cannot be \tRev{act}, and since \(X = \alpha[k].X'\), it must be \tRev{pre}, hence it must be the case that 	
			\(X = \alpha[k].X' \pr{b}[\theta_2] \alpha[k].X_2' = X_2\), and we know that \(X'\pr{b}[\theta_2]X_2'\).
			We conclude by using the induction hypothesis on the two backward transitions of \(X'\) and the observation that \tRev{pre} preserves the label and hence independence.
			\item[\tRev{res}] This is immediate by induction hypothesis.
			\item[\(\rlmidl\)]
			There exists \(X_{\L}\), \(X_{\R}\), \(\theta_1'\) and \(X_{1_{\L}}\) \st  \(X \pr{b}[\theta_1] X_1\) is
			\[X_{\L} \mid X_{\R} \pr{b}[\lmidl \theta_1'] X_{1_{\L}} \mid X_{\R}\text{.}\]
			Then, \(\phi_{\L}(X_{\L} \mid X_{\R} \pr{b}[\lmidl \theta_1'] X_{1_{\L}} \mid X_{\R}) = X_{\L} \pr{b}[\theta_1'] X_{1_{\L}}\) and the proof proceeds by case on \(\theta_2\), using \autoref{eq:extract} to decompose the paths:
			\begin{description}
				\item[\(\theta_2\) is \(\lmidr \theta_2'\)]
				Then \(\lmidl \theta_1' \ind \lmidr\theta_2'\) is immediate by P\(^2_k\) since we know that \(\kay{\theta_1} \neq \kay{\theta_2}\).
				\item[\(\theta_2\) is \(\lmidl \theta_2'\)] Then by \autoref{eq:extract} there exists \(X_{2_{\L}}\) such that \(X_{\L} \pr{b}[\theta_2'] X_{2_{\L}}\), and we conclude by induction on \(X_{\L}\)'s backward transitions using  P\(^1\).
				\item[\(\theta_2\) is \(\langle \lmidl {\theta_{2_{\L}}}, \lmidr {\theta_{2_{\R}}}\rangle\)]
				Then we know that
				\[X_{L} \mid X_{\R} \pr{b}[\langle \lmidl {\theta_{2_{\L}}}, \lmidr {\theta_{2_{\R}}}\rangle] X_{2_{\L}} \mid X_{2_{\R}}\text{.}\]
				For \(\lmidl  \theta_1' \ind \langle \lmidl {\theta_{2_{\L}}}, \lmidr {\theta_{2_{\R}}}\rangle\) to hold, S\(^1\) requires \(\theta_1' \ind  {\theta_{2_{\L}}}\).
				By induction hypothesis on \(X_{\L} \pr{b}[\theta_1'] X_{1_{\L}}\) and \(X_{\L} \pr{b}[\theta_{2_{\L}}] X_{2_{\L}}\), we know that those two transitions are independent, which concludes this case. 
			\end{description}
			\item[\(\rlmidr\)]This is symmetric to \(\rlmidl\).
			\item[\tRev{syn}] This case is similar to the two previous ones and does not offer any insight nor resistance.
		\end{description}
		
		\item[\(\rlplusl\)]
		There exists \(X_{\L}\), \(X_{\R}\), and \({X_1}_{\L}\) \st  \(X \pr{b}[\theta_1] X_1\) is
		\[X_{\L} + X_{\R} \pr{b}[\lplusl \theta_1'] {X_1}_{\L} + X_{\R}\text{.}\]			
		Then, note that \(\theta_2\) must also be of the form \( \lplusl \theta_2'\), as \(X_{\R}\) must be standard.
		Hence, this follows by induction hypothesis on the transitions \(X_{\L} \pr{b}[\theta_1'] {X_1}_{\L}\) and  \(X_{\L} \pr{b}[\theta_2'] {X_2}_{\L}\), using \autoref{eq:extract} to decompose the path and C\(^1\).\qedhere
	\end{description}
	
	\item[WF] We have to prove that for all (reachable) \(X\), there exists \(n \in \mathbb{N}\) and \(X_0, \hdots, X_n\) \st  \(X \pr{b} X_n \pr{b} \cdots \pr{b} X_1 \pr{b} X_0\), with \(\std{X_0}\).
	It is easy to observe that letting \(n = \card{\keys{X}}\) always gives the required result, since the target of a backward transition always contain one fewer key than the source of said transition, and since \(X\) can always revert its forward transitions in the opposite order.
\end{description}
\end{proof}

    \subsection{Proof of \autoref{lem:ind sdep coind}}
\label{app:ind sdep coind}
\indsdepcoind*

\begin{proof}
	\begin{enumerate}
		\item 
		By IRE (\autoref{prop:pccsk IRE}) and complementarity (\autoref{thm:complementarity}).
		\item
		By IRE.
		\item
		Suppose $e_1,e_2$ are composable and $e_1 \ind e_2$.  By IRE we have composable $t_1 \in e_1$ and $t_2 \in e_2$ such that $t_1 \ind t_2$.  By RPI (\autoref{prop:pccsk IRE}), $\rev{t_1} \ind t_2$.
		Since $\rev{t_1}$ and $t_2$ are coinitial, $e_1 \coind e_2$ as required.
		\qedhere
	\end{enumerate}
\end{proof}

\subsection{\autoref{ssec:key-props-ccsk}: Proofs of Lemmas and Propositions}
\label{app:key-based properties}

\inddiffkeys*
\begin{proof}
From the rules for $\ind$ we see that if $\theta_1 \ind \theta_2$ then $\key {\theta_1} \neq \key {\theta_2}$.  Hence for \pccsk, if $t_1 \ind t_2$ then $\key {t_1} \neq \key {t_2}$.
For \ccsk, using $\prov\cdot$ is as in \autoref{def:ccsk-pccsk-bijection}, if $t_1 \ind t_2$ then $\prov {t_1} \ind \prov {t_2}$, so that $\key {\prov {t_1}} \neq \key {\prov {t_2}}$ and so $\key {t_1} \neq \key {t_2}$.

\end{proof}

\bwdkeydet*
\begin{proof}
	For \ccsk, suppose that $t_1: X \r{b}[\alpha_1][k] X_1$, $t_2: X \r{b}[\alpha_2][k] X_2$,
	and that $t_1 \neq t_2$.
	By BTI (\autoref{prop:pccsk-basic}) %
	 we have $t_1 \ind t_2$.
	By \autoref{lem:ind-diff-keys} we have $\key {t_1} \neq \key {t_2}$, which is a contradiction.
	The proof is similar for \pccsk.
\end{proof}

\eventcoincide*

\begin{proof}%
	We present the proof for \pccsk; it works equally well for \ccsk.
	
	($\Rightarrow$)
	Suppose that $t_1 \eveqt t_2$.
	It is enough to consider a single diamond.  But then the targets of $t_1$ and $t_2$ are joined by $t$ such that $t_1 \ind t$.  Then $t_1$ and $t$ have different keys by \autoref{lem:ind-diff-keys}.
	
	($\Leftarrow$)
	Suppose that we have transitions $t_1:X_1 \pr{f}[\theta_1][k] X'_1$ and $t_2:X_2 \pr{f}[\theta_2][k] X'_2$ with $t_1 \evkeqt t_2$.
	We proceed by induction on the length of the path from $X'_1$ to $X'_2$.
	If the path is of length zero then we have $t_1 = t_2$ by backward key determinism (\autoref{lem:bwd_key_det}).
	
	If the path has non-zero length, by \autoref{lem:PL event}
	we can convert the path into a parabolic path without increasing length, and without introducing new labels, and so without introducing new keys.  Hence it is still the case that $k$ does not occur in the parabolic path.
	By the parabolic property, either the first transition is backward or the last transition is forward.  We consider the case where the last transition is forward; the other case is similar.
	
	So suppose that the last transition in the path is $t:Y' \pr{f}[\theta][k'] X'_2$, where $k' \neq k$.
	By BTI we have $\rev t \ind \rev {t_2}$.  We then use SP to complete a diamond with transitions $t'_2: Y \pr{f}[\theta_2][k] Y'$ and $t':Y \pr{f}[\theta][k'] X_2$ for some $Y$.
	By PCI $t'_2 \ind t'$, and so $t'_2 \eveqt t_2$.
	By induction hypothesis $t_1 \eveqt t'_2$.
	Hence $t_1 \eveqt t_2$.
\end{proof}

\directkeyind*
\begin{proof}
Suppose first that $t,u$ are directly key independent.  Then there are $t',u'$ which complete the square and satisfy the conditions for ID.  Here we use $\key t \neq \key u$ to deduce that the square is non-degenerate.
By ID (\autoref{prop:ID}, \autoref{thm:axioms-hold}), $t \ind u$.

Conversely, suppose that $t:P \r{fb}[\alpha][m] Q$, $u:P \r{fb}[\beta][n] R$ with $t \ind u$.
By~\autoref{lem:ind-diff-keys}, $m \neq n$.  By SP (\autoref{thm:axioms-hold}), there are cofinal transitions $u': Q \r{fb}[\beta][n] S$ and $t':R \r{fb}[\alpha][m] S$.
Hence $t,u$ are directly key independent.
\end{proof}

\keyindcoind*

\begin{proof}
By the definitions, \autoref{prop:direct key ind} and~\autoref{prop:event_coincide}.
\end{proof}

\subsection{Proof of \autoref{thm:ordering-event-key}}\label{app:ordering-event-key}

Proving this last result requires some additional lemmas and definitions.

We need a strengthened version of the Parabolic Lemma (PL): in addition to the statement in~\cite[Prop. 3.4]{LPU24}, it states that the new \enquote{parabolic} path should introduce no new events compared to the old path.

\begin{lemma}[Strengthened PL]\label{lem:PL event}
	In an LTSI satisfying SP and BTI, for any path $r$ there are forward-only paths $s,s'$ such that $r \ceqt \rev s s'$ and $\len s + \len {s'} \leq \len r$.
	Moreover, if $t$ in $s$ or $t$ in $s'$ then $t \eveqt t'$ for some $t'$ in $r$.
\end{lemma}
\begin{proof}
	The proof is the same as for~\cite[Prop. 3.4]{LPU24}, where it follows from axioms SP and BTI.  We just note that any new transitions in $s$ or $s'$ are obtained by replacing $t \rev u$ (for some forward transitions $t,u$) by $\rev {u'} t'$ where $t' \eveqt t$ and $u' \eveqt u$.
\end{proof}

\begin{lemma}[Order from events to keys]\label{lem:composable sdep}
	Suppose $X$ is reachable, and suppose that $e_1,e_2 \in \ev{X}$ are such that $e_1 \ip e_2$.
	Then $(\key {e_1}, \key {e_2}) \in \ord X$.
\end{lemma}

\begin{proof}
	We apply Lemmas~\ref{lem:immed pred composable} and \ref{lem:ind sdep coind} as well as  complementarity (\autoref{thm:complementarity}) to deduce that $e_1$ is composable with $e_2$ and $e_1 \sdep e_2$.
	Let $k_1 = \key {e_1}$, $k_2 = \key {e_2}$.
	Composability of $e_1$ and $e_2$ implies that $k_1 \neq k_2$.
	Let $rt_1t_2$ be a forward-only path with $t_1 \in e_1$, $t_2 \in e_2$.
	We have $\lblof {t_1} \sdep \lblof {t_2}$.  We refer to \autoref{fig:pccsk-relations} for rules for $\sdep$ and $\ind$.
	
	By structural induction on $X$.
	There are various cases:
	\begin{description}
		\item[$P$.]
		This case cannot arise, since $\ev P = \emptyset$.
		\item[${\alpha[k].X}$.]
		There are various sub-cases:
		\begin{enumerate}
			\item $k_1 = k$.
			Then $k_2 \in \keys X$, and so $(k_1,k_2) \in \ord {\alpha[k].X}$.
			\item $k_2 = k$.
			This cannot arise, since $t_1$ occurs before $t_2$ in the path $rt_1t_2$.
			\item $k_1,k_2 \neq k$.
			Then $k_1,k_2 \in \keys X$.
			There is a forward-only path $r't'_1t'_2$ obtained by omitting the first transition of $r$ and removing prefixes $\alpha[k]$ from processes.
			Clearly $\lblof{t'_1} \sdep \lblof{t'_2}$.
			By induction we have $(k_1,k_2) \in \ord X$, and so $(k_1,k_2) \in \ord {\alpha[k].X}$.
		\end{enumerate}
		\item[$X + Q$.]
		Then there is a forward-only path $r^{\L}t_1^{\L}t_2^{\L}$ obtained by projecting onto the left-hand component.
		We use rule C$^1$ to deduce $\lblof{t_1^{\L}} \sdep \lblof{t_2^{\L}}$.
		So $\evtkof {X} {k_1} \sdep \evtkof {X} {k_2}$ and these events are composable.
		By induction we have $(k_1, k_2) \in \ord X$, and so $(k_1, k_2) \in \ord {X + Q}$.
		
		\item[$P + X$.]
		Similar to the preceding case.
		
		\item[$X \bs \lambda$.]
		Then there is a forward-only path $r't'_1t'_2$ obtained by removing the restriction.
		Clearly $\evtkof {X} {k_1} \sdep \evtkof {X} {k_2}$ and these events are composable.
		By induction we have $(k_1, k_2) \in \ord X$, and so $(k_1, k_2) \in \ord {X \bs \lambda}$.
		
		\item[$X \Par Y$.]
		There are various sub-cases:
		\begin{enumerate}
			\item $k_1,k_2 \in \keys X \inter \keys Y$.
			Then there are forward-only paths $r^{\D}t_1^{\D}t_2^{\D}$ ($\D \in \{\L, \R\}$) obtained by projecting onto the left-hand and right-hand components.
			We use rule S$^3$ to deduce $\lblof{t_1^{\D}} \sdep \lblof{t_2^{\D}}$ for $\D \in \{\L, \R\}$.
			Wlog suppose $\lblof{t_1^{\L}} \sdep \lblof{t_2^{\L}}$.
			By induction we have $(k_1, k_2) \in \ord X$, and so $(k_1, k_2) \in \ord {X \Par Y}$.
			\item $k_1 \in \keys X \inter \keys Y$ and $k_2 \in \keys X \setminus \keys Y$.
			Then there is a forward-only path $r^{\L}t_1^{\L}t_2^{\L}$ obtained by projecting onto the left-hand component.
			We use rule S$^2$ to deduce $\lblof{t_1^{\L}} \sdep \lblof{t_2^{\L}}$.
			By induction we have $(k_1, k_2) \in \ord X$, and so $(k_1, k_2) \in \ord {X \Par Y}$.
			\item $k_1 \in \keys X \inter \keys Y$ and $k_2 \in \keys Y \setminus \keys X$.
			Similar to the preceding case.
			\item $k_1 \in \keys X \setminus \keys Y$ and $k_2 \in \keys X \inter \keys Y$.
			Then there is a forward-only path $r^{\L}t_1^{\L}t_2^{\L}$ obtained by projecting onto the left-hand component.
			We use rule S$^1$ to deduce $\lblof{t_1^{\L}} \sdep \lblof{t_2^{\L}}$.
			By induction we have $(k_1, k_2) \in \ord X$, and so $(k_1, k_2) \in \ord {X \Par Y}$.
			\item $k_1 \in \keys Y \setminus \keys X$ and $k_2 \in \keys X \inter \keys Y$.
			Similar to the preceding case.
			\item $k_1, k_2 \in \keys X \setminus \keys Y$.
			Then there is a forward-only path $r^{\L}t_1^{\L}t_2^{\L}$ obtained by projecting onto the left-hand component.
			We use rule P$^1$ to deduce $\lblof{t_1^{\L}} \sdep \lblof{t_2^{\L}}$.
			By induction we have $(k_1, k_2) \in \ord X$, and so $(k_1, k_2) \in \ord {X \Par Y}$.
			\item $k_1, k_2 \in \keys Y \setminus \keys X$.
			Similar to the preceding case.
			\item $k_1 \in \keys X \setminus \keys Y$ and $k_2 \in \keys Y \setminus \keys X$.
			This case cannot arise, since we have $\lblof {t_1} \ind \lblof {t_2}$ by rule P$^2_k$.
			\item $k_1 \in \keys Y \setminus \keys X$ and $k_2 \in \keys X \setminus \keys Y$.
			Similar to the preceding case.
			\qedhere
		\end{enumerate}
	\end{description}
\end{proof}

\begin{lemma}[Event keys properties]\label{lem:event key}
	For any reachable \ccsk process $X$:
	\begin{enumerate}
		\item if $e_1,e_2 \in \ev{X}$ then $\key {e_1} \neq \key {e_2}$;
		\item $\{\key e \setst e \in \ev{X}\} = \keys X$.
	\end{enumerate}
\end{lemma}

\begin{proof}%
	\begin{enumerate}
		\item We can consider a forward-only rooted path, which must exist by PL (\autoref{lem:PL event}).  Clearly all keys of transitions must be distinct.
		\item
		By structural induction.  The most interesting case is parallel composition.
		A rooted path with target $X \Par Y$ can be \enquote{projected}\footnote{%
			Here and below, the informal term \enquote{project} is to be read in the technical sense of the \(\pi_{\D}\) and \(\rho_{\D}\) functions discussed on p.~\pageref{item:def-projection}., around item~\itemreflipics{item:def-projection}.} %
		into rooted paths with targets $X$ and $Y$ respectively. \qedhere
	\end{enumerate}
\end{proof}
In view of \autoref{lem:event key} we can make the following definition.
\begin{definition}[Event key]
	Let $X$ be reachable and let $k \in \keys X$.  Define $\evtkof X k$ to be the unique event $e \in \ev{X}$ such that $\key e = k$.
\end{definition}

The events of $\ev{X}$ are partially ordered as in \autoref{def:event relations}.
Recall the ordering on keys generated from $\ord X$ from \autoref{def:partial-order-on-keys}.

\begin{restatable}[Order from keys to events]{lemma}{lemord}\label{lem:ord}
	Suppose $X$ is reachable and $(k_1,k_2) \in \ord X$.  Then $\evtkof X {k_1} < \evtkof X {k_2}$.
\end{restatable}

\begin{proof}
	By structural induction on $X$.
	In the definition of $e_1 < e_2$ we can restrict to forward-only paths using~\cite[Lemma 4.26]{LPU24}, since the LTSI of \pccsk is pre-reversible.
	The proof is best done using proved transitions, since this helps to clarify parallel composition.
	It is convenient to use event equivalence $\evkeqt$ as in \autoref{def:event-key}, which is equivalent to $\eveqt$ as in \autoref{def:event-general} by \autoref{prop:event_coincide}.
	
	There are various cases:
	\begin{description}
		\item[$P$.]
		This case cannot arise, since $\keys P = \emptyset$.
		\item[${\alpha[k].X}$.]
		There are two sub-cases:
		\begin{enumerate}
			\item \label{item:ord base}
			Suppose that $(k_1,k_2) \in \ord {\alpha[k].X}$, and this is derived from $k_1 = k$ and $k_2 \in \keys X$.
			Suppose that $r$ is any rooted forward-only path with $\cte(r,\evtkof {\alpha[k].X} {k_2}) = 1$.
			Then the first transition of $r$ is labelled with $\alpha[k]$, so that $\cte(r,\evtkof {\alpha[k].X} {k_1}) = 1$ also.
			This shows that $\evtkof {\alpha[k].X} {k_1} < \evtkof {\alpha[k].X} {k_2}$.
			
			\item \label{item:ord ind}
			Suppose that $(k_1,k_2) \in \ord {\alpha[k].X}$, and this is derived from $(k_1,k_2) \in \ord X$.
			Suppose also that $\evtkof {\alpha[k].X} {k_1} \not < \evtkof {\alpha[k].X} {k_2}$.
			Then there is a rooted forward-only path $r$ with $\cte(r,\evtkof {\alpha[k].X} {k_2}) = 1$ and $\cte(r,\evtkof {\alpha[k].X} {k_1}) = 0$.
			Suppose that $t$ is the transition in $r$ which belongs to $\evtkof {\alpha[k].X} {k_2}$.  Then $t \evkeqt t'$ where $t'$ belongs to some rooted forward-only path $r'$ with target $\alpha[k].X$.
			Let $s$ be a path not containing $k_2$ from $\tgtof {t}$ to $\tgtof {t'}$.
			
			We can omit the initial $\alpha[k]$ transition in $r$ and $r'$, and delete all $\alpha[k]$ prefixes in $r,r',s$, yielding paths $r_0,r'_0,s_0$ and transitions $t_0 \in r_0$ and $t' \in r'_0$. Using $s_0$ we see that $t_0 \evkeqt t'_0$.  Using $r'_0$ we see that $t'_0 \in \evtkof X {k_2}$, and so $t_0 \in \evtkof X {k_2}$.
			Now $\cte(r,\evtkof X {k_2}) = 1$ and $\cte(r,\evtkof X {k_1}) = 0$,
			showing that $\evtkof X {k_1} \not < \evtkof X {k_2}$, contradicting the induction hypothesis.
		\end{enumerate}
		\item[$X + Q$] Similar to sub-case \itemreflipics{item:ord ind} for $\alpha[k].X$.
		\item[$P + X$] Similar to the preceding case.
		\item[$P \bs \lambda$] Similar to sub-case \itemreflipics{item:ord ind} for $\alpha[k].X$.
		\item[$X \Par Y$]
		Suppose that $(k_1,k_2) \in \ord {X \Par Y}$, and this is derived from $(k_1,k_2) \in \ord X$.
		Suppose also that $\evtkof {X \Par Y} {k_1} \not < \evtkof {X \Par Y} {k_2}$.
		Then there is a rooted forward-only path $r$ with $\cte(r,\evtkof {X \Par Y} {k_2}) = 1$ and $\cte(r,\evtkof {X \Par Y} {k_1}) = 0$.
		Suppose that $t$ is the transition in $r$ which belongs to $\evtkof {X \Par Y} {k_2}$.  Then $t \evkeqt t'$ where $t'$ belongs to some rooted forward-only path $r'$ with target $X \Par Y$.
		Let $s$ be a path not containing $k_2$ from $\tgtof {t}$ to $\tgtof {t'}$.
		
		We can project $r,r',s$ onto their left-hand components (by omitting any moves made solely on the right-hand component), yielding paths $r_\Left,r'_\Left,s_\Left$ and transitions $t_\Left \in r_\Left$ and $t' \in r'_\Left$. Using $s_\Left$ we see that $t_\Left \evkeqt t'_\Left$.  Using $r'_\Left$ we see that $t'_\Left \in \evtkof X {k_2}$, and so $t_\Left \in \evtkof X {k_2}$.
		Now $\cte(r,\evtkof X {k_2}) = 1$ and $\cte(r,\evtkof X {k_1}) = 0$,
		showing that $\evtkof X {k_1} \not < \evtkof X {k_2}$, contradicting the induction hypothesis.
		
		The case where $(k_1,k_2) \in \ord {X \Par Y}$ is derived from $(k_1,k_2) \in \ord Y$ is similar.
		\qedhere
	\end{description}
\end{proof}

\thmorderings*

\begin{proof}
	($\Rightarrow$)
	By \autoref{lem:composable sdep} and \autoref{def:partial-order-on-keys}.
	
	($\Leftarrow$)
	By \autoref{lem:ord}, \autoref{def:partial-order-on-keys} and $\leq$ on $\ev{X}$ being a partial ordering~\cite[Lemma 4.24]{LPU24}.
\end{proof}

\Comment{
\subsection{Maximal Keys}\label{subsec:maxkeys}

\begin{definition}
[Keys of a \ccsk process]\label{def:keys}
\begin{align*}
\keys \nil &= \emptyset &&& \keys {X \bs a} &= \keys X\\
\keys {\alpha.X} &= \keys X &&& \keys {X \Par Y} &= \keys X \union \keys Y \\
\keys {\alpha[k].X} &= \keys X \union \{k\} &&& \keys {X + Y} &= \keys X \union \keys Y 
\end{align*}
\end{definition}

\begin{definition}[Maximal keys of a \ccsk process]\label{def:maxkeys}
	\begin{align*}
		\maxkeys{\nil} & = \emptyset \\
		\maxkeys{\alpha.X} & = \maxkeys{X} \\
		\maxkeys{\alpha[n].X} & = \begin{cases}
		 \maxkeys{X} & \keys X \neq \emptyset \\
		 \{n\} &  \keys X = \emptyset \\
\end{cases}\\
		\maxkeys{X \bs \lambda} & = \maxkeys{X} \\
		\maxkeys{X + Y} &= \maxkeys{X} \union \maxkeys{Y} \\
		\maxkeys{X \Par  Y} &= (\maxkeys{X} \inter \maxkeys{Y}) \\
		& \quad \union (\maxkeys{X} \setminus \keys{Y}) \union  (\maxkeys{Y} \setminus \keys{X})\\
  	\end{align*}
\end{definition}

\begin{lemma}\label{lem:maxkeys}
For any reachable $X$, $\maxkeys X$ is the set of maximal keys of $X$ according to the ordering $\leq_X$ (\autoref{def:partial-order-on-keys}).
\end{lemma}
\begin{proof}
First note that $k$ is maximal in $\keys X$ (under $\leq_X$) iff $k \in \keys X$ and $(k,\star) \notin \ord X$, where we use $\star$ to stand for any key.
We therefore show by structural induction that $k \in \maxkeys X$ iff $k \in \keys X$ and $(k,\star) \notin \ord X$.
We give the case of parallel composition; the other cases are immediate.

Suppose $k \in \maxkeys{X \Par Y}$.
There are three cases:
\begin{enumerate}
\item $k \in \keys X \inter \keys Y$.
Then $k \in \maxkeys X \inter \maxkeys Y$.
By induction $k \in \keys X$, $(k,\star) \notin \ord X$ and $k \in \keys Y$, $(k,\star) \notin \ord Y$.
So $k \in \keys {X \Par Y}$ and $(k,\star) \notin \ord {X \Par Y}$ by definition of $\keys {X \Par Y}$ and  $\ord {X \Par Y}$.
\item $k \in \keys X \setminus \keys Y$.
Then $k \in \maxkeys{X} \setminus \keys{Y}$.
By induction $k \in \keys X$, $(k,\star) \notin \ord X$.
Since $k \notin \keys Y$, we have $(k,\star) \notin \ord Y$.
So $k \in \keys {X \Par Y}$ and $(k,\star) \notin \ord {X \Par Y}$ by definition of $\keys {X \Par Y}$ and  $\ord {X \Par Y}$.
\item $k \in \keys Y \setminus \keys X$.
Similar to the previous case.
\end{enumerate}
Thus in all three cases $k \in \keys {X \Par Y}$ and $(k,\star) \notin \ord {X \Par Y}$ as required.

Conversely, suppose that $k \in \keys {X \Par Y}$ and $(k,\star) \notin \ord {X \Par Y}$.  So $(k,\star) \notin \ord X$ and $(k,\star) \notin \ord Y$ by definition of $\ord {X \Par Y}$.
There are again three cases:
\begin{enumerate}
\item $k \in \keys X \inter \keys Y$.
By induction $k \in \maxkeys X \inter \maxkeys Y$.
\item $k \in \keys X \setminus \keys Y$.
By induction $k \in \maxkeys X \setminus \keys Y$.
\item $k \in \keys Y \setminus \keys X$.
Similar to the previous case.
\end{enumerate}
Thus in all three cases $k \in \maxkeys{X \Par Y}$ as required.
\end{proof}
}
\Comment{
\subsection{Further complexity-related material to be omitted}
The following is very conjectural and not ready to be included:

There are different ways to measure the size of a \ccsk process.
In \autoref{def:input size} we used depth when measuring input size.
However, arguably for finite (non-recursive) \ccsk the size of the origin standard process is more appropriate.
\begin{definition}
[Size of a \ccsk process]
\begin{align*}
\sizeof \nil &= 1 \\
\sizeof {\alpha.X} &= \sizeof X + 1 \\
\sizeof {\alpha[k].X} &= \sizeof X + 1 \\
\sizeof {X \bs a} &= \sizeof X + 1 \\
\sizeof {X \Par Y} &= \sizeof X + \sizeof Y + 1 \\
\sizeof {X + Y} &= \sizeof X + \sizeof Y + 1 \\
\end{align*}
\end{definition}

\begin{proposition}
Given a \ccsk process $X$ with $\sizeof X = N$, we can compute $\keys X$ in time $O(N)$.
\end{proposition}
Recall $\bts X$ from %
\autoref{prop:bts finite}. We can use the backward rules in \autoref{fig:provedltsrulesccskfw} to define this for \ccsk as follows:
\begin{definition}\label{def:bts rules}
\begin{align*}
\bts {\alpha.X} &= \emptyset \\
\bts {\alpha[k].X} &= \begin{cases}
\{(\alpha[k],\alpha.X)\}  & \text{ if } \keys X = \emptyset \\
\{(\theta,\alpha[k]X') : (\theta,X') \in \bts X \text{ and } \key \theta \neq k \} & \text{ if } \keys X \neq \emptyset
\end{cases}\\
\bts {X \bs a} &= \{(\theta,X'\bs a) : (\theta,X') \in \bts X \text{ and } \key \theta \notin \{a,\out a\} \} \\
\bts {X \Par Y} &= \{(\lmidl \theta,X' \Par Y) : (\theta,X') \in \bts X \text{ and } \key \theta \notin \keys Y \} \\
& \union \{(\lmidr \theta,X \Par Y') : (\theta,Y') \in \bts Y \text{ and } \key \theta \notin \keys X \} \\
& \union \{(\cpair{\upsilon_{\L} \lambda [k]}{\upsilon_{\R} \out{\lambda} [k]},X' \Par Y') : (\upsilon_{\L} \lambda[k], X') \in \bts X \text{ and }
(\upsilon_{\R} \out{\lambda}[k],Y') \in \bts Y \} \\
\bts {X + Y} &= \begin{cases}
\{ (\lplusl\theta,X'+Y) : (\theta,X') \in \bts X \} & \text{ if } \keys Y = \emptyset \\
\{ (\lplusr\theta,X+Y') : (\theta,Y') \in \bts Y \} & \text{ if } \keys X = \emptyset 
\end{cases}\\
\end{align*}
\end{definition}
\begin{lemma}\label{lem:bts size}
Given a (reachable) \ccsk process $X$ with $\sizeof X = N$, then $\sizeof {\bts X}$, the number of transitions in $\bts X$, is bounded by $\sizeof {\keys X}$.
\end{lemma}
Rather than using \autoref{def:bts rules} to compute $\bts X$ we might first find the maximal keys in $\keys X$ and then just reverse these.  If $k$ is maximal, then intuitively we just delete $k$ in$X$ to reverse this action---more details to be supplied.

\begin{conjecture}
We can compute $\maxkeys X$ in time $O(N^2)$.
\end{conjecture}
\begin{proof}
The interesting case is $X \Par Y$, where we 
\end{proof}
Can also think of finding those keys which occur twice---could be done efficiently.

\begin{proposition}
Given a (reachable) \ccsk process $X$ with $\sizeof X = N$, we can compute $\bts X$ in time $O(N^2)$.
\end{proposition}
\begin{proof}
We find the maximal keys $\maxkeys X$.  For each $k \in \maxkeys X$ we find the corresponding label $\alpha$.  We find $X'$ such that $X \r{b} [\alpha][k] X'$ by deleting the one or two occurrences of $k$ in $X$.
\end{proof}

}

    \section{\autoref{sec:bisimulations}: Proofs and Additional Results}
\label{app:bisimulation}

\Comment{
\subsection{\autoref{ssec:dep-is-enough}: Proofs and Additional Results}

\conntrandepiscausality*
\Comment{\begin{lemma}\label{lem:im-fcau-iff-dep}
		Let $t_1, t_2$ be composable forward transitions in a pre-reversible LTSI satisfying \IRE. 
		Then  $t_1 < t_2$ iff $t_1 \sdep t_2$.  
	\end{lemma}
}

\irek{proof adjusted for forward-only transitions.} 
\begin{proof}
	\begin{enumerate}
		\item 
		$\Rightarrow$ Assume $t_1 < t_2$. Assume for contradiction $\rev{t}_1 \ind t_2$. By SP there are transitions $t_2'$ and $\rev{t}_1'$ making up a commuting diamond with 
		$\rev{t}_1$ and $t_2$. Since $t_2'(\eveqt t_2)$ is coinitial with $t_1$ there must be a transition $t_1''$ such that
		$t_1''\eveqt t_1$, which is before $t_2'$ in order for 
		$t_1 < t_2$ to hold. But that contradicts NRE. Hence  $\rev{t}_1 \not\ind t_2$, and $\rev{t}_1 \sdep t_2$ by complementarity. Since
		the labels of $\rev{t}_1$ and $t_1$ are the same we obtain $t_1 \sdep t_2$.
		
		$\Leftarrow$  By \autoref{lem:ind sdep coind} \clem{Check me}. 
		\item
		\iain{Not sure what is intended here for events?  Proof actaully in main text preceding \autoref{{lem:im-fcau-iff-dep}}}  This follows from \autoref{lem:immed pred composable} and \autoref{lem:ind sdep coind}, noting that $e_1 < e_2$ iff $e_1 \ip e_2$ for composable forward events.
	\end{enumerate}
	\qedhere
\end{proof}

\subsection{\autoref{ssec:bisim}: Proofs and Additional Results}
}

The goal of this appendix section is to prove \autoref{thm:bisim}, which only applies to standard \pccsk processes. This means that,
especially in the case of DP bisimulation, it is important to distinguish between arbitrary triples and DP-grounded triples. Thus we work here with KP-grounded and DP-grounded triples. 

Since we consider DP bisimulation only on standard processes, we have that maximal events get mapped to maximal events by label-preserving bijections of DP-grounded triples:

\Comment{
\begin{lemma}
	\label{prop:dep-implies-not-concurrency}
	Let $t_1, t_2$ be connected transitions in a pre-reversible LTSI. %
	Then $t_1 \sdep t_2$ implies not $t_1 \coind t_2$.
\end{lemma}

\begin{proof}
	Assume for contradiction  $t_1 \coind t_2$. This means that  there are
	cointial transitions $t_1', t_2'$, respectively, such that $t_1\eveqt t_1'$, 
	$t_2\eveqt t_2'$ and $t_1' \ind t_2'$. Since $\eveqt$ preserves independence we obtain $t_1 \ind t_2$. This contradicts complementarity on transitions (\autoref{prop:complementarity transitions}) given that $t_1\sdep t_2$ holds.
\end{proof}
}

\begin{lemma}\label{lem:max-events}
	Let $(X, Y, f)$ be DP-grounded for some $\Rel{r}[DP]$. Then  \(\forall e \in \Max{\ev{X}}\), \(f(e) \in \Max{\ev{Y}}\).
\end{lemma}

\begin{proof}
	We proceed by a proof by contradiction and suppose that \(\exists e \in \Max{\ev{X}}\) such that \(f(e) \notin \Max{\ev{Y}}\).
	Then, wlog, there exists \(f(e') \in \ev{Y}\) (since \(f\) is a label-preserving bijection between \(\ev{X}\) and \(\ev{Y}\)) such that \(f(e) \ip f(e')\).
	Then by \autoref{lem:immed pred composable}, $f(e)$ is composable with $f(e')$ and by \autoref{lem:im-fcau-iff-dep}
	we get  $f(e) \sdep f(e')$. 

	Now, observe that  $(X, Y, f) \in \Rel{r}[DP] \iff  (Y, X, f^{-1}) \in \Rel{r}[DP]$ by definition.
	Hence, just before \(f(e')\) was triggered, \(f(e)\) was maximal (this would otherwise contradict \(f(e) \ip f(e')\)), and the transition associated to \(f(e')\) was a forward transition (since the triple is grounded) such that
	\[f(e) \sdep f(e') \iff f^{-1}(f(e)) \sdep f^{-1}(f(e'))\]
	Since we have already established that $f(e) \sdep f(e')$, then it must be the case that \(e \sdep e'\).

	Since $e,e'$ are events of $X$ we have $e\not\Cf e'$. Also, $e\not\coind e'$ follows from \autoref{lem:ind sdep coind} since the
	events have dependent labels. 
	Moreover, $e\neq e'$ since $f$ is a bijection. The two remaining options are $e<e'$ and
	$e'<e$ by polychotomy. 
	The first contradicts $e$ being maximal. Consider $e'<e$. When $(e', f(e'))$ was added to $f$, then
	$(f^{-1}(f(e)), f(e))$ already must have been in $f$ since \(f(e) \ip f(e')\). This means there is a path where
	$  f^{-1}(f(e))=e$ precedes $e'$ contradicting $e'<e$. Hence  the result.
\end{proof}

\thmbisimkpdp*

\begin{proof}Assume any standard \pccsk processes $P$ and $Q$. We shall work with KP-grounded and DP-grounded
tuples, namely with $(X,Y,f)$, where $X,Y$ are the derivatives of $P, Q$, respectively, obtained by applying either \autoref{def:HPbisim} or \autoref{def:Dsim}, and $f$ is an appropriate bijection between
$X$ and $Y$. It implies that if \((X, Y, f) \in \Rel{r}[KP]\) then $f$ is label- and order-preserving bijection,
and $f$ is label-preserving in \((X, Y, f) \in \Rel{r}[DP]\).
 
When proving each of the implications, we shall consider only cases for transitions of $X$ since the cases for 
transitions of $Y$ are proved correspondingly.		

\begin{description}
		
\item[$\Rightarrow$]
			
Assume \( P \Rel{b}[KP] Q\). By \autoref{def:HPbisim} there exists  \(\Rel{r}[KP]\) between $P$ and $Q$ 
such that $(P,Q,\emptyset)\in  \Rel{r}[KP]$, and there are KP-grounded triples  \((X, Y, f)\in  \Rel{r}[KP]\) for all the appropriate $X,Y$ and label- and order-preserving \(f\). 
Hence, we similarly assume  that 	\((P, Q, \emptyset) \in \Rel{r}[DP]\) for some $\Rel{r}[DP]$.
We then need to prove that, given a KP-grounded triple \((X, Y, f)\),  for each pair \(([t], [t'])\) of events mapped by \(f\) as in \autoref{def:HPbisim}, \(\forall e\in \Max{\ev{X}},  e \sdep [t] \iff f(e) \sdep [t']\):

	\begin{description}
	\item[$\Rightarrow$]
			Suppose \(e\in \Max{\ev{X}}\) and \(e \sdep [t]\). Since $t$ is a transition of $X$ and $e$ is a maximal event of $X$,
			we have that $e$ and $[t]$ are connected. They are also composable since assuming the opposite would imply 
			that $e$ is not maximal in $\ev{X}$. Then $e<[t]$ by \autoref{lem:im-fcau-iff-dep}, 
			and $\key{e}<_{X'} \key{[t]}$ by \autoref{thm:ordering-event-key}.
			Since $f$ is order preserving, we get  $\key{f(e)}<_{Y'} \key{[t']}$ for $[t']=f([t])$.
			We have $f(e), [t'] \in \ev{Y'}$, so by applying   \autoref{thm:ordering-event-key}
			we get $f(e)<[t']$. By \autoref{lem:max-events}, $f(e)$ is a maximal in $Y$ and using 
			the same argument as in \autoref{lem:max-events}, we show that $f(e), [t']$ are composable. 
			Hence   $f(e) \sdep [t']$ by  \autoref{lem:im-fcau-iff-dep}.

	\item[$\Leftarrow$] The argument is very much like in the $\Rightarrow$ case using  \autoref{lem:max-events} 
	and the fact that $f$ is order-preserving bijection.
		 
	\end{description}

\item[$\Leftarrow$]
Assume \( P \Rel{b}[DP] Q\). By \autoref{def:Dsim}, there is \(\Rel{r}[DP]\)  between $P$ and $Q$ such that 
\((P, Q, \emptyset)\in  \Rel{r}[DP]\), and there are DP-grounded triples \((X, Y, f)\in  \Rel{r}[DP]\) for all the appropriate $X,Y$
and label-preserving $f$. So we similarly assume  that 
			 \((P, Q, \emptyset) \in \Rel{r}[DP]\) for some $\Rel{r}[KP]$.

		Next we prove that \(f\) of any DP-grounded triple for \(\Rel{r}[DP]\), namely constructed from $\emptyset$ 
		by adding pairs $([t],[t'])$ for the matching transitions of $P,Q$ and 
		their respective derivatives according to conditions of  \autoref{def:Dsim}, is order preserving.  Assume 
		\((X, Y, f) \in \Rel{r}[DP]\) for some label-preserving $f$.  We show each implication in \autoref{def:ord-pres} separately.
		
 		\begin{enumerate}
 		\item[$\Rightarrow$]
		Consider $e\in \Max{\ev{X}}$ which is composed with $[t]$, and $e\sdep [t]$ holding.
		Assume for contradiction that \(\key{e} <_{X} \key{[t]}\) and \(\key{f(e)} \not<_{Y} \key{[t']}\), 
		using the notations of \autoref{def:HPbisim}. We get $e<[t]$ by applying \autoref{thm:ordering-event-key}.
		By \autoref{def:Dsim} we get $[t']$ and by applying \autoref{lem:max-events} we also get  
		$f(e)\in \Max{\ev{Y}}$ and $f(e)\sdep [t']$. Since $f(e)$ and $[t']$ are composable
		(shown as in the proof of \autoref{lem:max-events}), we
		obtain $f(e)< [t']$ by \autoref{lem:im-fcau-iff-dep}, and \(\key{f(e)} <_{Y} \key{[t']}\) by
		\autoref{thm:ordering-event-key}: contradiction.
		Hence, \(\key{e} <_{X} \key{[t]}\) implies \(\key{f(e)} <_{Y} \key{[t']}\).

 		\item[$\Leftarrow$] Consider $f(e)\in \Max{\ev{X}}$ and $[t']$ with $f(e)\sdep [t']$ holding.
 		Assume for contradiction
		\(\key{f(e)} <_{Y} \key{[t']}\) and \(\key{e} \not<_{X} \key{[t]}\) using the notations of \autoref{def:HPbisim}. 
		The last implies 
		$e\not<[t]$ by \autoref{thm:ordering-event-key}. Since $e, [t]$  are composable we get
		$e\not\sdep [t]$ (\autoref{lem:im-fcau-iff-dep}).  However, since $f(e)\sdep [t']$ holds we obtain by  \autoref{def:Dsim}
		that $e\sdep [t]$ holds: contradiction.
		 
		\Comment{
		Since \(\key{f(e)} <_{Y} \key{[t']}\) we obtain $f(e)< [t']$ by \autoref{thm:ordering-event-key}.
		Since $f(e)$ and $[t']$ are composable
		(shown as in the proof of \autoref{lem:max-events}) we get $f(e)\sdep [t']$ by 
		\autoref{lem:im-fcau-iff-dep}. This gives us  $e\sdep [t]$  by \autoref{def:Dsim}: contradiction.
		Hence,
		\(\key{f(e)} <_{Y} \key{[t']}\) implies \(\key{e} <_{X} \key{[t]}\).
		}
		 
		\end{enumerate}
This means that $f$ of any DP-grounded triple for \(\Rel{r}[DP]\) is order preserving.	
\qedhere		
	\end{description}
\end{proof}

Proving the last result requires to define the \enquote{No repeated events} (NRE) axiom~\cite[Def.\ 4.18]{LPU24}, which holds for pre-reversible LTSes~\cite[Prop.\ 4.21]{LPU24}:%
\begin{align}
	\text{For any rooted path $r$ and any event $e$ we have $\cte(r, e)\leq 1$} \tag{NRE}
\end{align}

\propfrimplieskp*

\begin{proof}
	FR bisimulation matches forward transitions of \(P\), \(Q\) by equating their labels and keys (and does the same to reverse
	transitions), whereas KP bisimulation only looks at labels of forward transitions.
	Moreover, in addition to preserving labels of matched transitions (and events) of \(P\), \(Q\) FR bisimulation preserves causal order on events~\cite[Proposition 5.6]{PU07a} provided that 
	NRE holds. Since the LTSIs for \pccsk is pre-reversible, NRE holds. 
	Finally, since any bijection that preserves causal order
	is also order preserving (\autoref{thm:ordering-event-key}) we get the result.
	\qedhere
\end{proof}

\Comment{
	\begin{align*}
	f (\leq_X) =& \leq_Y \\
	\forall t:  X \pr{f}[\theta] X' \Rightarrow &
	\exists t': Y\pr{f}[\theta']  Y', 
	\labl{\theta} = \labl{\theta'} \text{ and } \\
	& (X', Y', f\cup\{[t] \mapsto [t']\}) \in \Rel{r}[KP] \\
	\forall t':  Y\pr{f}[\theta]  Y' \Rightarrow & 
	\exists t: X\pr{f}[\theta']  X', \labl{\theta} = \labl{\theta'} \text{ and } \\
	&  (X', Y', f\cup\{[t] \mapsto [t']\})\in \Rel{r}[KP]
\end{align*}
If there exists a KP-simulation \(\Rel{r}[KP]\) such that \((X, Y, f)\in \Rel{r}[KP]\), we say that $X$ and $Y$ are KP-bisimilar, written $X \Rel{b}[KP] Y$.

A relation $\Rel{r}[D_{\textrm{m}}] \subseteq \kprocset \times \kprocset \times (\events \to \events)$ is a \emph{\(\Maxop\)-dependence (D\(_{\textrm{m}}\)) bisimulation} if $(\orig{X}, \orig{Y}, \emptyset) \in  \Rel{r}[D_{\textrm{m}}]$, and if whenever $(X, Y, f) \in  \Rel{r}[D_{\textrm{m}}]$, then \(f\) is a bijection between \(\ev{X}\) and \(\ev{Y}\) and
\begin{align*}
	\forall t: X \pr{f}[\theta] X' \implies & \exists t': Y \pr{f}[\theta'] Y',  \labl{\theta} = \labl{\theta'} \text{ and }\\
	&  \forall e\in \Max{\ev{X}},  e \sdep [t] \iff f(e) \sdep [t']\\
	&  \text{ and } (X', Y', f \cup \{[t] \mapsto [t']\}]) \in  \Rel{r}[D_{\textrm{m}}]\\
	\forall t': Y \pr{f}[\theta'] Y' \implies & \exists t: X \pr{f}[\theta] X',  \labl{\theta} = \labl{\theta'} \text{ and }\\
	& \forall e\in \Max{\ev{Y}},  e \sdep [t'] \iff f^{-1}(e) \sdep [t] \\
	& \text{ and } (X', Y', f \cup \{[t] \mapsto [t']\}) \in  \Rel{r}[D_{\textrm{m}}]
\end{align*}
}

    \endgroup
\end{document}